%% file: arxiv.tex
\documentclass[11pt]{article}
\usepackage{preamble}
\usepackage{bm}
\let\mathbm\bm

\newcommand{\ba}{\mathbm{a}}

\newcommand{\RR}{\mathbb{R}}
\newcommand{\CC}{\mathbb{C}}
\newcommand{\TT}{\mathbb{T}}
\newcommand{\ZZ}{\mathbb{Z}}
\newcommand{\NN}{\mathbb{N}}
\newcommand{\EE}{\mathbb{E}}
\newcommand{\HH}{\mathbb{H}}
\renewcommand{\bm}{\mathbf{m}}
\newcommand{\bn}{\mathbf{n}}
\newcommand{\bd}{\mathbf{d}}
\newcommand{\bS}{\mathbf{S}}
\newcommand{\be}{\mathbf{e}}
\newcommand{\br}{\mathbf{r}}
\newcommand{\Dtr}{D_{\mathrm{tr}}}
\newcommand{\tr}{\mathrm{tr}}
\newcommand{\Cov}{\mathsf{Cov}}
\newcommand{\calH}{\mathcal{H}}
\DeclareMathOperator{\Sp}{Sp}
\DeclareMathOperator{\OO}{O}
\DeclareMathOperator{\UU}{U}
\DeclareMathOperator{\One}{\mathbf{1}}

\newtheorem{task}[theorem]{Task}
\newcommand{\ketbra}[1]{\ket{#1}\bra{#1}}
\newcommand{\sech}{\mathrm{sech}}
\newcommand{\arccosh}{\mathrm{arccosh}}
\newcommand{\arcsinh}{\mathrm{arcsinh}}

\usepackage{braket}

\usepackage[T1]{fontenc}

\usepackage{times}
\AtBeginDocument{\DeclareFontShape{T1}{ptm}{m}{scit}{<->ssub*ptm/m/sc}{}}

\usepackage{setspace}
\usepackage[framemethod=TikZ]{mdframed}
\usepackage{multicol}
\usepackage[utf8]{inputenc}

\usepackage[scr=rsfs]{mathalpha}

\usepackage{upgreek}
\usepackage{thm-restate}

\usepackage[linesnumbered,ruled,vlined]{algorithm2e}
\SetKwInput{KwInput}{Input}
\SetKwInput{KwOutput}{Output}
\SetKwInput{Promise}{Promise}
\SetKwInput{KwRequire}{Require}
\SetKwInput{KwEnsure}{Ensure}
\SetKwProg{Fn}{Function}{:}{} 

\title{The $\log \log$ jam in Gaussian state tomography}
\author{
Sitan Chen
\thanks{SEAS, Harvard University. Email: \href{mailto:sitan@seas.harvard.edu}{sitan@seas.harvard.edu}.}
\qquad\qquad
Weiyuan Gong
\thanks{SEAS, Harvard University. Email: \href{mailto:wgong@g.harvard.edu}{wgong@g.harvard.edu}.}
\qquad\qquad
Qi Ye
\thanks{IIIS, Tsinghua University. Email: \href{mailto:yeq22@mails.tsinghua.edu.cn}{yeq22@mails.tsinghua.edu.cn}.}
\qquad\qquad
Zhihan Zhang
\thanks{Caltech. Email: \href{mailto:zzhang10@caltech.edu}{zzhang10@caltech.edu}.}
}

\hypersetup{
  pdftitle={The log log jam in Gaussian state tomography},
  pdfauthor={Sitan Chen, Weiyuan Gong, Qi Ye, Zhihan Zhang}
}

\begin{document}

\maketitle

\begin{abstract}
    Unlike in finite dimensions, quantum information in continuous-variable systems has the peculiar feature that without imposing physical constraints, the sample complexity of state tomography can be unbounded. Remarkably, this is even the case for state-of-the-art protocols for learning \emph{Gaussian} states, which have finite-dimensional descriptions: the best known rates scale with $\log \log E$, where $E$ is the energy of the system. We prove this is not an artifact of existing analyses, but a fundamental limitation of the measurements used. We show:
    \begin{itemize}
        \item Any protocol that uses Gaussian measurements, even entangled or adaptively chosen ones, must incur a $\log \log E$ dependence. This answers an open question posed in~\cite{chen2026towards,meleAdvancesQuantumLearning2026,chen2026optimaltomographybosonicfermionic}.
        \item There is a smooth tradeoff between the number of rounds of adaptivity and the energy dependence, and we give a matching protocol achieving this interpolated rate.
        \item With highly entangled, non-Gaussian measurements, one can learn $n$-mode pure Gaussian states with $O(n^2 / \epsilon^2)$ samples, independent of $E$. This answers an open question posed by~\cite{chen2026towards}.
        \item A simple protocol based on the single-copy canonical phase POVM of Holevo and Helstrom learns single-mode pure Gaussian states with $O(1/\epsilon^2)$ samples, again independent of $E$.
    \end{itemize}
    Our results clarify the role of energy in bosonic state tomography and shed new light on the intriguing interplay between adaptivity, entanglement, and magic in quantum learning.
\end{abstract}

\newpage

\tableofcontents

\newpage

\section{Introduction}

How efficiently can we learn about an unknown quantum system given the ability to perform experiments on it? While this has been the subject of considerable attention within quantum information for qudit systems, the analogous theory for bosonic (continuous-variable) systems remains nascent. This gap is especially striking given that on the experimental side, the most precise quantum sensing protocols performed to date have been native to the continuous-variable setting~\cite{acerneseIncreasingAstrophysicalReach2019,groteFirstLongtermApplication2013,brewerAl27QuantumLogicClock2019,tseQuantumEnhancedAdvancedLIGO2019}.

A host of mathematical challenges unique to bosons arise when passing from finite-dimensional spaces to infinite-dimensional ones. First and foremost, whereas the number of samples needed to learn a general qudit system naturally scales with the dimension of the Hilbert space~\cite{haah2016sample,o2016efficient}, in the bosonic setting where the Hilbert space is infinite-dimensional, it is necessary and natural to assume the system is physically constrained in some way. One standard constraint is that the \emph{mean energy $E$} be bounded, thus limiting how much of the state can occupy parts of space with high photon count. 

Alternatively, one can also assume a \emph{parametric ansatz} for the state, namely that it belongs to a family of states described by a finite number of parameters. Within this landscape, a widely studied ansatz is \emph{Gaussianity}. The Gaussian formalism is ubiquitous in theory and experiment because Gaussian states form an exactly solvable model with rich phenomenology and are straightforward to generate and manipulate in contemporary quantum optical platforms. Importantly for the purposes of this discussion, a Gaussian state is infinite-dimensional yet entirely specified by its mean and covariance matrix. Naively, this might suggest that any Gaussian state can be learned to constant precision in just a finite number of samples, independent of its mean and covariance and, in particular, independent of its energy $E$. This is certainly true classically: given unknown samples from a finite-dimensional Gaussian distribution $\mathcal{N}(\mu, \Sigma)$, the distribution can be learned to within total variation distance $\epsilon$ in $O(1/\epsilon^2)$ samples regardless of how large $\mu$ and $\Sigma$ are. 

Yet intriguingly, the best known protocols for Gaussian state tomography to date still have a sample complexity that depends on energy, scaling as $O(n\log \log E + \mathrm{poly}(n)/\epsilon^2)$~\cite{bittel2025energy,chen2026towards}.
We thus ask:

\begin{center}
    {\em Is energy dependence necessary for Gaussian state tomography?}
\end{center}

\noindent We entirely resolve this question, showing that the answer hinges crucially on the kinds of measurements used.

\paragraph{General Gaussian measurements.} Gaussian measurements are a canonical class of measurements forming the basis of state-of-the-art protocols for Gaussian state tomography. These are appealing from an experimental perspective because they can be implemented using simple primitives like squeezing, displacement, passive linear optics, and homodyne readouts. Just as Gaussian states can be thought of as the bosonic analogue of stabilizer states, Gaussian measurements can thus be thought of as the bosonic analogue of Clifford measurements.

The aforementioned $O(n\log \log E + \mathrm{poly}(n)/\epsilon^2)$ upper bound of~\cite{bittel2025energy} is achieved with a sequence of adaptively chosen Gaussian measurements. Our first result is to show that this dependence on $E$ is optimal for any protocol using such measurements, even if the measurements are allowed to be entangled across copies:

\begin{theorem}\label{thm:adaptive_loglogE_lower_bound}
    Any protocol for learning pure Gaussian states to trace distance $\epsilon$ using (possibly adaptive, possibly entangled) Gaussian measurements requires $\Omega(\log \log E + 1/ \epsilon^2)$ copies.
\end{theorem}

\noindent This answers an open question posed by~\cite{chen2026towards,meleAdvancesQuantumLearning2026,chen2026optimaltomographybosonicfermionic}. As a result, with only Gaussian measurements, it is impossible to learn Gaussian states with sample complexity independent of energy. This applies even to single-mode Gaussian states with zero displacement. To our knowledge, this is the first energy-dependent lower bound for Gaussian state learning that applies to adaptive measurements.

Our proof is based on showing that the direction along which the unknown state is significantly squeezed cannot be substantially narrowed down in a small number of measurements. For this, we provide a new way to recursively control the overall Fisher information of an adaptive learning protocol in terms of the Fisher information of individual experiments in the protocol. This is a significant departure from prior methods for proving adaptive sample complexity lower bounds in quantum state learning based on ``learning trees'' \cite{chenExponentialSeparationsLearning2022}, and we believe this particular way of using Fisher information may be of independent interest.

\paragraph{Bounded adaptivity.} Next, we decouple the role of adaptivity and the role of Gaussianity in designing measurement protocols. In the upper bound of~\cite{bittel2025energy} with $O(\log \log E)$ dependence, the protocol crucially relies on adaptivity in order to progressively unsqueeze the state. In contrast, it was shown in~\cite{chen2026towards} that for single-copy Gaussian measurements, any protocol that simply performs multiple \emph{nonadaptive} such measurements must scale much worse, namely linearly, in $E$.

Our next result significantly generalizes this by extending their impossibility result to general (possibly multi-copy) Gaussian measurements and to protocols with \emph{intermediate adaptivity}. This establishes an interpolation between our $\Omega(\log \log E + 1/\epsilon^2)$ lower bound for fully adaptive protocols, and the $\Omega(E / \epsilon^2)$ lower bound for nonadaptive protocols: 

\begin{theorem}\label{thm:bounded_adaptivity_lower_bound}
    Any protocol that uses Gaussian measurements and $T$ rounds of adaptivity to learn pure Gaussian states to trace distance $\epsilon$ (where $T = 1$ corresponds to nonadaptive protocols) requires sample complexity \begin{equation}
        \Omega(T^{-1} \epsilon^{-2^T/(2^T-1)} E^{1/(2^T-1)})\,.
    \end{equation} 
\end{theorem}

\noindent Additionally, we complement this with a qualitatively matching upper bound.

\begin{theorem}\label{thm:bounded_adaptivity_upper_bound}
    For $T\geq2$, there exists a learning protocol that only uses single-copy Gaussian measurements and has $T$ rounds of adaptivity, which learns any single-mode Gaussian state $\rho(\bm,V)$ with energy at most $E$ to trace distance $\epsilon$ with probability at least $1-\delta$, with sample complexity \[O\Bigl((T\epsilon^{-2^T/(2^T-1)}E^{1/(2^T-1)}+\epsilon^{-2})\log(T/\delta)\cdot \log(E)\cdot \log(E/\epsilon\delta)\Bigr).\]
\end{theorem}
\noindent The algorithm is designed for the bounded-adaptivity regime: for constant $T$, the upper and lower bound match up to logarithmic factors, but in the almost fully adaptive regime, for instance, $T=\Theta(\log\log E)$, modifications are needed to match the lower bound.

\paragraph{Non-Gaussian measurements.}

Finally, we show that if one goes beyond Gaussian measurements, it is possible to achieve fully energy-independent sample complexity for pure Gaussian states.

\begin{theorem}[Informal, see Theorem~\ref{thm:main_energyfree}]\label{thm:main_nongaussian}
    There is an entangled non-Gaussian measurement protocol on $m = O((n^2 + \log(1/\delta))/\epsilon^2)$ copies which learns any given pure Gaussian state on $n$ modes to trace distance $\epsilon$ with probability at least $1 - \delta$.
\end{theorem}

\noindent This answers another open question raised by~\cite{chen2026towards}. Our protocol can be regarded as the bosonic analogue of Hayashi's optimal estimator for pure state tomography in the finite-dimensional setting~\cite{hayashi1998asymptotic}. The construction of our measurement is rooted in the \emph{Siegel disk} parametrization of the set of pure, centered Gaussian states, and a certain measure over this set which is invariant under the action of the symplectic group. While these ideas are well-known within the theory of several complex variables~\cite{siegel1939einfuhrung,siegel1943symplectic,hua1963harmonic}, to our knowledge ours is one of the first works to leverage them algorithmically for the purposes of quantum learning.

Like Hayashi's estimator, the amount of entanglement needed by our protocol in the theorem above is substantial, and it is natural to ask whether both a large amount of entanglement \emph{and} non-Gaussian measurements are necessary to circumvent our no-go results for Gaussian measurements. We show below that this is not the case at least for $n = 1$:

\begin{theorem}\label{thm:simple_nongaussian}
    There is a protocol that only performs nonadaptive two-copy, non-Gaussian measurements on $N = O(\log(1/\delta)/\epsilon^2)$ copies which learns any given pure Gaussian state on one mode to trace distance $\epsilon$ with probability at least $1 - \delta$. Furthermore, if the state is promised to be zero-displacement, there is even a protocol that achieves the same using the same number of nonadaptive \emph{single-copy} non-Gaussian measurements.
\end{theorem}

\noindent Our learning protocol uses a classic construction from the quantum optics literature, originally devised by Helstrom~\cite{helstrom1974estimation} and Holevo~\cite[Chapter 4]{holevo2011probabilistic} (see also the thesis of~\cite{Pellonpaa2002CovariantPhase} and the references therein), called the \emph{canonical phase POVM}. While this particular measurement is not new, to our knowledge we give the first finite-sample analysis.

We leave as an interesting open question to achieve an energy-independent rate for general $n$-mode Gaussian states using few-copy measurements.

\paragraph{Concurrent work.} Independently, the very recent work of~\cite{chen2026optimaltomographybosonicfermionic} also achieved the optimal rate of Theorem~\ref{thm:main_nongaussian} for pure state tomography using the bosonic analogue of Hayashi's optimal estimator, albeit with a very different set of mathematical tools not based on the Siegel disk parametrization. Remarkably, they also achieved the optimal rate for \emph{mixed} Gaussian states by constructing a suitable random purification channel, whereas our upper bound only applies to pure states. On the other hand, the upper and lower bounds for Gaussian measurements in Theorems~\ref{thm:adaptive_loglogE_lower_bound}-\ref{thm:bounded_adaptivity_upper_bound}, as well as our upper bound using few-copy non-Gaussian measurements in Theorem~\ref{thm:simple_nongaussian}, are unique to our work.

\subsection{Related work}

\noindent\textbf{Bosonic learning theory.} Our work lies in the broader context of learning continuous-variable quantum systems. The earliest experiments on bosonic state tomography date back to the 1990s~\cite{smitheyMeasurementWignerDistribution1993,darianoTomographicMeasurementDensity1995,wallentowitzReconstructionQuantumMechanical1995,babichevHomodyneTomographyCharacterization2004} and have since become a standard tool in quantum optics~\cite{lvovskyContinuousvariableOpticalQuantumstate2009}.
These early works, however, did not provide rigorous guarantees on the reconstruction error.
Only recently has a line of theoretical work begun to systematically analyze the sample complexity of bosonic state tomography.
These developments belong to the emerging field of quantum learning theory for continuous-variable systems; we refer the reader to~\cite{meleAdvancesQuantumLearning2026} for a recent survey.
For general $n$-mode states of energy at most $E$, Mele et al.~\cite{meleLearningQuantumStates2025} showed that the sample complexity of tomography to trace-distance error $\epsilon$ is $\widetilde{\Theta}((E/n)^{n}/\epsilon^{2n})$.
Although the energy constraint makes the sample complexity finite, it remains exponential in the number of modes and is therefore inefficient for general bosonic states.

\vspace{0.5em}

\noindent\textbf{Gaussian state tomography.} This exponential cost motivates restricting attention to structured families of states, the most prominent of which is the class of Gaussian states.
Gaussian states are specified entirely by their first moments and covariance matrices, and they are routinely prepared and manipulated in quantum-optical platforms. 
As Gaussian measurements of such states can be classically simulated, one can regard such states as a natural continuous-variable analogue of qubit stabilizer states and fermionic Gaussian states.

In the Gaussian setting, the landscape for learnability improves dramatically.
Already in~\cite{meleLearningQuantumStates2025}, it was shown that an $n$-mode Gaussian state can be learned to trace-distance error $\epsilon$ using a number of samples polynomial in both $n$ and the energy $E$, by estimating its first and second moments through heterodyne detection.
Although the algorithm is simple, its analysis relies on nontrivial bounds relating trace distance to the distance between moments, and the initial bounds in~\cite{meleLearningQuantumStates2025} were not tight.
Subsequent works~\cite{bittelOptimalEstimatesTrace2025,fanizzaEfficientHamiltonianStructure2026,holevoEstimatesTracenormDistance2024,bittel2025energy} refined these perturbation bounds, culminating in a heterodyne tomography sample complexity of order $O(nE^2/\epsilon^2)$, which was shown in~\cite{chen2026towards} to be optimal among protocols restricted to \emph{heterodyne} measurements.

The energy dependence can be almost eliminated once adaptively chosen Gaussian measurements are permitted.
Bittel et al.~\cite{bittel2025energy} gave an adaptive protocol using $O(n^{3}/\epsilon^{2}+(n+\log\log\log E)\log\log E)$ Gaussian measurements. 
Their protocol iteratively learns the squeezing directions of the unknown state and applies passive Gaussian unitaries to approximately unsqueeze the state.
Once the state is sufficiently well conditioned, the protocol estimates the moments using heterodyne detection.
Chen et al.~\cite{chen2026towards} showed that the $n^3/\epsilon^2$ term is necessary for any protocol using Gaussian measurements.
They also showed that, if the Gaussian state is promised to be pure, the sample complexity can be improved to $\widetilde{\Theta}(n^2/\epsilon^2)$, up to the remaining $\log\log E$ dependence. \cite{bittel2025energy} entirely removed the energy dependence under the additional assumption that one has access to both $\rho$ \emph{and} the transposed state $\rho^\top$. However, in the standard setting where one only has access to $\rho$, resolving this remaining energy-dependent gap was raised as an open question by~\cite{chen2026towards} and~\cite{meleAdvancesQuantumLearning2026} and was one of the main motivations of our work.

\vspace{0.5em}

\noindent\textbf{Non-Gaussian measurements.} So far, the discussion has focused on protocols using Gaussian measurements, which are particularly simple to implement in contemporary experimental platforms.
A natural question is whether non-Gaussian operations can do better.
Chen et al.~\cite{chen2026towards} proved a general lower bound of $\Omega(n^2/\epsilon^2)$ for arbitrary tomography protocols, including protocols with non-Gaussian measurements.
Thus, for pure Gaussian states, non-Gaussian measurements cannot improve on the $n^2/\epsilon^2$ rate already achievable with Gaussian measurements.
On the other hand, for passive Gaussian states, namely squeezing-free Gaussian states,~\cite{chen2026towards} gives a non-Gaussian protocol with sample complexity $\widetilde{O}(n^2/\epsilon^2)$, whereas all Gaussian protocols require $\Theta(n^3/\epsilon^2)$ samples.
This established the first non-Gaussian advantage in bosonic state tomography, specifically with respect to $n$.
Whether non-Gaussian measurements can provide an advantage in the dependence on energy was raised as an open question in~\cite{chen2026towards}.

\vspace{0.5em}

\noindent\textbf{Role of adaptivity and magic in finite-dimensional systems.} Our work also lies in a broader body of work exploring the roles of adaptivity and non-Gaussian (non-stabilizer or magic) resources in quantum learning.

For qubit and qudit systems, the role of adaptivity is comparatively well understood.
For tomography of general $d$-dimensional states, adaptivity improves the optimal dependence on target infidelity~\cite{chen2023does}, but it provably does not improve the dependence on trace distance~\cite{o2016efficient,haah2016sample} or dimension, even for restricted classes of measurements~\cite{chen2023does,chen2024optimal,acharya2025pauli}.
Strong separations between adaptive and nonadaptive protocols have been shown for other tasks such as shadow tomography~\cite{chen2024adaptivity}, state certification~\cite{gupta2026few}, and estimation of certain structured states~\cite{goldar2026exponential}.

For qubit/qudit systems, the closest analogue of bosonic Gaussianity is the stabilizer formalism, where the non-stabilizerness resource is known as magic. \cite{acharya2025pauli} showed that tomography with Pauli measurements is strictly less efficient than tomography with general single-copy measurements, and \cite{kwon2025nonstabilizerness} obtained a separation between Clifford and non-Clifford measurements for a specific state discrimination task. Analogous separations between fermionic Gaussian measurements and non-Gaussian measurements in the fermionic setting remain largely unexplored.

\section{Overview of techniques}
Here, we give a high-level overview of the techniques used to prove our main results.

\subsection{Lower bound via Fisher information} In this overview, we will sketch the argument for \emph{fully adaptive} Gaussian measurements, deferring the details for the bounded-adaptivity setting to Section~\ref{sec:energy-dependent-lbd}.

The basic premise behind our lower bound is that because of the uncertainty principle, protocols that only use Gaussian measurements are unable to efficiently learn the direction of squeezing to sufficient accuracy with fewer than $\log \log E$ copies. To quantify this, we restrict our attention to a one-parameter family of single-mode squeezed states. Specifically, let $\rho_\theta$ denote the single-mode pure state with zero displacement and covariance matrix 
\begin{equation}
    V_\theta=R(\theta)\diag(\lambda,\lambda^{-1})R(\theta)^\top\,, \label{eq:overview_Vtheta}
\end{equation}
where $\theta$ is a uniformly random angle in $[0,\pi)$, $R(\theta)$ is the rotation matrix, and $\lambda\in[1,\infty)$ is the unique solution of $E=\frac14(\lambda+\lambda^{-1})$ such that each $\rho_\theta$ then has energy $E$. The stronger the squeezing, the more sensitive the state is to the angle $\theta$, and achieving trace distance $\epsilon$ requires estimating $\theta$ to error $O(\epsilon/E)$ (Lemma~\ref{lem:success_set}).

To prove a lower bound on the number of copies needed to achieve this accuracy, we will use a Fisher information argument.
Suppose a measurement $M$ produces an outcome $y$ with probability density $f^M_\theta(y)$.
The Fisher information of $M$ at $\theta$ is
\begin{equation}
    J_M(\theta) = \EE_{f^M_{\theta}}\left[\left(\partial_\theta\log f^M_{\theta}(y)\right)^2\right].
\end{equation}
Intuitively, $J_M(\theta)$ measures how sensitively the outcome distribution changes with $\theta$.
If $J_M(\theta)$ is small, then nearby values of $\theta$ induce nearly indistinguishable outcome distributions, so the measurement reveals little information about $\theta$.

More generally, it is standard to extend this definition to the Fisher information of an entire learning protocol at $\theta$. The classical \emph{Cram\'er-Rao bound} implies that any protocol with $o(E^2)$ Fisher information and which provides an unbiased estimate for $\theta$ has mean squared error $\omega(1/E^2)$, which is too large to learn $\rho_\theta$ to small trace distance error. There are several drawbacks to this approach. Firstly, we would ideally like a lower bound which applies even to biased estimators. Moreover, Cram\'er-Rao would not be able to rule out the existence of protocols that estimate $\theta$ to sufficient accuracy with small constant failure probability, as these might have large mean squared error. Lastly, for any \emph{fixed} $\theta$, there certainly exist protocols that perform very few Gaussian measurements yet have large Fisher information, simply because every measurement is highly squeezed in the direction of $\theta$.

These issues are alleviated by passing from a fixed angle $\theta$ to the \emph{uniform distribution} over angles $\theta$. In principle, averaging over $\theta$ lets us appeal to the \emph{Bayesian Cram\'er-Rao bound / van Trees inequality}, which does not require unbiasedness. Instead, we use an alternative argument (Lemma~\ref{lem:posterior_fisher_success_probability}) that allows us to rule out protocols with small constant failure probability, not just ones with $O(1/E^2)$ mean squared error. Finally, by randomizing $\theta$, we ensure that even though a single protocol might achieve high Fisher information for a narrow range of angles, no protocol can achieve high Fisher information in expectation over a random angle.

\subsection{Reducing to a single-step bound} It remains then to bound the expected Fisher information of any Gaussian measurement protocol. First, consider the expected Fisher information of a single measurement. In Lemma~\ref{lem:fisher_of_gaussian}, we prove that
\begin{equation}
    \sup_{k\text{-copy Gaussian measurement} \ M}\mathbb{E}_\theta [J_M(\theta)] = k(4E - 2)\,. \label{eq:singlestep}
\end{equation}
We briefly outline the proof of Eq.~\eqref{eq:singlestep} at the end of this overview. We now discuss its implications. 

Because Fisher information is additive for learning protocols that perform nonadaptive measurements, it immediately implies that any nonadaptive Gaussian measurement protocol requires $\Omega(E)$ copies to achieve total Fisher information $\Omega(E^2)$, which is necessary for estimating the unknown state to sufficient accuracy with small constant failure probability. This already recovers the nonadaptive lower bound of \cite[Theorem 7]{chen2026towards} and generalizes it to multi-copy Gaussian measurements. 

Furthermore, we give a new way to recursively control the Fisher information even for \emph{adaptive} protocols by reducing to the single-step estimate in Eq.~\eqref{eq:singlestep}. By chain rule for Fisher information, if $\bar{J}_t$ denotes the expected Fisher information for the protocol after the $t$-th measurement, then
\begin{equation}
    \bar{J}_t = \bar{J}_{t-1} + \EE_\theta \EE[J_{M_t}(\theta)]\,, \label{eq:overview_chain}
\end{equation}
where the inner expectation is over the randomness of the ``history'' of outcomes from the first $t - 1$ measurements. Crucially, because the protocol can be adaptive so that $M_t$ depends on this history, and this history in turn depends on $\theta$, the overall dependence of $M_t$ on $\theta$ can be extremely complicated. Our key insight is that provided that $\bar{J}_{t-1}$ is still not too large, this expectation over $M_t$ can be controlled in terms of the quantity on the left-hand side of Eq.~\eqref{eq:singlestep}. Indeed, by H\"older's inequality,
\begin{equation}
    \EE_\theta \EE[J_{M_t}(\theta)] = \int\mathrm{d}H_{t-1}\,\EE_\theta[q_{H_{t-1}}(\theta) \cdot J_{M_t}(\theta)] \le \int \mathrm{d}H_{t-1} \max_\theta q_{H_{t-1}}(\theta) \cdot \sup_M\EE_\theta[J_M(\theta)]\,,
\end{equation}
where $\mathrm{d}H_{t-1}$ is the Lebesgue measure over the space of histories, and $q_{H_{t-1}}(\theta)$ is the probability of seeing history $H_{t-1}$ when the true parameter is $\theta$. Intuitively, $\int \mathrm{d}H_{t-1} \max\theta q_{H_{t-1}}(\theta)$ is close to $1$ provided that the history $H_{t-1}$ is not too sensitive to $\theta$, and more quantitatively we can prove (Lemma~\ref{lem:Linf_by_Fisher}) that
\begin{equation}
    \int \mathrm{d}H_{t-1} \max\theta q_{H_{t-1}}(\theta) \le 1 + \pi\sqrt{\bar{J}_t}\,. \label{eq:overview_Linf}
\end{equation}
Putting Eqs.~\eqref{eq:singlestep}, \eqref{eq:overview_chain}, and \eqref{eq:overview_Linf} together, we deduce the key recursion
\begin{equation}
    \bar{J}_t \le \bar{J}_{t-1} + O(kE)\cdot (1 + \pi\sqrt{\bar{J}_t})\,.
\end{equation}
For instance, if $k = 1$, then unrolling this recursion shows that the Fisher information starts out proportional to $E$, then increases to $E^{3/2}$ in the next round, then to $E^{7/4}$, $E^{15/8}$, $E^{31/16}$... In particular, after $\log \log E$ rounds, the Fisher information grows to $E^2$; only at this point, and not before, can the protocol estimate $\theta$ to sufficient accuracy.

While the above gives a self-contained argument for protocols that use adaptive single-copy Gaussian measurements, handling general $k$ is more subtle. In particular, the number of copies to measure in the next round with an entangled Gaussian measurement can even be chosen adaptively depending on the history. To handle this, we use a novel refinement of the standard Fisher information recipe where we instead explicitly track the Fisher information of the \emph{posterior distribution} on the unknown parameter, conditioned on the history. The expectation of this posterior Fisher information, averaged over the history, is equal to $\bar{J}_t$ by Bayes' rule (because the Fisher information of the uniform prior distribution over $\theta$ is zero). But if we want to \emph{condition} on a history, which is needed to quantify the information accumulated by the protocol when the number of copies it measures next depends on that history, it becomes necessary to track the Fisher information of the posterior instead of merely $\bar{J}_t$. We defer the details of this to Section~\ref{sec:energy-dependent-lbd}.

\subsection{Proof of single-step bound} We briefly remark upon how to bound the expected Fisher information for a single fixed measurement in Eq.~\eqref{eq:singlestep}. This is the technically most involved part of the lower bound argument. To provide intuition, we specialize to the single-mode case. Suppose the Gaussian measurement $M$ is defined by seed covariance matrix $W$ (see Section~\ref{sec:boson_basics}), in which case the distribution over measurement outcomes after measuring $\rho_\theta$ is given by a classical Gaussian distribution with mean zero and covariance $V_\theta + W$, where $V_\theta$ is defined in Eq.~\eqref{eq:overview_Vtheta}. For simplicity, here we will assume further that $W$ is the covariance matrix of a pure Gaussian state, as this is the case in which Eq.~\eqref{eq:singlestep} holds with equality.

Standard calculations for the Fisher information of a Gaussian distribution, together with integration by parts, yield
\begin{align}
    \EE_\theta J_M(\theta) &= \EE_\theta \Tr\Bigl( ((V_\theta + W)^{-1} \partial_\theta V_\theta)^2 \Bigr) \\
    &= \EE_\theta \Tr((V_\theta + W)^{-1} \partial^2_\theta V_\theta) \\
    &= -2\EE_\theta \Tr((V_\theta + W)^{-1}V_\theta) + (\lambda + \lambda^{-1})\EE_\theta \Tr((V_\theta + W)^{-1})\,. \label{eq:overview_EJ}
\end{align}
As the symplectic eigenvalues of $W$ are 1 by assumption, it is easy to verify that for all $\theta$,
\begin{equation}
    \Tr((V_\theta + W)^{-1}V_\theta) = \frac{\lambda}{\lambda+1} + \frac{\lambda^{-1}}{\lambda^{-1}+1} = 1\,.
\end{equation}

It then remains to control the last term $\EE_\theta \Tr((V_\theta + W)^{-1})$. For this, we make a brief aside to introduce a central ingredient of this paper, the \emph{Siegel disk}.

\subsection{The Siegel disk} To every $V_\theta$ of the form in Eq.~\eqref{eq:overview_Vtheta}, consider the following point $z$ in the complex unit disk:
\begin{equation}
    R(\theta) \begin{pmatrix}
        \lambda & 0 \\
        0 & \lambda^{-1}
    \end{pmatrix} R(\theta)^\top \qquad \Longleftrightarrow \qquad \frac{\lambda - 1}{\lambda + 1} e^{2i\theta} \triangleq z_\theta\,. \label{eq:siegel_singlemode}
\end{equation}
Equivalently, $z_\theta$ is the complex scalar for which the pure state $\rho_\theta$ is given, up to normalization, by
\begin{equation}
    (1 - |z_\theta|^2)^{1/4} \exp\Bigl(\frac{z_\theta}{2} \hat{a}^{\dagger 2}\Bigr)\ket{0}\,. \label{eq:bargmann1}
\end{equation}
This parametrization has a host of useful properties that we will exploit both here and later in the proof of Theorem~\ref{thm:main_nongaussian}. To our knowledge, while such objects have a long history in the theory of Gaussian pure states~\cite{folland2016harmonic}, this work is the first to leverage the Siegel disk for the purposes of quantum learning.

For the purposes of controlling $\Tr((V_\theta + W)^{-1})$ in the proof of Eq.~\eqref{eq:singlestep}, we can use the fact that under the Siegel disk parametrization, $\Tr((V_\theta + V_{\theta'})^{-1})$ is equal to the \emph{Poisson kernel} evaluated at $z_\theta\overline{z}_{\theta'}$, and the classical Poisson integral formula implies that its expectation over $\theta$ is $1$. Putting everything together, this implies the claimed bound of
\begin{equation}
    \EE_\theta J_M(\theta) = -2 + \lambda + \lambda^{-1} = 4E - 2\,.
\end{equation}
In our main proof, which handles general $k$-copy Gaussian measurements, the argument is more indirect. While the argument is still centered around bounding terms that look like $\Tr((V_\theta + W)^{-1}V_\theta)$ and $\Tr((V_\theta + W)^{-1})$, instead of a direct computation using the Poisson integral formula, we instead derive two linear equations that are satisfied by these trace terms and solve the linear system. We defer the details to the proof of Lemma~\ref{lem:fisher_of_gaussian}.

\subsection{Energy-dependent upper bounds with bounded adaptivity}

We complement Theorem~\ref{thm:bounded_adaptivity_lower_bound} with a matching algorithm: a
single-copy Gaussian protocol using $T$ rounds of adaptivity whose sample complexity meets the
lower bound up to logarithmic factors.
We consider a general single-mode Gaussian state, with mean $\bm$ and covariance matrix $V=R(\theta)\diag(a,b)R(\theta)^\top$.
We require $\frac{a+b}{4}+\frac12\|\bm\|_2^2\leq E$ to ensure that the energy is at most $E$.
If we define the condition number $\kappa=\frac ab$, then reaching trace distance $\epsilon$ amounts to estimating $a,b$ to relative error $O(\epsilon)$, $\theta$ to additive error $O(\epsilon/\sqrt{\kappa})$ and $V^{-1/2}\bm$ to additive error $O(\epsilon)$.
Since $\kappa$ can be as large as $\Theta(E^2)$, the task is dominated by estimating $\theta$.

The algorithm spends its first $T-1$ rounds adaptively shrinking a confidence interval for $\theta$, and a final nonadaptive round that converts a good-enough interval into a full estimate of $(\bm,V)$.
Write $\Delta_t$ and $B_t$ as the width of the confidence interval of $\theta$ and $\sqrt{\kappa}\theta$ after round $t$, respectively.

The first $T-1$ rounds of the algorithm mimic the adaptive unsqueezing protocol of~\cite{bittel2025energy}.
At round $t\leq T-1$, we perform a Gaussian measurement whose seed covariance is squeezed along the current estimate of $\theta$.
The key departure from the unsqueezing protocol is that we match the seed squeezing with the inverse of $\Delta_t$, rather than to the potentially much larger squeezing of the unknown state.
This allows us to use the per-round samples more efficiently and effectively takes a square root of $B$ in each round:
\[B_t\approx\sqrt{\frac{B_{t-1}}{N_t}},\]
where $N_t$ is the number of samples used in the $t$-th round.
Hence if we take $N_t\approx\epsilon^{-2^T/(2^T-1)}E^{1/(2^T-1)}$ and $B_0\approx E$, we find $B_{T-1}\approx\epsilon(E/\epsilon)^{1/(2^T-1)}$.

Once $\sqrt{\kappa}\theta$ is confined to an interval of width $B_{T-1}$, in the final round we adapt the homodyne tomography algorithm of~\cite{chen2026towards}.
We randomly pick measurement directions inside our confidence interval of $\theta$.
To get accurate estimates of all the parameters, in particular the smaller variance $b$, we need a measurement direction that is within $O(1/\sqrt{\kappa})$ to the small variance direction.
Picking about $B_{T-1}$ random directions ensures that this happens with high probability.
For each direction, we need about $\epsilon^{-2}$ measurements to obtain an accurate estimate, hence the number of samples needed for the final round is about $B_{T-1}/\epsilon^2$, which turns out to also be about $\epsilon^{-2^T/(2^T-1)}E^{1/(2^T-1)}$.
Hence the sample complexity is roughly $\epsilon^{-2^T/(2^T-1)}E^{1/(2^T-1)}$.

\subsection{An optimal non-Gaussian POVM} 
We now give the intuition behind our protocol achieving optimal, energy-independent sample complexity for tomography of pure, $n$-mode Gaussian states. The first step in our protocol is a simple reduction to the special case where the state has zero displacement, by applying a $50:50$ beamsplitter to pairs of copies.

For zero displacement pure Gaussian states, we will again use the Siegel disk parametrization. On $n$ modes, the Siegel disk consists of the set
\begin{equation}
    \mathfrak{D}_n \triangleq \{K\in\CC^{n\times n}: K = K^\top, I - KK^\dagger \succ 0\}\,,
\end{equation}
and one can associate any such Gaussian state to some $K \in \mathfrak{D}_n$ via the parametrization
\begin{equation}
    \det(I - KK^\dagger)^{1/4} \exp\Bigl(\frac{1}{2} \hat{\ba}^\dagger K \hat{\ba}^\dagger\Bigr) \ket{0}\,.
\end{equation}
generalizing the form in Eq.~\eqref{eq:bargmann1}. Intuitively, the origin of this disk corresponds to the vacuum state, and states close to its boundary correspond to highly squeezed states.

A natural strategy to achieve energy independence is to design a measurement whose behavior does not differ meaningfully between states with different energy levels. This could be achieved, for instance, with a measurement that is \emph{covariant} with respect to arbitrary centered Gaussian unitaries, in the sense that if the unknown state being measured is transformed by such a unitary, then the outcome of the measurement is likewise transformed. This covariant property would then ensure that in order to establish correctness of the protocol for general, possibly highly squeezed pure Gaussian states, it suffices to establish correctness when the state is a \emph{fixed} pure Gaussian state of our choosing, e.g., the vacuum state.

To design such a measurement, we first design a \emph{measure} over the Siegel disk which is \emph{invariant} under the action of centered Gaussian unitaries $U$. \emph{A priori}, it is not obvious whether such a measure should even exist.
Indeed, the Siegel disk is stratified into regions corresponding to different levels of squeezing, and the action of a non-passive Gaussian unitary will by definition transport mass between these different strata in a complicated way.

To develop intuition for our construction, consider the single-mode setting, where $\mathfrak{D}_1$ is simply the set of complex numbers $z$ for which $|z| < 1$, and $\frac{1}{1 - |z|^2}$ directly corresponds to the amount of squeezing. If we simply considered the Lebesgue measure $\mathrm{d}z$ on this disk, it is a classical fact that any automorphism of this disk takes the form $z\mapsto z' = \frac{az + b}{\bar{b}z + \bar{a}}$ for $|a|^2 - |b|^2 = 1$, and each such map corresponds to a centered Gaussian unitary on a single mode. Under this map, the amount of squeezing and the Lebesgue measure transform as
\begin{equation}
    \frac{1}{1 - |z'|^2} = \frac{|\bar{b}z + \bar{a}|^2}{1 - |z|^2}\,, \qquad \mathrm{d}z' = \frac{\mathrm{d}z}{|\bar{a} + \bar{b}z|^4}\,.
\end{equation}
Attempting to cancel $|\bar{a} + \bar{b}z|^4$ naturally suggests a construction: the measure $\frac{\mathrm{d}z}{(1 - |z|^2)^2}$.
It turns out that the natural generalization of this to the $n$-mode case is the measure
\begin{equation}
    \frac{\mathrm{d}K}{\det(I - KK^\dagger)^{n+1}}\,,
\end{equation}
which is known to be invariant under the action of centered Gaussian unitaries (Lemma~\ref{lem:invariance}).

This suggests a natural POVM, namely the operator-valued measure $\ketbra{\psi_K}\frac{\mathrm{d}K}{\det(I - KK^\dagger)^{n+1}}$, reminiscent of the uniform POVM~\cite{guctua2020fast,wright2016learn} used to achieve optimal pure state tomography in the finite-dimensional setting. Unfortunately, this is invalid because neither this operator-valued measure nor the scalar measure $\frac{\mathrm{d}K}{(I - KK^\dagger)^{n+1}}$ is integrable. However, remarkably if one passes to the \emph{multi-copy} version of this POVM, namely $\ketbra{\psi_K}^{\otimes m}\mathrm{d}\mu(K)$, the integral becomes finite once $m$ is sufficiently large (Lemma~\ref{lem:bergman}). 

The proof of this is rather involved, but the central observation is that even though the scalar measure $\frac{\mathrm{d}K}{\det(I - KK^\dagger)^{n+1}}$ is not integrable, the measure $\det(I - KK^\dagger)^\gamma\mathrm{d}K$ is integrable, provided $\gamma > -1$ \--- this is the famous \emph{Selberg integral}~\cite{selberg1944bemerkninger} (see Theorem~\ref{thm:selberg}). The multi-copy POVM construction provides us with extra powers of $\det(I - KK^\dagger)$ to mimic the Selberg integral. Then by invoking a known correspondence between the ``Gaussian-symmetric subspace'' $\mathrm{span}(\{\ket{\psi_K}^{\otimes m}\}_{K\in\mathfrak{D}_n})$ and a certain \emph{weighted Bergman space}, we can show that provided $m > 2n$, the multi-copy operator-valued measure above is integrable and has integral proportional to the projector onto the Gaussian-symmetric subspace, giving rise to the POVM that we consider.

Because this POVM is, by design, covariant with respect to the action of centered Gaussian unitaries, we can without loss of generality assume that the unknown state being measured is the vacuum state, and it suffices to prove that in this case, the outcome $\ket{\psi_K}$ of the measurement is sufficiently close to vacuum. We prove this in Lemma~\ref{lem:vacuum_analysis} by leveraging an elementary quantitative estimate for the Selberg integral derived from Selberg's explicit formula (Lemma~\ref{lem:hua_estimate}).

\subsection{A simpler non-Gaussian measurement for \texorpdfstring{$n = 1$}{n = 1}}
Finally, we describe our single-mode protocol which also achieves energy-independent sample complexity but without using highly entangled measurements.

As before, our guiding design principle is to construct a measurement which is \emph{covariant} with respect to some group action, but in this case instead of the entire group of centered Gaussian unitaries, we only consider the group of \emph{rotations}, allowing us to eventually restrict our analysis to the special case of squeezed states with diagonal covariance.

To deal with the squeezing, we will instead construct a measurement whose outcome distribution concentrates more heavily around the parameters of the true state as the squeezing increases. Recalling the parametrization in Eq.~\eqref{eq:siegel_singlemode} and letting $\ket{\psi_z}$ denote the state associated to $z = \frac{\lambda - 1}{\lambda + 1}e^{2i\theta}\in\mathfrak{D}_1$, our starting point is the following formula for its Fock coefficients:
\begin{equation}
    \ket{\psi_z} \propto \sum^\infty_{k=0} \frac{\sqrt{\binom{2k}{k}}}{2^k} z^k \ket{2k} = \sum^\infty_{k=0} \frac{\sqrt{\binom{2k}{k}}}{2^k} \Bigl(\frac{\lambda - 1}{\lambda + 1}\Bigr)^k e^{2ik\theta} \ket{2k}\,,
\end{equation}
where $\propto$ denotes equality up to normalizing constant. The Fock amplitudes are visualized in Figure~\ref{fig:fock}. In particular, because of the $\Bigl(\frac{\lambda - 1}{\lambda+1}\Bigr)^k = (1 - \Theta(1/E))^k$ factor, the amplitudes decay at a roughly exponential rate over a horizon of length proportional to the energy $E$.

\begin{figure}[h!]
    \centering
    \includegraphics[width=0.5\linewidth]{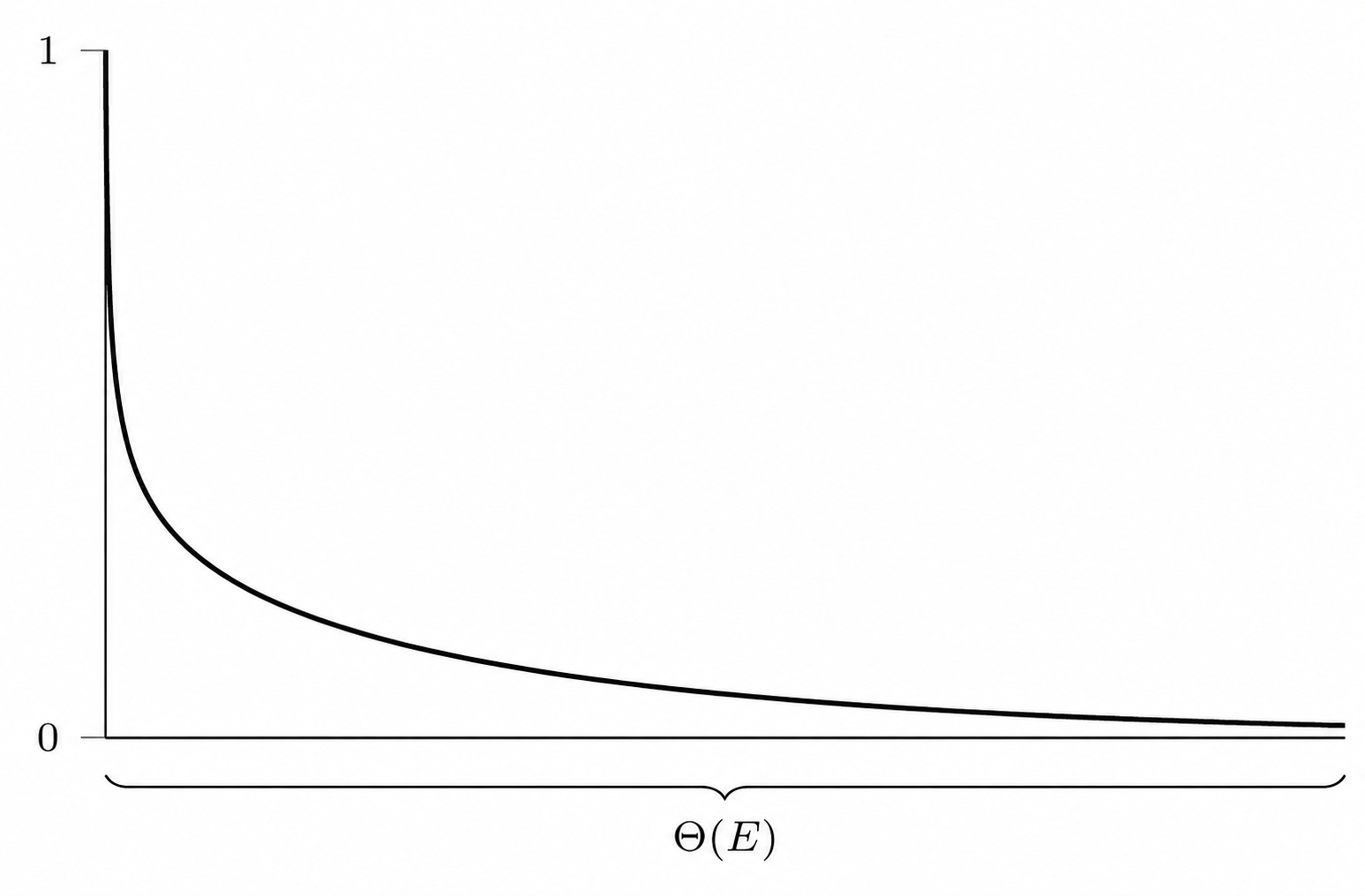}
    \caption{Exponential decay of Fock coefficients for $\ket{\psi_z}$ over a window of length $\Theta(E)$}
    \label{fig:fock}
\end{figure}

A natural strategy is thus to access the \emph{Fourier transform} of the amplitudes. We consider a POVM whose elements $\ket{\alpha}$ are parametrized by angles in $[0,2\pi)$:\footnote{We note that technically these are not valid states as they do not belong to $L_2(\mathbb{R})$, but we provide a rigorous definition of the POVM in Definition~\ref{def:canonical}.}
\begin{equation}
    \ket{\alpha} = \sum^\infty_{k=0} e^{ik\alpha}\ket{2k}\,.
\end{equation}
These states are closely related to the so-called \emph{Susskind-Glogower phase states}~\cite{susskind1964quantum}, and the corresponding POVM is closely related to the so-called \emph{canonical phase POVM} originally derived by~\cite{helstrom1974estimation} and \cite[Chapter 4]{holevo2011probabilistic}, the only difference being that, for convenience, our definition restricts to the even Fock sector because $\ket{\psi_z}$ does not have odd Fock coefficients.

As desired, this POVM is covariant with respect to rotations: if $\ket{\psi_z}$ is rotated by an angle $\Delta\in [0,\pi)$, then the resulting distribution over measurement outcomes $\alpha \in [0,2\pi)$ is correspondingly shifted by $2\Delta$. It thus suffices to show that when $z\in\mathbb{R}$, i.e., when $\theta = 0$, the measurement outcome concentrates around $\alpha = 0$. As we show in Lemmas~\ref{lem:conc} and~\ref{lem:faraway} (see also Figure~\ref{fig:phase}), in this case the probability mass for the distribution over outcomes concentrates in a window of radius $O(1/E)$ around $\alpha = 0$; the intuition for this is that this distribution is given by the Fourier transform of the Fock amplitudes of $\ket{\psi_z}$, which we recall from Figure~\ref{fig:fock} are effectively supported on an interval of length $\Theta(E)$. Given this concentration, we can simply measure $O(1/\epsilon^2)$ copies of $\ket{\psi_z}$ with this POVM and output the median outcome as an estimate for the angle of rotation (Lemma~\ref{lem:median}). A complication arises from periodicity because the samples live on the unit circle instead of an interval, but this can be handled by first coarsely estimating angle (Lemma~\ref{lem:coarse}) before re-centering around it (Lemma~\ref{lem:circular_median}).

Having estimated the \emph{direction} of squeezing, it remains to estimate the \emph{amount} of squeezing, which can be done with a standard photon number count and median-of-means estimate (Lemma~\ref{lem:photoncount}).

\begin{figure}[h!]
    \centering
    \includegraphics[trim=0mm 4mm 0mm 0mm, clip, width=0.5\linewidth]{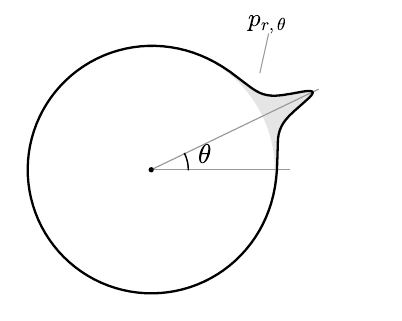}
    \caption{Distribution $p_{r,\theta}$ over measurement outcomes (gray) under the even canonical phase POVM concentrates around true angle $\theta$ when measuring $\ket{\psi_z}$ for $z = \tanh(r)e^{2i\theta}$}
    \label{fig:phase}
\end{figure}

\section{Outlook}

In this work, we showed that energy-dependent sample complexity is necessary for pure Gaussian state tomography if and only if the learner restricts to Gaussian measurements. 
For Gaussian measurements, we prove that the $O(\log\log E)$ gap in existing upper bounds is unavoidable even with adaptive and entangled measurements, and we further identify a fine-grained tradeoff between sample complexity and the number of adaptive rounds.
In contrast, once non-Gaussian measurements are allowed, pure Gaussian states can be learned with sample complexity independent of the energy. 
Altogether, our results disentangle the roles of three key resources in continuous-variable tomography: Gaussianity of the measurement, adaptivity, and entanglement across copies, and suggest a broader complexity hierarchy for learning structured quantum states under restricted measurement models. Below we discuss some open directions.

\vspace{0.5em}\noindent\textbf{Energy-dependent lower bounds for many modes.}
Our energy-dependent lower bounds are proved for single-mode Gaussian states.
For $n$ modes, the best known Gaussian protocol is due to \cite{bittel2025energy} with sample complexity $O(n^3/\epsilon^2 + (n+\log\log\log E)\log\log E)$.
Combining the $\Omega(n^3/\epsilon^2)$ lower bound of~\cite{chen2026towards} with our $\Omega(\log\log E)$ lower bound from Theorem~\ref{thm:adaptive_loglogE_lower_bound} still leaves a gap in the energy dependence for many modes.
In particular, it remains open whether the $n\log\log E$ term in the upper bound is necessary.

\vspace{0.5em}\noindent\textbf{Other applications of adaptive Fisher information control.}
A central ingredient in our lower bounds is a recursion that controls the growth of posterior Fisher information in an adaptive protocol in terms of the Fisher information of the individual measurements.
It would be interesting to find further applications of this method beyond Gaussian state tomography.
In finite-dimensional quantum learning, a powerful tool for adaptive lower bounds is the learning-tree framework~\cite{chenExponentialSeparationsLearning2022,chenComplexityNISQ2023,chenOptimalTradeoffsEstimating2024}.
Understanding whether these two approaches can be related or combined is an interesting direction for future work.

\vspace{0.5em}\noindent\textbf{Interpolating between Gaussian and non-Gaussian rates.} 
Our results leave open a finer complexity hierarchy interpolating between Gaussian and non-Gaussian measurement models. Given measurements that can be implemented by Gaussian circuits that have been doped by a small number $t$ of non-Gaussian gates, can we show a tradeoff between sample complexity and $t$ for Gaussian state tomography, and other bosonic learning tasks? Moreover, how are these tradeoffs affected by the amount of entanglement in the measurements?

\vspace{0.5em}\noindent\textbf{Analogues in qubit and fermionic systems.} 
The most natural candidate for an analogue of our setting for qubit / qudit systems would be to show separations between adaptive and nonadaptive protocols, and between Clifford and non-Clifford measurements, for learning stabilizer states. However, the optimal protocol for learning stabilizer states~\cite{montanaro2017learning,aaronson2008identifying} is nonadaptive and uses Clifford measurements. Where should one look instead for qubit / qudit analogues of the phenomena from this work? And similarly for fermionic systems?

\section*{Acknowledgments}

SC thanks Vishnu Iyer and Antonio Anna Mele for inspiring discussions and encouragement to work on continuous-variable systems, and Junseo Lee and Chirag Wadhwa for illuminating discussions about magic-sample tradeoffs in quantum learning.
The authors thank Francesco Anna Mele, Senrui Chen, Marco Fanizza, Filippo Girardi, Ludovico Lami, Michael Walter, and Freek Witteveen for sharing their work and coordinating the posting of our manuscripts.
ZZ acknowledges support from Kortschak Scholarship and the Broadcom Innovation Fund. Part of this work was done while SC and WG were visiting the Simons Institute for the Theory of Computing.

GPT-5.5 was used as a tool during the development of several of the results in this paper to help identify promising proof strategies and to pull in relevant mathematical techniques from adjacent areas, for instance the Siegel disk parametrization and the existence of an invariant measure over this set. All arguments generated by the model were manually checked and further developed by the authors, and the final manuscript was entirely human-written. Claude Opus 4.8 and GPT-5.5 were used to generate the figures. Any mathematical errors are the responsibility of the authors.

\section{Preliminaries}
\label{sec:preliminaries}
\subsection{Notation}

\paragraph{Matrices} For a positive integer $n$, $I_n$ denotes the $n\times n$ identity matrix. For a matrix $A\in\CC^{n\times n}$, we write $A^\top, A^*, A^\dagger$ for its transpose, complex conjugate, and conjugate transpose, respectively.
We will be careful to distinguish between Hermitian matrices ($A = A^\dagger$) and symmetric matrices ($A = A^\top$).
We use $\mathrm{Sym}_n(\mathbb{R})$ to denote the set of symmetric matrices in $\RR^{n\times n}$.
$A\in \RR^{2n\times 2n}$ is called (real) symplectic if $A\Omega_n A^\top = \Omega_n$, where $\Omega_n$ is the standard symplectic form defined as 
\[
    \Omega_n \coloneqq 
    \begin{pmatrix}
        0 & I_n \\
        -I_n & 0
    \end{pmatrix}.
\]
We use $\UU(n)$, $\OO(n)$, and $\Sp(2n)$ to denote the groups of $n\times n$ unitary, orthogonal, and $2n\times 2n$ symplectic matrices, respectively.

\paragraph{Decompositions}
Several matrix decompositions will be used in this paper. We refer to~\cite{houdeMatrixDecompositionsQuantum2024} for a summary of useful matrix decompositions in quantum optics.
\begin{itemize}
    \item \textit{Spectral decomposition:} A Hermitian matrix $A\in \CC^{n\times n}$ can be decomposed as $A=U D U^\dagger$, where $U\in \UU(n)$ is unitary and $D\in \RR^{n\times n}$ is diagonal. The diagonal entries of $D$ are the eigenvalues of $A$.
    \item \textit{Williamson decomposition:} A positive definite matrix $A\in \RR^{2n\times 2n}$ can be decomposed as $A=S(I_2\otimes D) S^\top$, where $S\in \Sp(2n)$ is symplectic and $D\in \RR^{n\times n}$ is diagonal. The diagonal entries of $D$ are the \emph{symplectic eigenvalues} of $A$ and are all positive.
    \item \textit{Singular value decomposition:} A matrix $A\in \CC^{n\times n}$ can be decomposed as $A=U D V^\dagger$, where $U,V\in \UU(n)$ are unitary and $D\in \RR^{n\times n}$ is diagonal. The diagonal entries of $D$ are the singular values of $A$ and are all nonnegative. In particular, if $A$ is Hermitian, then its singular values are the absolute values of its eigenvalues.
    \item \textit{Takagi decomposition:} A symmetric matrix $A\in \CC^{n\times n}$ can be decomposed as $A=UDU^T$, where $U\in \UU(n)$ is unitary and $D\in \RR^{n\times n}$ is diagonal and consists of the singular values of $A$. This is a special case of the singular value decomposition for symmetric matrices and should not be confused with the standard eigendecomposition.
    \item \textit{Euler decomposition}: A symplectic matrix $S\in \Sp(2n)$ can be decomposed as $O_1D O_2$, where $O_1, O_2$ are orthogonal symplectic matrices, and $D = \mathrm{diag}(e^{r_1},\ldots, e^{r_n}, e^{-r_1}, \ldots, e^{-r_n})$ for squeezing parameters $r_1,\ldots,r_n \ge 0$.
\end{itemize}

\paragraph{Classical Gaussian distributions}
For $\bm\in\RR^n$ and a symmetric positive-definite matrix $V\in\RR^{n\times n}$, we write $\calN(\bm,V)$ for the Gaussian distribution on $\RR^n$ with mean $\bm$ and covariance matrix $V$. 
At a point $x\in\RR^n$, its probability density with respect to Lebesgue measure is
\[
    \frac{1}{(2\pi)^{n/2}\sqrt{\det V}}
    \exp\left(
        -\frac12 (x-\bm)^\top V^{-1}(x-\bm)
    \right)\,.
\]

\subsection{Continuous-variable quantum systems and Gaussian formalism}
\label{sec:boson_basics}

Continuous-variable quantum systems model quantized modes of fields, harmonic oscillators, and other continuous-variable degrees of freedom. We refer the reader to \cite{serafiniQuantumContinuousVariables2023,weedbrookGaussianQuantumInformation2012,eisertIntroductionBasicsEntanglement2003} for comprehensive details and the preliminary sections of~\cite{meleLearningQuantumStates2025,bittelOptimalEstimatesTrace2025,bittel2025energy} for a concise introduction.
We will follow the normalization convention in~\cite{bittel2025energy} and the definition of covariance matrix will differ by a factor of 2 from~\cite{chen2026towards}.

A quantum system is called continuous-variable if its Hilbert space is infinite-dimensional.
Typically in quantum information we will consider continuous-variable systems with a finite number of modes, where each mode is a quantum harmonic oscillator. 
The Hilbert space of a single mode is $L^2(\mathbb{R})$, one basis for which is the \emph{Fock basis}. The Hilbert space of $n$ modes is the tensor product of $n$ single-mode Hilbert spaces:
\[
    \HH_1 = \overline{\Span}\{\ket{m}: m\in \NN\}, \qquad
    \HH_n = \HH_1^{\otimes n} = \overline{\Span}\{\ket{m}: m \in \NN^n\}.
\]
Each mode (say, the $i$-th mode) has a pair of bosonic field operators, namely the \emph{annihilation operator} $\hat{a}_i$ and the \emph{creation operator} $\hat{a}_i^\dagger$, defined by their action on the Fock basis,
\[
    \hat{a}_i\ket{m} = \sqrt{m_i}\ket{m-e_i}, \qquad
    \hat{a}_i^\dagger\ket{m} = \sqrt{m_i+1}\ket{m+e_i},
\]
where $e_i$ is the $i$-th standard basis vector in $\NN^n$. 
It is easy to verify that $\hat{a}_i^\dagger \hat{a}_i\ket{m} = m_i \ket{m}$, so $\hat{n}_i = \hat{a}_i^\dagger \hat{a}_i$ is called the \emph{number operator} of the $i$-th mode.
We denote $\hat{\mathbm{a}}=(\hat{a}_1,\ldots,\hat{a}_n)^\top$, $\hat{\mathbm{a}}^\dagger=(\hat{a}_1^\dagger,\ldots,\hat{a}_n^\dagger)^\top$.
These bosonic field operators satisfy the canonical commutation relations
\[
    [\hat{a}_i, \hat{a}_j^\dagger] = \delta_{ij} I.
\]
Alternatively, we can describe the system by another class of field operators called the \emph{quadrature field operators}, denoted by 
\[
    \hat{\bR} = (\hat{q}_1, \hat{q}_2, \cdots, \hat{q}_n, \hat{p}_1, \hat{p}_2, \cdots, \hat{p}_n)^\top
\]
We say $\hat{q}_i$ and $\hat{p}_i$ are the \emph{position} and \emph{momentum} quadrature operators of the $i$-th mode, respectively.
Denote $\hat{\mathbm{q}} = (\hat{q}_1, \hat{q}_2, \cdots, \hat{q}_n)^\top$ and $\hat{\mathbm{p}} = (\hat{p}_1, \hat{p}_2, \cdots, \hat{p}_n)^\top$.
The quadrature operators are related to the creation and annihilation operators by
\[
    \hat{q}_i = \frac{\hat{a}_i + \hat{a}_i^\dagger}{\sqrt{2}}, \qquad
    \hat{p}_i = \frac{\hat{a}_i - \hat{a}_i^\dagger}{\sqrt{2}i}.
\]
It is easy to verify that the quadrature operators satisfy the canonical commutation relations
\[
    [\hat{\bR}_i, \hat{\bR}_j] = i\Omega_{n,ij},
\]
where $\Omega_n$ is the standard symplectic form.

A general $n$-mode state can be represented by a density operator, that is, a trace-one and positive semidefinite operator $\rho: \HH_n\to \HH_n$.
When $\rho$ is a projector $(\rho^2=\rho)$, we say that $\rho$ is pure and can be represented as $\rho=\ketbra{\psi}$ where $\ket{\psi}\in \HH_n$. 
We can also represent the density operator in terms of a quasi-probability distribution called the \emph{Wigner function}.
We will not use the full expression for the Wigner function but only its first two moments, defined as
\[
    \bm(\rho) = \Tr(\rho \hat{\bR}), \qquad V(\rho)_{ij} = \Tr(\rho \{\hat{\bR}_i - \bm(\rho)_i, \hat{\bR}_j - \bm(\rho)_j\}).
\]
Here $\bm(\rho)\in \RR^{2n}$ is the first moment of $\rho$, also called the mean vector or displacement vector, and $V(\rho)\in \RR^{2n\times 2n}$ is the \emph{covariance matrix}.
The covariance matrix is real and symmetric. Moreover, it satisfies the following uncertainty relation
\[
    V(\rho) + i\Omega_n \succeq 0. 
\]
The uncertainty relation implies that $V(\rho)$ is positive definite.
In the Williamson decomposition, the uncertainty relation translates to the condition that the symplectic eigenvalues of $V(\rho)$ are at least $1$.
We denote the set of all valid $n$-mode covariance matrices by 
\begin{gather*}
    \Cov_n \coloneqq \{V\in\RR^{2n\times 2n}: V=V^\top,\ V+i\Omega_n\succeq 0\},\\
    \Cov_n^* \coloneqq \{V\in\Cov_n: \text{all symplectic eigenvalues are }1\}.
\end{gather*}
The (harmonic-oscillator) energy operator and the energy of $\rho$ are defined as
\begin{equation}
    \hat{E}_n \coloneqq \frac12 \sum_{j=1}^n (\hat{q}_j^2 + \hat{p}_j^2) = \frac12 \hat{\bR}^\top \hat{\bR}, \qquad E(\rho) \coloneqq \Tr(\rho \hat{E}_n) = \frac14 \Tr(V(\rho)) + \frac12 \|\bm(\rho)\|_2^2.
\end{equation}
Since $\hat{E}_n = \sum_{j=1}^n \hat{n}_j + \frac{n}{2}I$, the energy of $\rho$ is at least $\frac{n}{2}$, with equality if and only if $\rho=\ketbra{0^n}$.
The state $\ket{0^n}$ is called the vacuum state, with the first moment $\bm(\ketbra{0^n})=0$ and the covariance matrix $V(\ketbra{0^n})=I_{2n}$.

\paragraph{Gaussian states}
A particularly important class of continuous-variable states is formed by Gaussian states, whose Wigner functions are Gaussian distributions on phase space.
An $n$-mode Gaussian state $\rho$ is completely determined by its first moment $\bm$ and second moment $V$.
Conversely, for each $(\bm, V) \in \RR^{2n}\times \Cov_n$, there exists a unique $n$-mode Gaussian state, denoted by $\rho(\bm, V)$, with mean vector $\bm$ and covariance matrix $V$.
This establishes a one-to-one correspondence between $\RR^{2n}\times \Cov_n$ and the set of $n$-mode Gaussian states. In the single-mode case, we can write $\Cov_1$ and $\Cov_1^*$ more explicitly:
\[
    \Cov_1 = \{V\in\RR^{2\times 2}: V=V^\top,\ V\succeq 0,\ \det(V)\ge 1\}, \qquad \Cov_1^* = \{V\in\Cov_1: \det(V)=1\}\,.
\]

A Gaussian state $\rho(\bm, V)$ is pure if and only if the covariance matrix $V$ has all symplectic eigenvalues equal to $1$. 
In the Williamson decomposition, this is equivalent to $V = S S^\top$ for some symplectic matrix $S\in \Sp(2n)$. 

Tensor products of Gaussian states correspond to direct sums of their first and second moments. 
Specifically, if $\rho_A=\rho(\bm_A ,V_A)$ and $\rho_B=\rho(\bm_B, V_B)$, then $\rho_A\otimes\rho_B$ is Gaussian with mean vector $\bm_A\oplus\bm_B$ and covariance matrix $V_A\oplus V_B$. Therefore, the tensor product of $k$ copies of $\rho(\bm, V)$ is $\rho(\bm^{\oplus k}, V^{\oplus k})$.
A subtlety here is that the ordering of quadrature operators in $\bm_A\oplus\bm_B$ and $V_A\oplus V_B$ is $(\hat{\mathbm{q}}_A^\top, \hat{\mathbm{p}}_A^\top, \hat{\mathbm{q}}_B^\top, \hat{\mathbm{p}}_B^\top)^\top$. 
If we want to work in the ``$q$-then-$p$'' ordering of quadrature operators, we need to apply a permutation to the coordinate to get the $(\hat{\mathbm{q}}_A^\top, \hat{\mathbm{q}}_B^\top, \hat{\mathbm{p}}_A^\top, \hat{\mathbm{p}}_B^\top)^\top$ ordering. 
The tensored state $\rho(m, V)^{\otimes k}$ will become $\rho(\bm\otimes \One_k, V\otimes I_k)$,  where $\One_k$ is the $k$-dimensional all-one vector.
In this paper, we will use both orderings.

We will use the following standard formulas for the overlap and trace distance of pure Gaussian states.
\begin{lemma}\label{lem:distance_pure}
    Let $\rho_1=\rho(\bm_1, V_1)$ and $\rho_2=\rho(\bm_2, V_2)$ be two $n$-mode Gaussian states. If at least one of them is pure, then the fidelity between them is given by 
    \begin{equation}
        F(\rho_1, \rho_2) \coloneqq \Tr(\rho_1\rho_2) = \frac{2^n}{\sqrt{\det(V_1+V_2)}}\exp\left(-\frac12 (\bm_1-\bm_2)^\top (V_1+V_2)^{-1}(\bm_1-\bm_2)\right).
    \end{equation}
    If both $\rho_1$ and $\rho_2$ are pure, then the trace distance between them is given by
    \begin{equation}
        \Dtr(\rho_1, \rho_2) = \sqrt{1-F(\rho_1, \rho_2)}.
    \end{equation}
\end{lemma}

\paragraph{Gaussian operations}
In general, a quantum operation is a completely positive and trace-non-increasing map.
The simplest example is a unitary transformation $U$ that maps a state $\rho$ to $U\rho U^\dagger$.
We say that a unitary transformation is Gaussian if it maps Gaussian states to Gaussian states.
In terms of the quadrature operators, a Gaussian unitary can be simply represented as an affine transformation $\hat{\bR}\mapsto \bS\hat{\bR}+\bd$, where $\bd\in \RR^{2n}$ and $\bS\in \RR^{2n\times 2n}$. This map has to preserve the canonical commutation relations, which is equivalent to $\bS$ being symplectic.
Therefore, a Gaussian unitary is specified by a displacement vector $\bd$ and a symplectic matrix $\bS$, and we denote it by $U_{\bS, \bd}$.
When we apply a Gaussian unitary $U_{\bS, \bd}$ to a Gaussian state $\rho(\bm, V)$, it transforms the first and second moments covariantly:
\[
    \rho(\bm, V) \mapsto U_{\bS, \bd}\rho(\bm, V) U_{\bS, \bd}^\dagger = \rho(\bS\bm+\bd, \bS V \bS^\top).
\]
We often decompose a Gaussian unitary into a displacement unitary $D(\bd)\coloneqq U_{I_{2n}, \bd}$ and a displacement-free part $U_\bS\coloneqq U_{\bS, 0}$, so that $U_{\bS, \bd}=D(\bd)U_\bS$.

Here we list some important examples of Gaussian unitaries. They are all displacement-free so we only specify the symplectic matrix $\bS$.
\begin{itemize}
    \item \textit{Squeezing operation:} For $\br\in \RR^n$, the squeezing operation $S(\br)$ is the Gaussian unitary $U_\bS$ with $\bS = \diag(e^{r_1}, e^{r_2}, \cdots, e^{r_n}, e^{-r_1}, e^{-r_2}, \cdots, e^{-r_n})$. In particular, a single-mode squeezing operation $S(r)$ has $\bS = \diag(e^r, e^{-r})$.
    \item \textit{Rotation:} For $\theta\in \RR$, the single-mode phase-space rotation $R(\theta)$ is the Gaussian unitary $U_\bS$ with $$\bS = R(\theta)\coloneqq \begin{pmatrix} \cos\theta & -\sin\theta \\ \sin\theta & \cos\theta \end{pmatrix}.$$
    Here we abuse the notation $R(\theta)$ to denote both the rotation operation and the $2\times2$ rotation matrix.
    \item \textit{Beam splitter:} For two $n$-mode systems $\rho_A, \rho_B$, the beam splitter $B(\eta)$ with transmissivity $\eta$ is the Gaussian unitary $U_\bS$ with
    \[
        \bS = B(\eta) \coloneqq \begin{pmatrix}
            \sqrt{\eta} I_{2n} & \sqrt{1-\eta} I_{2n} \\
            -\sqrt{1-\eta} I_{2n} & \sqrt{\eta} I_{2n}
        \end{pmatrix}, \qquad 0\leq \eta\le 1
    \]
    in the $(\hat{\mathbm{q}}_A^\top, \hat{\mathbm{p}}_A^\top, \hat{\mathbm{q}}_B^\top, \hat{\mathbm{p}}_B^\top)^\top$ ordering.
    In particular, a 50:50 beam splitter has $\eta=1/2$. If we input two identical Gaussian states $\rho(\bm, V)$ into a 50:50 beam splitter, then, up to an irrelevant ordering of the output modes, the resulting state is $\rho(0, V)\otimes \rho(\sqrt{2}\bm, V)$. This allows one to reduce tomography for general Gaussian states to the zero displacement case (see Lemma~\ref{lem:reduce_to_meanzero} for a formal argument).
\end{itemize}

\paragraph{Gaussian measurements}
A general measurement is described by a positive operator-valued measure (POVM).
A particular class of measurements is the set of \emph{Gaussian measurements} (also called \emph{general-dyne measurements}) such that, when applied to a Gaussian state, the outcome distribution is Gaussian.
Following~\cite{chen2026towards} (but with the normalization convention differing by a factor of 2), such a measurement is specified by a seed covariance matrix $W\in \Cov_n$, and the outcome distribution is $\calN(\bm, \frac{V+W}{2})$ when applied to a Gaussian state $\rho(\bm, V)$. 
Two typical examples of Gaussian measurements are heterodyne measurements ($W=I_{2n}$) and homodyne ($W=K\diag(a, \cdots, a, a^{-1}, \cdots, a^{-1})K^\top$, where $a\to \infty$ and $K$ is an orthogonal symplectic matrix).
In practice, a general Gaussian measurement can be realized by applying a Gaussian unitary to the system and an ancilla in the vacuum state, followed by a heterodyne measurement.

More generally, a joint Gaussian measurement on $k$ copies of $\rho(\bm, V)$ is a Gaussian measurement on $\rho(\bm, V)^{\otimes k}$.
It is therefore specified by a seed covariance matrix $W^{(k)}\in\Cov_{nk}$, and its outcome is distributed as $\calN(\bm^{\oplus k}, \frac{V^{\oplus k}+W^{(k)}}{2})$.
This definition includes independent single-copy Gaussian measurements as a special case. Indeed, taking $W^{(k)}$ as a direct sum $W_1\oplus W_2\oplus \cdots \oplus W_k$, the measurement outcome is the concatenation of $k$ independent samples from $\calN(\bm, \frac{V+W_1}{2})$, $\calN(\bm, \frac{V+W_2}{2})\cdots$, $\calN(\bm, \frac{V+W_k}{2})$, respectively. 
Note that if we reorder the quadrature operators to the ``$q$-then-$p$'' ordering, the outcome distribution is $\calN(\bm\otimes \One_k, \frac{V\otimes I_k+W^{(k)}}{2})$.

We will also allow randomized Gaussian measurements defined as an ensemble of Gaussian measurements $\{(W_\alpha, p_\alpha)\}$.
It first samples a seed covariance matrix $W_\alpha$ from the distribution $(p_\alpha)$ and then performs the corresponding Gaussian measurement. 
Here we write it as a discrete ensemble for ease of notation, but the definition can be extended to continuous ensembles in a straightforward way.

\subsection{Gaussian state tomography}
Gaussian state tomography is the task of learning an unknown Gaussian state from measurement outcomes. 
We will measure accuracy in trace distance.

\begin{task}[Gaussian state tomography]\label{task:tomography}
    Given copy access to an unknown $n$-mode Gaussian state $\rho(\bm, V)$ with energy at most $E$, the goal of Gaussian state tomography is to output a classical description of a state $\hat{\rho}$ such that $\Dtr(\rho, \hat{\rho})\leq \epsilon$ with probability at least $2/3$.
\end{task}

\noindent Unless otherwise specified, we allow protocols to be adaptive and to perform joint measurements on multiple fresh copies in each round. 
The measurement chosen in a given round may depend on all previous measurements and outcomes. 
We now formalize this model.

\begin{definition}[Observations and histories]\label{def:history}
    An observation is a pair $h=(M,y)$, where $M$ is a measurement and $y$ is an outcome in the outcome space of $M$. 
    We write $k(M)$ for the number of copies consumed by the measurement $M$.

    A history is a finite sequence of observations
    \(
        H=(h_1,\ldots,h_t).
    \)
    We write $t(H)=t$ for the length of the history. 
    If $h_i=(M_i,y_i)$, we write
    \(
        s(H)
        \coloneqq
        \sum_{i=1}^t k(M_i)
    \)
    for the total number of copies consumed along the history. 
    The empty history is denoted by $H_{\emptyset}$, and satisfies
    \(
        t(H_{\emptyset})=s(H_{\emptyset})=0.
    \)

    In the special case of a $k$-copy Gaussian measurement, we identify $M$ with its seed covariance matrix $W^{(k)}\in\Cov_{nk}$. 
    The corresponding outcome is a vector $y\in\RR^{2nk}$, and $k(M)=k$.
\end{definition}

\begin{definition}[$T$-adaptive learning protocols]\label{def:adaptive_protocol}
    A deterministic $T$-adaptive learning protocol $\calP$ consists of two maps $(\calP_M,\calP_O)$. 
    The measurement map $\calP_M$ assigns to each history $H$ with $t(H)<T$ a measurement $\calP_M(H)$ on some number of fresh copies of the unknown state. 
    The output map $\calP_O$ assigns to each history $H$ with $t(H)=T$ a classical output, interpreted as a description of the estimated state.

    The protocol starts from the empty history $H_{\emptyset}$. 
    At a nonterminal history $H$, it applies the measurement $\calP_M(H)$ to fresh copies of the unknown state. 
    After observing an outcome $y$, the history is updated to $(H,(\calP_M(H),y))$.
    Once the history has length $T$, the protocol stops and outputs $\calP_O(H)$.
    
    The sample complexity of the protocol $\calP$ is defined as the maximum total number of samples consumed in any consistent length-$T$ history.
\end{definition}

\noindent The convention that a protocol performs exactly $T$ rounds entails no loss of generality. 
A protocol that would otherwise stop early can be padded with null measurements $M_{\emptyset}$ that consume no copies and deterministically return a fixed outcome. 
Thus $T=1$ corresponds to nonadaptive protocols, while allowing arbitrary $T$ recovers the usual fully adaptive model.
When $\calP_M(H)$ is a randomized Gaussian measurement for every nonterminal history $H$, we call $\calP$ a $T$-adaptive protocol with Gaussian measurements. 

We remark that the randomness in choosing the measurement and outputting the result can be absorbed into the measurement itself, so we can assume without loss of generality that the maps $\calP_M$ and $\calP_O$ are deterministic.

Finally, we remark that whereas protocols that perform a single entangled (possibly non-Gaussian) measurement on $N$ copies of a state are at least as powerful as ones that use adaptive (possibly non-Gaussian) few-copy measurement protocols on the same total number of copies, the same is not true if the measurements in both cases are restricted to be Gaussian. The reason is that adaptivity can encode computation that cannot be simulated with any entangled Gaussian unitary.

\subsection{Siegel representations of pure Gaussian states}
\label{sec:siegel} 
In this section we recall the Siegel disk parametrization of pure bosonic
Gaussian states. The parameter space is the Siegel disk
\[\mathfrak D_k=\{K\in\mathbb{C}^{k\times k}:K=K^\top,\ KK^\dagger\prec I_k\}.\]
It is the bounded realization of the Hermitian symmetric space associated with $\Sp(2k)/U(k)$.
Equivalently, it is biholomorphic to the Siegel upper half-plane
\[\mathfrak H_k=\{Z\in\mathbb{C}^{k\times k}:Z=Z^\top,\ \operatorname{Im}Z\succ0\},\]
via the Cayley transform.
Both realizations as well as their equivalence already appear in Cartan's classification of bounded symmetric domains~\cite{cartan1935domaines}.
The upper-half-plane realization is classical in Siegel's theory of modular functions~\cite{siegel1939einfuhrung}.

These parametrizations also appear in the mathematical-physics literature on Gaussian wave packets~\cite{folland2016harmonic,de2011symplectic}, which are precisely pure Gaussian states viewed in the position representation.
In the quantum information literature, explicit uses include the graphical calculus for Gaussian pure states~\cite{menicucci2011graphical}, study of entanglement~\cite{bianchi2015entanglement}, and formulation of Gaussian dynamics~\cite{pantaleoni2026gaussian}.

We work throughout with the Siegel disk.
$\mathfrak{D}_k$ is in one-to-one correspondence with $k$-mode pure Gaussian states with mean zero.
For $K\in\mathfrak{D}_k$, the corresponding state is~\cite[Eq. (8.2.6)]{perelomov1986}
\begin{equation}\label{eq:diskparam}
    \ket{\psi_K}=\det(I_k-KK^\dagger)^{1/4}\exp\Bigl(\frac{1}{2}\hat{\ba}^{\dagger\top}K\hat{\ba}^\dagger\Bigr)\ket{0}.
\end{equation}
It satisfies the following important property; while this fact is standard, we provide an elementary proof for completeness:

\begin{lemma}\label{lem:null_characterization}
    Given $K\in\mathfrak{D}_n$, the unique state $\ket{\psi}$ for which $(\hat{\mathbm{a}} - K\hat{\mathbm{a}}^\dagger)\ket{\psi} = 0$ is $\ket{\psi_K}$, up to complex phase.
\end{lemma}

\begin{proof}
    Take the Takagi decomposition $K = UDU^\top$, where $U$ is unitary and $D$ is diagonal with entries $0 \le t_1\le \cdots \le t_n < 1$. Define $\hat{\mathbf{b}} = U^\dagger \hat{\ba}$ and $\hat{\mathbf{b}}^\dagger = U^\top \hat{\ba}^\dagger$.
    One can verify that they are also canonical bosonic field operators.
    Observe that
    \begin{equation}
        \hat{\ba} - K\hat{\ba}^\dagger = U(\hat{\mathbf{b}} - D\hat{\mathbf{b}}^\dagger)\,.
    \end{equation}
    The condition $(\hat{\ba} - K\hat{\ba}^\dagger)\ket{\psi} = 0$ is thus equivalent to the condition that 
    \begin{equation}
        (\hat{b}_i - t_i \hat{b}_i^\dagger)\ket{\psi} = 0 \ \ \forall \ i\in[n]\,. \label{eq:modewise}    
    \end{equation}
    Under the basis $\ket{m}_b \triangleq \prod^n_{i=1} \frac{(\hat{b}_i^\dagger)^{m_i}}{\sqrt{m_i!}}\ket{0}$, express $\ket{\psi}$ as
    \begin{equation}
        \ket{\psi} = \sum_{m\in\mathbb{Z}^n_{\ge 0}}c_m \ket{m}_b\,.
    \end{equation}
    Eq.~\eqref{eq:modewise} implies that for every $m\in\mathbb{Z}^n_{\ge 0}$, 
    \begin{equation}
        \sqrt{m_i + 1}c_{m+e_i} = t_i \sqrt{m_i} c_{m-e_i}
    \end{equation} 
    (where $e_i\in\mathbb{Z}^n_{\ge 0}$ denotes the $i$-th standard basis vector and we let $c_{m-e_i} = 0$ if $m_i = 0$). From this we deduce that $c_m = 0$ for all $m$ with at least one odd entry. Furthermore, it also implies that for all $m$ with all even entries, if $m_i = 2n_i$, then
    \begin{equation}
        c_{2n} = c_0 \prod^n_{i=1} t^{n_i}_i \frac{\sqrt{(2n_i)!}}{2^{n_i}n_i!}\,.
    \end{equation}
    So $\ket{\psi}$ is uniquely determined by $c_0$. When $c_0 = 1$, the resulting state is proportional to
    \[
        \exp(\tfrac{1}{2}\sum^n_{i=1} t_i b^{\dagger 2}_i)\ket{0} = \exp(\tfrac{1}{2}\hat{\ba}^{\dagger\top} K \hat{\ba}^\dagger)\ket{0}.
    \]
    And because $\ket{\psi}$ must be normalized, it must be related to $\ket{\psi_K}$ by a complex phase.
\end{proof}

\begin{lemma}\label{lem:Siegel_cov}
For $K\in\mathfrak{D}_k$, the covariance matrix of $|\psi_K\rangle$ is
\[V_K=
\begin{pmatrix}
I_k+2\operatorname{Re}(Q_K+R_K)&2\operatorname{Im}(Q_K+R_K)\\
2\operatorname{Im}(Q_K+R_K)^\top&I_k+2\operatorname{Re}(Q_K-R_K)
\end{pmatrix},
\]
where
\[Q_K=K^\dagger K(I_k-K^\dagger K)^{-1},\qquad R_K=(I_k-KK^\dagger)^{-1}K.\]
\end{lemma}

\begin{proof}
We will omit the subscript $K$ when clear from context. Use $\langle \hat{O}\rangle$ to denote the $\bra{\psi_K}\hat{O}\ket{\psi_K}$, and define
\[P_{ij}\coloneqq\langle \hat{a}_i\hat{a}_j^\dagger\rangle,\qquad Q_{ij}\coloneqq\langle \hat{a}_i^\dagger \hat{a}_j\rangle,\qquad R_{ij}\coloneqq\langle \hat{a}_i\hat{a}_j\rangle.\]
Since $P_{ji}^*=\langle \hat{a}_j\hat{a}_i^\dagger\rangle^*=\langle \hat{a}_i\hat{a}_j^\dagger\rangle=P_{ij}$, $P=P^\dagger$.
Since $R_{ji}=\langle \hat{a}_j\hat{a}_i\rangle=\langle \hat{a}_i\hat{a}_j\rangle=R_{ij}$, $R=R^\top$.
Moreover, $\langle \hat{a}_i^\dagger \hat{a}_j^\dagger\rangle=\langle \hat{a}_j\hat{a}_i\rangle^*=R_{ij}^*$.
We also have
\[R_{ij}=\langle \hat{a}_i\hat{a}_j\rangle=\sum_kK_{jk}\langle \hat{a}_i\hat{a}_k^\dagger\rangle=\sum_kK_{jk}P_{ik},\]
hence $R=PK^\top=PK$.
On the other hand,
\begin{align*}
P_{ij}&=\langle \hat{a}_i\hat{a}_j^\dagger\rangle=\delta_{ij}+\langle \hat{a}_j^\dagger \hat{a}_i\rangle=\delta_{ij}+Q_{ji}\\
&=\delta_{ij}+\sum_kK_{ik}\langle \hat{a}_j^\dagger \hat{a}_k^\dagger\rangle=\delta_{ij}+\sum_kK_{ik}\langle \hat{a}_k\hat{a}_j\rangle^*=\delta_{ij}+\sum_kK_{ik}R_{kj}^*,
\end{align*}
hence $P=I_k+Q^\top$, and $P=I_k+KR^\dagger=I_k+KK^\dagger P^\dagger=I_k+KK^\dagger P$, so
\[P=(I_k-KK^\dagger)^{-1},\qquad Q=P^\top-I_k=K^\dagger K(I_k-K^\dagger K)^{-1},\qquad R=PK=(I_k-KK^\dagger)^{-1}K.\]
Now consider the covariance matrix.
The $xx$ block is
\[(V_{xx})_{ij}=2\langle\hat{x}_i\hat{x}_j\rangle=\langle(\hat{a}_i+\hat{a}_i^\dagger)(\hat{a}_j+\hat{a}_j^\dagger)\rangle=R_{ij}+P_{ij}+Q_{ij}+R_{ij}^*,\]
so
\[V_{xx}=P+Q+R+R^*=I_k+2\operatorname{Re}(Q+R).\]
The $pp$ block is
\[(V_{pp})_{ij}=2\langle\hat{p}_i\hat{p}_j\rangle=-\langle(\hat{a}_i-\hat{a}_i^\dagger)(\hat{a}_j-\hat{a}_j^\dagger)\rangle=-R_{ij}+P_{ij}+Q_{ij}-R_{ij}^*,\]
so
\[V_{pp}=P+Q-R-R^*=I_k+2\operatorname{Re}(Q-R).\]
Finally, the $xp$ block is
\[(V_{xp})_{ij}=\langle\hat{x}_i\hat{p}_j+\hat{p}_j\hat{x}_i\rangle=-i\langle(\hat{a}_i+\hat{a}_i^\dagger)(\hat{a}_j-\hat{a}_j^\dagger)\rangle-i\delta_{ij}=-i(R_{ij}-P_{ij}+Q_{ij}-R_{ij}^*+\delta_{ij}),\]
so
\[V_{xp}=-i(R-P+Q-R^*+I)=2\operatorname{Im}(Q+R).\]
\end{proof}

\noindent For $k=1$, the Siegel disk model becomes the Poincar\'e disk model of the hyperbolic plane.

\begin{lemma}\label{lem:diskfock}
For $K\in\mathfrak{D}_1$,
\[|\psi_K\rangle=(1-KK^\dagger)^{1/4}\sum_{k=0}^\infty\frac{\sqrt{\binom{2k}{k}}}{2^k}K^k|2k\rangle.\]
\end{lemma}

\begin{proof}
By Equation~\eqref{eq:diskparam},
\begin{align*}
|\psi_K\rangle&=(1-KK^\dagger)^{1/4}\exp\Bigl(\frac12K\hat{a}^{\dagger2}\Bigr)|0\rangle\\
&=(1-KK^\dagger)^{1/4}\sum_{k=0}^\infty\frac{1}{k!}\Bigl(\frac12K\hat{a}^{\dagger2}\Bigr)^k|0\rangle\\
&=(1-KK^\dagger)^{1/4}\sum_{k=0}^\infty\frac{K^k}{2^kk!}\sqrt{(2k)!}|2k\rangle\\
&=(1-KK^\dagger)^{1/4}\sum_{k=0}^\infty\frac{\sqrt{\binom{2k}{k}}}{2^k}K^k|2k\rangle\,.\qedhere
\end{align*}
\end{proof}

The following family of states will appear repeatedly:
\begin{lemma}\label{lem:kmode}
For $k$-mode pure Gaussian covariance matrix $V_\theta^{\oplus k}$, the corresponding point in $\mathfrak D_k$ is $\rho e^{2i\theta}I_k$, where $\rho=\frac{\lambda-1}{\lambda+1}$.
\end{lemma}

\begin{proof}
By Lemma~\ref{lem:Siegel_cov}, $Q_K=\frac{\rho^2}{1-\rho^2}I_k$, $R_K=\frac{\rho e^{2i\theta}}{1-\rho^2}I_k$,
\begin{align*}
V_K&=
\begin{pmatrix}
\Bigl(1+2\frac{\rho^2+\rho\cos2\theta}{1-\rho^2}\Bigr)I_k&2\frac{\rho\sin2\theta}{1-\rho^2}I_k\\
\frac{2\rho\sin2\theta}{1-\rho^2}I_k&\Bigl(1+2\frac{\rho^2-\rho\cos2\theta}{1-\rho^2}\Bigr)I_k\\
\end{pmatrix}\\
&=
\begin{pmatrix}
\frac{2\lambda^2+2+2(\lambda^2-1)\cos2\theta}{4\lambda}I_k&2\frac{(\lambda^2-1)\sin2\theta}{4\lambda}I_k\\
2\frac{(\lambda^2-1)\sin2\theta}{4\lambda}I_k&\frac{2\lambda^2+2-2(\lambda^2-1)\cos2\theta}{4\lambda}I_k
\end{pmatrix}\\
&=
\begin{pmatrix}
(\lambda\cos^2\theta+\lambda^{-1}\sin^2\theta)I_k&(\lambda-\lambda^{-1})\sin\theta\cos\theta I_k\\
(\lambda-\lambda^{-1})\sin\theta\cos\theta I_k&(\lambda\sin^2\theta+\lambda^{-1}\cos^2\theta)I_k
\end{pmatrix}\\
&=V_\theta^{\oplus k}\,.\qedhere
\end{align*}
\end{proof}

\noindent Finally, we also record the following convenient formulas under the Siegel representation.

\begin{lemma}[{\cite[Eq.~(8.2.12)]{perelomov1986}}]\label{lem:overlap}
    Given $K,K' \in \mathfrak{D}_k$,
    \begin{equation}
        \braket{\psi_K | \psi_{K'}} = \frac{\det(I_k-KK^\dagger)^{1/4}\det(I_k-K'K'^\dagger)^{1/4}}{\det(I_k-K^\dagger K')^{1/2}}.
    \end{equation}
\end{lemma}

\begin{lemma}\label{lem:siegel_det}
For all $K,K'\in\mathfrak D_k$,
\[\det\Bigl(\frac{V_K+V_{K'}}2\Bigr)=\frac{|\det(I_k-K^\dagger K')|^2}{\det(I_k-KK^\dagger)\det(I_k-K'K'^\dagger)}\,.\]
\end{lemma}

\begin{proof}
By Lemma~\ref{lem:distance_pure},
\begin{align}|\langle\psi_K|\psi_{K'}\rangle|^2&=\frac{2^n}{\sqrt{\det(V_K+V_{K'})}}=\det\Bigl(\frac{V_K+V_{K'}}2\Bigr)^{-1/2}\,.\qedhere\end{align}
\end{proof}

\section{Recursive control of Fisher information}

In this section, we define a single-parameter family of states that will form the basis of our lower bound argument (Section~\ref{sec:squeezed}) and specify the probability space induced by a fixed learning protocol using adaptively chosen Gaussian measurements (Section~\ref{sec:probability_space}). We then show that for any such protocol, its success probability can be bounded in terms of the Fisher information of the protocol (Section~\ref{sec:success}). Finally, we show that this Fisher information can be controlled recursively (Section~\ref{sec:recursive}). This effectively reduces bounding the Fisher information of the protocol to bounding the Fisher information of a single measurement, which we do explicitly in Section~\ref{sec:energy-dependent-lbd}.

\subsection{Basic setup}
\label{sec:squeezed}

Given an energy parameter $E > 1/2$, let $\lambda>1$ be the unique solution to $E=\frac14(\lambda+\lambda^{-1})$. For $\theta\in \TT \coloneqq \RR/\pi \ZZ \cong [0, \pi)$, we define $\rho_{\lambda,\theta}=\rho(0, V_{\lambda,\theta})$ to be the single-mode pure Gaussian state with zero displacement and covariance matrix
\begin{equation}\label{eq:Vtheta}
    V_{\lambda,\theta}=R(\theta) \begin{pmatrix}\lambda&0\\
    0&\lambda^{-1}\end{pmatrix} R(\theta)^\top, \quad \text{where} \ \ R(\theta)=\begin{pmatrix}\cos\theta&-\sin\theta\\\sin\theta&\cos\theta\end{pmatrix}.
\end{equation}
Whenever clear from context, we will omit $\lambda$ in the subscripts. Each state in this family has energy $E$. Note that any zero displacement pure Gaussian state can be written in this form for some $\lambda > 1$ and $\theta \in \TT$.

We will show that for fixed $\lambda$, learning $\rho_{\lambda,\theta}$ for a uniformly random $\theta$ is already hard. A lower bound for this restricted task immediately implies a lower bound for Gaussian state tomography.

\begin{task}\label{task:single-parameter}
    Let $\theta$ be an unknown parameter drawn uniformly from $\TT$. Given copy access to an unknown single-mode squeezed zero-mean Gaussian state $\rho_\theta$ with energy $E$, the goal is to output $\hat{\theta}$ such that $\Dtr(\rho_{\hat{\theta}}, \rho_{\theta})\le \epsilon$ with probability at least $2/3$.
\end{task}
\begin{lemma}\label{lem:reduction}
    If there is a learning protocol for Task~\ref{task:tomography} with error parameter $\epsilon$, then there is a learning protocol for Task~\ref{task:single-parameter} with error parameter $2\epsilon$. The two protocols only differ in the post-processing step (i.e., the output map $\calP_O$ in Definition~\ref{def:adaptive_protocol}), thus they have the same sample complexity and adaptivity.
\end{lemma}
\begin{proof}
    Suppose the learning protocol for Task~\ref{task:tomography} outputs $\hat{\rho}$ such that $\Dtr(\hat{\rho}, \rho_\theta)\le \epsilon$. We can classically post-process $\hat{\rho}$ to output $\hat{\theta}$ that minimizes $\Dtr(\hat{\rho}, \rho_{\hat{\theta}})$.
    Then the minimality implies $\Dtr(\hat{\rho}, \rho_{\hat{\theta}}) \le \Dtr(\hat{\rho}, \rho_\theta) \le \epsilon$, so, by the triangle inequality, $\Dtr(\rho_{\hat{\theta}}, \rho_\theta)\le 2\epsilon$.
\end{proof}

\noindent We now focus on proving the lower bound for Task~\ref{task:single-parameter}. For a subset $A\subseteq \TT$, we use $\abs{A}$ to denote its Lebesgue measure. The following lemma bounds the size of the success sets $A_\epsilon(\theta) = \{\theta'\in \TT \mid \Dtr(\rho_{\theta'}, \rho_\theta) \le \epsilon\}$. Roughly, to estimate the unknown squeezed state to trace distance $\epsilon$, the protocol must estimate $\hat{\theta}$ to additive error $O(\epsilon/E)$: the larger the energy $E$, the more sensitive the state is to changes in $\theta$.

\begin{lemma}\label{lem:success_set}
    For $\epsilon\in (0, 1/2)$, $E>1$, and $\theta\in \TT$, we have
    \begin{equation}
        \abs{A_\epsilon(\theta)} \leq \frac{2\pi\epsilon}{E}.
    \end{equation}
\end{lemma}
\begin{proof}
    Since $\rho_\theta$ and $\rho_{\theta'}$ are pure states, by Lemma~\ref{lem:distance_pure}, we have
    \begin{equation}
        \Dtr(\rho_\theta, \rho_{\theta'}) = \sqrt{1 - F(\rho_\theta, \rho_{\theta'})}, \quad \text{where} \ \ F(\rho_\theta, \rho_{\theta'}) = \frac{2}{\sqrt{\det(V_\theta + V_{\theta'})}}.
    \end{equation}
    By rotational invariance, with $u=\theta'-\theta$, we have
    \begin{equation}
        \det(V_\theta+V_{\theta'}) = \det(V_0 + V_u) = \det\begin{pmatrix}\lambda c^2+\lambda^{-1}s^2+\lambda & (\lambda-\lambda^{-1})cs\\ (\lambda-\lambda^{-1})cs & \lambda^{-1}c^2+\lambda s^2+\lambda^{-1} \end{pmatrix} = 4 + (\lambda-\lambda^{-1})^2 \sin^2u,
    \end{equation}
    where we abbreviate $c=\cos u$ and $s=\sin u$. Therefore, $\Dtr(\rho_\theta, \rho_{\theta'}) \leq \epsilon$ implies
    \begin{equation}
        \sin^2 (\theta'-\theta) \leq \frac{1}{(\lambda-\lambda^{-1})^2}\left(\frac{4}{(1-\epsilon^2)^2} - 4\right)\leq \frac{1}{(\lambda-\lambda^{-1})^2}\cdot \frac{8\epsilon^2}{(1-\epsilon^2)^2}\leq \frac{16\epsilon^2}{(\lambda-\lambda^{-1})^2} \leq \frac{4\epsilon^2}{E^2}.
    \end{equation}
    Here in the third inequality we use $\epsilon < 1/2$, and in the last inequality we use $\lambda-\lambda^{-1} = \sqrt{16E^2-4} > 2E$ for $E > 1$.
    The measure of such $\theta'$ is at most $2\arcsin(2\epsilon /E) \leq 2\pi \epsilon/E$ ($\arcsin x\leq \pi x/2$ for $x\in [0, 1]$). 
\end{proof}

\subsection{The induced probability space}
\label{sec:probability_space}

Fix a learning protocol $\calP$.
The unknown parameter $\theta$ is drawn uniformly from $\TT$, and the history is then generated by running $\calP$ on copies of the state $\rho_\theta$.
This induces a joint probability space over the unknown parameter and all measurement outcomes, which we specify here.
Throughout this section and the next section, all probabilities and expectations are taken with respect to this probability space unless stated otherwise.

We restrict attention to the histories that $\calP$ can actually generate, which we call \emph{consistent}. Formally, $H=(h_1,\dots,h_T)$ is consistent if and only if, for every $t\in[T]$, $M_t$ is the measurement chosen by the learner given the preceding history $H_{t-1}=(h_1,\dots,h_{t-1})$, and $y_t$ is the outcome of that measurement. We write $\calH$ for the set of all consistent histories, suppressing its dependence on $\calP$ and $T$ for brevity.

The sample space is $\Omega=\TT\times\calH$, and we define the joint density $q$ recursively. The parameter is uniform, $q(\theta)=\frac1\pi$ for $\theta\in\TT$; given $\theta$ and partial history $H_{t-1}$, the next observation $h_t=(M_t,y_t)$ has conditional density $q(h_t\mid\theta,H_{t-1})=f_\theta^{M_t}(y_t)$, where $f_\theta^{M}$ is the density of the outcome of measurement $M$ applied to $\rho_\theta$. The overall density is therefore
\begin{equation}
    q(\theta,h_1,\dots,h_T)=\frac{1}{\pi}q_\theta(H_T)\, \qquad \text{for} \qquad q_\theta(H_T) \triangleq \prod_{t=1}^T f_\theta^{M_t}(y_t)\,.
\end{equation}

\subsection{Bounding success probability using Fisher information}
\label{sec:success}

For a $k$-copy measurement $M$, the density of the outcome distribution $f^M_\theta$ is linear in $\rho^{\otimes k}_\theta$ and thus smooth in $\theta$. We may thus define the \emph{Fisher information of $M$ at $\theta$} by
\begin{equation}
    J_M(\theta) = \EE_{y\sim f^M_\theta}[(\partial_\theta \log f^M_\theta(y))^2]\,.
\end{equation}
Intuitively, $J_M(\theta)$ quantifies how much a measurement reveals about $\theta$: when it is small, the outcomes are insensitive to $\theta$, and $\theta$ is correspondingly hard to learn.

The definition of Fisher information naturally extends from single measurements to sequences of measurements. Any learning protocol $\mathcal{P}$ that performs a sequence of (possibly adaptive) measurements induces a distribution over histories $H_T = ((M_t, y_t))_{t=1}^T$, and we may define for any $1 \le t \le T$ the \emph{Fisher information of $\mathcal{P}$ up to time $t$ at $\theta$} by
\begin{equation}
    J_t(\theta) = \EE_{H_t}\Bigl[(\partial_\theta \log q_\theta(H_t))^2\Bigr]\,,
\end{equation}
where we have omitted $\mathcal{P}$ from the notation for brevity.

We will be especially interested in the \emph{expected Fisher information}, defined by averaging over the unknown parameter $\theta\in\TT$:
\begin{equation}
    \bar{J}_t = \EE_{\theta\sim\mathrm{Unif}(\TT)}[J_t(\theta)]\,,
\end{equation}
as well as the analogous quantity $\bar{J}_{M}$ for a single measurement $M$.

\noindent We will prove that the success probability of any learning protocol can be bounded in terms of its expected Fisher information:

\begin{lemma}\label{lem:posterior_fisher_success_probability}
    The success probability of $\mathcal{P}$ in solving Task~\ref{task:single-parameter} is at most
    \begin{equation}
        p_{\rm succ} \le \frac{2\epsilon}{E}\Bigl(1 + \pi\sqrt{\bar{J}_T}\Bigr)\,.
    \end{equation}
\end{lemma}

\begin{proof}
    By definition, the success probability is $p_{\rm succ} = \Pr[\theta, H_T]{\calP_O(H_T)\in A_\epsilon(\theta)}$. Given $H_T$, define 
    \begin{equation}
        B_\epsilon(H_T) = \{\theta'\in\TT \mid D_{\rm tr}(\rho_{\theta'}, \rho_{\calP_O(H_T)}) \le \epsilon\}\,.
    \end{equation}
    Then we may rewrite the success probability dually as
    \begin{equation}
        p_{\rm succ} = \Pr[\theta, H_T]{\theta \in B_\epsilon(H_T)}\,,
    \end{equation}
    noting that Lemma~\ref{lem:success_set} implies that $|B_\epsilon(H_T)| \le 2\pi\epsilon/E$ for all $H_T$. We may then bound the success probability by H\"older's inequality:
    \begin{equation}
        p_{\rm succ} \le \frac{2\epsilon}{E}\int \max_{\theta\in\TT} q_\theta(H_T)\,\mathrm{d}H_T\,.
    \end{equation}
    The claim then follows by Lemma~\ref{lem:Linf_by_Fisher}.
\end{proof}

\begin{lemma}\label{lem:Linf_by_Fisher}
    For any joint probability density $q(\theta, H_t)$ over the set $\TT\times\calH$, we have
    \begin{equation}
        \int \max_{\theta\in\TT} q_\theta(H_t)\,\mathrm{d}H_t \le 1 + \pi\sqrt{\bar{J}_t}\,.
    \end{equation}
\end{lemma}

\begin{proof}
    For brevity, we omit the subscript $t$. Let $\theta^*(H)$ be the maximizer of $q_\theta(H)$ over $\theta\in\TT$. Then for any fixed $\theta$, we have
    \begin{align}
         \int q_{\theta^*(H)}(H)\,\mathrm{d}H &= \int q_\theta(H)\,\mathrm{d}H + \int (q_{\theta^*(H)}(H) - q_\theta(H))\,\mathrm{d}H \\
        &\le 1 + \int\mathrm{d}H \int^{\theta^*(H)}_{\theta} \mathrm{d}\theta'\,|\partial_\theta q_{\theta'}(H)| \\
        &\le 1 + \int\mathrm{d}H\int^\pi_0 \mathrm{d}\theta'\, |\partial_\theta q_{\theta'}(H)| \\
        &= 1 + \int\mathrm{d}H\int^\pi_0 \mathrm{d}\theta' q_{\theta'}(H)^{1/2} \frac{|\partial_\theta q_{\theta'}(H)|}{q_{\theta'}(H)^{1/2}} \\
        &\le 1 + \int^\pi_0 \sqrt{J(\theta)}\,\mathrm{d}\theta\,,
    \end{align}
    where the last inequality follows from the Cauchy-Schwarz inequality. The claim follows by Jensen's inequality and the definition of $\bar{J}_t$.
\end{proof}

\subsection{Chain rule for Fisher information}
\label{sec:recursive}

In this section, we show that the Fisher information of a learning protocol can be bounded recursively in terms of the Fisher information of its constituent measurements. This is a consequence of the following standard fact:

\begin{lemma}[Chain rule for Fisher information]\label{lem:chain}
    For any family of joint distributions $q_\theta(X, Y)$, we have
    \begin{equation}
        J_{X,Y}(\theta) = J_X(\theta) + \EE_{X\sim q_\theta(X)}[J_{Y|X}(\theta)]\,,
    \end{equation}
    where $J_{Y|X}(\theta) = \EE_{Y\sim q_\theta(Y|X)}[(\partial_\theta \log q_\theta(Y|X))^2]$ is the conditional Fisher information of $Y$ given $X$.
\end{lemma}

\begin{corollary}\label{cor:pre_posterior_fisher_update}
    $\bar{J}_t = \bar{J}_{t-1} + \EE_{\theta, H_{t-1}}[J_{M_t}(\theta)]$.
\end{corollary}

\begin{proof}
    Applying Lemma~\ref{lem:chain} with $X = H_{t-1}$ and $Y = (M_t, y_t)$, we have 
    \begin{equation}
        J_t(\theta) = J_{t-1}(\theta) + \EE_{H_{t-1}\sim q_\theta(H_{t-1})}[J_{M_t}(\theta)]\,.
    \end{equation}
    Averaging over $\theta$ gives the claim.
\end{proof}

\noindent $\EE_{\theta, H_{t-1}}[J_{M_t}]$ is difficult to handle as $M_t$ depends on $H_{t-1}$, which depends in some complicated way on $\theta$. The following lemma shows that we can circumvent this by bounding this in terms of the Fisher information of a single \emph{fixed} measurement, in expectation over $\theta$. The argument is essentially the same application of H\"older's inequality from Lemma~\ref{lem:posterior_fisher_success_probability}.

\begin{lemma}\label{lem:recursive_bound}
    \begin{equation}
        \EE_{\theta, H_{t-1}}[J_{M_t}(\theta)] \le (1 + \pi\sqrt{\bar{J}_{t-1}})\cdot \sup_M \EE_\theta[J_M(\theta)]\,,
    \end{equation}
    where the supremum is over all Gaussian measurements on the number of copies used in round $t$.
\end{lemma}

\begin{proof}
    We have
    \begin{equation}
        \EE_{\theta, H_{t-1}}[J_{M_t}(\theta)] \le \int\mathrm{d}H_{t-1} \max_{\theta\in\TT} q_\theta(H_{t-1}) \cdot \EE_\theta[J_{M_t}(\theta)]\,.
    \end{equation}
    The claim follows by Lemma~\ref{lem:Linf_by_Fisher}.
\end{proof}

\noindent If we could upper bound the averaged single-step Fisher information $\EE_\theta[J_M(\theta)]$ for any Gaussian measurement $M$, then Lemma~\ref{lem:recursive_bound} bounds the per-step increase of the Fisher information, and unrolling this recurrence gives a bound on the Fisher information after all $T$ measurements. This strategy already suffices to give a lower bound for protocols where the number of copies measured in each round is nonadaptive. 

\subsection{Conditioning on histories and posterior Fisher information}

However, in order to handle general protocols where the number of copies depends adaptively on the history, we must be more careful. We will instead work with the Fisher information of the \emph{posterior distribution}. Given a history $H_t$, let $p_{H_t}$ denote the posterior distribution over the unknown parameter $\theta$ conditioned on observing the measurement outcomes given by $H_t$. That is,
\begin{equation}
    p_{H_t}(\theta) = \frac{q_\theta(H_t)}{\int_{\TT} q_{\theta'}(H_t)\,\mathrm{d}\theta'}\,.
\end{equation}
We will denote by $F(p_{H_t})$ the Fisher information of $p_{H_t}$, that is,
\begin{equation}
    F(p_{H_t}) \triangleq \EE_{\theta\sim p_{H_t}}[(\partial_\theta \log p_{H_t}(\theta))^2]\,.
\end{equation}
Roughly speaking, this captures the total information accumulated about $\theta$ after observing the history $H_t$; in particular, as the posterior sharpens around the true parameter over many measurements, $F(p_{H_t})$ increases.

We first observe that in expectation over $H_t$, this posterior Fisher information is equal to $\bar{J}_t$.
\begin{lemma}
    $\bar{J}_t = \mathbb{E}_{H_t} F(p_{H_t})$.
\end{lemma}

\begin{proof}
    Because the prior over $\theta\in\TT$ is uniform,
    \begin{equation}
        \partial_\theta \log p_{H_t}(\theta) = \partial_\theta \log q_\theta(H_t)\,,
    \end{equation}
    Moreover, the expectation $\mathbb{E}_{\theta}\mathbb{E}_{H_t\sim q_{\theta}(H_t)}$ in $\bar{J}_t$ has the same distribution as $\mathbb{E}_{H_T}\mathbb{E}_{\theta\sim p_{H_t}}$. 
    So squaring both sides and taking expectation over $\theta, H_t$ yields the claimed identity. 
\end{proof}

\noindent It thus suffices to track the Fisher information of the posterior. In direct analogy to Corollary~\ref{cor:pre_posterior_fisher_update} and Lemma~\ref{lem:recursive_bound}, this quantity also obeys a chain rule, as a direct consequence of standard chain rule for Fisher information (Lemma~\ref{lem:chain}), and we can also control its per-step increase in terms of the averaged single-step Fisher information $\EE_\theta[J_M(\theta)]$.

\begin{lemma}\label{lem:posterior_fisher_update}
    Conditioned on history $H_{t-1}$, let $M_t = \mathcal{P}_M(H_{t-1})$ be the next measurement. Then
    \begin{equation}
        \EE[F(p_{H_t}) \mid H_{t-1}] = F(p_{H_{t-1}}) + \EE_{\theta\sim p_{H_{t-1}}}[J_{M_t}(\theta)]\,. \label{eq:posterior_chain}
    \end{equation}
    Furthermore,
    \begin{equation}
        \EE_{\theta\sim p_{H_{t-1}}}[J_{M_t}(\theta)] \le \Bigl(1 + \pi\sqrt{F(p_{H_{t-1}})}\Bigr)\cdot \sup_M \EE_{\theta\sim\mathrm{Unif}(\TT)}[J_M(\theta)]\,.
    \end{equation}
\end{lemma}

\begin{proof}
    We apply Lemma~\ref{lem:chain} to the family of joint distributions $q_\upsilon(\vartheta, y) = p_{H_{t-1}}(\vartheta - \upsilon) f^{M_t}_{\vartheta - \upsilon}(y)$. Taking $X$ and $Y$ in Lemma~\ref{lem:chain} to be $\vartheta$ and $y$ respectively and $\theta = \upsilon = 0$, we get
    \begin{equation}
        J_{\vartheta, y}(0) = J_\vartheta(0) + \EE_{\vartheta\sim q_0(\vartheta)}[J_{y\mid \vartheta}(0)]\,,
    \end{equation}
    where $q_0(\vartheta) = p_{H_{t-1}}(\vartheta)$ and
    \begin{equation}
        J_\vartheta(0) = \int_{\TT} p_{H_{t-1}}(\vartheta) (\partial_\upsilon \log p_{H_{t-1}}(\vartheta - \upsilon) |_{\upsilon = 0})^2\,\mathrm{d}\vartheta = \int_{\TT} p_{H_{t-1}}(\vartheta) (\partial_\vartheta \log p_{H_{t-1}}(\vartheta))^2\,\mathrm{d}\vartheta = F(p_{H_{t-1}})
    \end{equation}
    \begin{equation}
        J_{y\mid \vartheta}(0) = \int f^{M_t}_{\vartheta}(y) (\partial_\upsilon \log f^{M_t}_{\vartheta - \upsilon}(y)|_{\upsilon = 0})^2\,\mathrm{d}y = \int f^{M_t}_\vartheta(y) (\partial_\vartheta \log f^{M_t}_\vartheta(y))^2\,\mathrm{d}y = J_{M_t}(\vartheta)\,.
    \end{equation}
    In other words, the right-hand side of Eq.~\eqref{eq:posterior_chain} is equal to $J_{\vartheta, y}(0)$.

    Next, swap the order by taking $X$ and $Y$ in Lemma~\ref{lem:chain} to be $y$ and $\vartheta$ respectively and $\theta = \upsilon = 0$, to get
    \begin{equation}
        J_{y, \vartheta}(0) = J_y(0) + \EE_{y\sim q_0(y)}[J_{\vartheta\mid y}(0)]\,.
    \end{equation}
    Note that for any $\upsilon$, $q_\upsilon(y) = \int_\TT p_{H_{t-1}}(\vartheta) f^{M_t}_{\vartheta}(y)$ is independent of $\upsilon$, so $J_y(0) = 0$. Indeed, $q_0(y)$ is simply the probability density of observing $H_t$ conditioned on $H_{t-1}$, and $J_{\vartheta\mid y}(0) = F(p_{H_t})$. In other words, the left-hand side of Eq.~\eqref{eq:posterior_chain} is equal to $J_{y,\vartheta}(0)$. As $J_{\vartheta, y}(0) = J_{y,\vartheta}(0)$, this concludes the proof of the first claim. The second claim follows immediately from H\"older's inequality and Lemma~\ref{lem:Linf_by_Fisher2} below.
\end{proof}

\noindent The following is used in the argument above and proceeds identically to the proof of Lemma~\ref{lem:Linf_by_Fisher}.

\begin{lemma}\label{lem:Linf_by_Fisher2}
    For every positive smooth probability density $p$ on $\TT$, $\norm{p}_{L^\infty(\TT)} \triangleq \max_{\theta\in\TT} p(\theta)$ is upper bounded by
    \begin{equation}
        \norm{p}_{L^\infty(\TT)} \le \frac{1}{\pi} + \sqrt{F(p)}\,,
    \end{equation}
    where $F(p) = \EE_{\theta\sim p}[(\partial_\theta \log p(\theta))^2]$.
\end{lemma}

\begin{proof}
    Since $\int_\TT p(\theta) \,\mathrm{d}\theta = 1$, there exists $\theta_0$ such that $p(\theta_0) \leq 1/\pi$. For every $\theta$, we have
    \begin{equation}
        p(\theta) = p(\theta_0) + \int_{\theta_0}^{\theta} p'(u) \,\mathrm{d}u \leq \frac{1}{\pi} + \int_{0}^\pi \abs{p'(u)} \,\mathrm{d}u.
    \end{equation}
    By Cauchy-Schwarz against $\sqrt{p(u)}$, we have
    \begin{align}
        \int_{0}^\pi \abs{p'(u)}\,\mathrm{d} u = \int_{0}^\pi \frac{\abs{p'(u)}}{\sqrt{p(u)}} \sqrt{p(u)} \,\mathrm{d}u \leq \sqrt{\int_{0}^\pi \frac{(p'(u))^2}{p(u)}\,\mathrm{d}u} \sqrt{\int_0^\pi p(u) \,\mathrm{d}u} = \sqrt{F(p)}. &\qedhere
    \end{align}
\end{proof}

\noindent Armed with Lemma~\ref{lem:posterior_fisher_update}, in the next section we proceed to bound the averaged single-step Fisher information and conclude the proof of our lower bound.

\section{Energy-dependent lower bounds for Gaussian measurements}
\label{sec:energy-dependent-lbd}

In this section, we use the framework developed in the previous section to establish energy-dependent lower bounds for learning Gaussian states with Gaussian measurements. We first compute the single-step Fisher information of joint Gaussian measurements in Section~\ref{sec:fisher_of_Gaussian}. Then we use this to derive the adaptivity-sample tradeoffs in Section~\ref{sec:adaptivity-sample tradeoffs}, which include the crossover from the nonadaptive regime $T=1$, to the bounded-adaptivity regime $T=O(1)$, and to the fully adaptive regime $T=\infty$.

Throughout this section, we focus on the single-parameter family of squeezed zero-mean states defined in Section~\ref{sec:squeezed}, and the measurements are restricted to be Gaussian measurements. 
As noted in Definition~\ref{def:history}, the $k$-copy joint measurement $W$ is identified as a valid covariance matrix in $\Cov_k$, and the measurement outcome $y$ is a $2k$-dimensional vector in $\RR^{2k}$.
We use the superscript $(k)$ on $W$ to emphasize the number of copies used in the measurement when necessary.
When the state is $\rho_\theta$, the outcome $y$ is distributed according to a Gaussian distribution with mean zero and covariance matrix $\frac{V_\theta^{\oplus k} + W^{(k)}}{2}$, therefore, the probability density function is
\begin{equation}
    f_{\theta}^{W} (y) = \frac{1}{\pi^k \sqrt{\det(V_\theta^{\oplus k} + W^{(k)})}} \exp\left(-y^\top (V_\theta^{\oplus k} + W^{(k)})^{-1} y\right).
\end{equation}
Recall that the single-step Fisher information of $W$ at $\theta$ is
\begin{equation}
    J_W(\theta) = \EE_{Y\sim f_\theta^W}[(\partial_\theta \log f_\theta^W(Y))^2].
\end{equation}

\subsection{Averaged Fisher information of joint Gaussian measurements}
\label{sec:fisher_of_Gaussian}
We calculate the averaged single-step Fisher information $\int_0^\pi J_W(\theta) \,\mathrm{d}\theta$ for any $k$-copy joint Gaussian measurement $W$ in the following lemma.

\begin{lemma}\label{lem:fisher_of_gaussian}
    For every $E\ge 1/2$, $k\ge 1$, and $W^{(k)}\in \Cov_k$, we have
    \begin{equation}
        \int_0^\pi J_{W^{(k)}}(\theta) \,\mathrm{d}\theta \leq k\pi (4E-2),
    \end{equation}
    with the equality achieved if and only if the seed covariance matrix $W^{(k)}$ corresponds to a pure Gaussian state, that is, $\det(W^{(k)})=1$.
\end{lemma}

While we prove the upper bound for deterministic Gaussian measurements, it can be generalized to randomized Gaussian measurements easily due to the linearity of the Fisher information.
According to the lemma, no Gaussian measurement is superior in this averaged sense.

\begin{proof}
By Williamson decomposition
\[W^{(k)}=S(I_2\otimes\Lambda)S^\top,\]
where $\Lambda\succeq I_k$ is a $k\times k$ diagonal matrix and $S$ is symplectic.
$W^{(k)}$ corresponds to a pure Gaussian state if and only if $\Lambda=I_k$.
Define $W_0^{(k)}=SS^\top$, then
\[M\coloneqq W^{(k)}-W_0^{(k)}=S(I_2\otimes(\Lambda-I_k))S^\top\succeq0,\]
and $M=0$ if and only if $W^{(k)}$ corresponds to a pure Gaussian state.
Let
\[D_\theta\coloneqq V_\theta^{\oplus k}+W^{(k)},\qquad D_{0,\theta}\coloneqq V_\theta^{\oplus k}+W_0^{(k)},\]
then $D_\theta=D_{0,\theta}+M$, $\partial_\theta D_{\theta}=\partial_\theta D_{0,\theta}$.
For a centered Gaussian distribution $\mathcal{N}(0,\frac{D_{0,\theta}}{2})$, the Fisher information is (cf.~\cite{mardia1984maximum,malago2015information})
\[J_{W_0^{(k)}}(\theta)=\frac12\tr((D_{0,\theta}^{-1}\partial_\theta D_{0,\theta})^2)=\frac12\tr((D_{0,\theta}^{-1}\partial_\theta D_\theta)^2)=\frac12\norm{D_{0,\theta}^{-1/2}\partial_\theta D_\theta D_{0,\theta}^{-1/2}}_F^2,\]
and similarly,
\[J_{W^{(k)}}(\theta)=\frac12\tr((D_\theta^{-1}\partial_\theta D_\theta)^2)=\frac12\tr((D_{0,\theta}^{-1/2}C_\theta D_{0,\theta}^{-1/2}\partial_\theta D_\theta)^2)=\frac12\norm{C_\theta^{1/2}D_{0,\theta}^{-1/2}\partial_\theta D_\theta D_{0,\theta}^{-1/2}C_\theta^{1/2}}_F^2,\]
where $C_\theta\coloneqq (I_{2k}+D_{0,\theta}^{-1/2}MD_{0,\theta}^{-1/2})^{-1}$.
We have $0\prec C_\theta\preceq I_{2k}$ since $M\succeq0$.
Now suppose $M\neq0$. Then $C_\theta\neq I_{2k}$.
Diagonalize $C_\theta=O\diag(c_1,\ldots,c_{2k})O^\top$, we have $0<c_i\leq1$ and at least one $c_i<1$.
Suppose $c_{i_0}<1$.
Let $A_\theta\coloneqq O^\top D_{0,\theta}^{-1/2}\partial_\theta D_\theta D_{0,\theta}^{-1/2}O$, then $J_{W_0^{(k)}}(\theta)=\frac12\sum_{i,j}(A_\theta)_{ij}^2$ and $J_{W^{(k)}}(\theta)=\frac12\sum_{i,j}c_ic_j(A_\theta)_{ij}^2\leq J_{W_0^{(k)}}(\theta)$.
If $J_{W^{(k)}}(\theta)=J_{W_0^{(k)}}(\theta)$, then the $i_0$-th row and $i_0$-th column of $A_\theta$ must vanish, so $A_\theta$ is singular.
But $\partial_\theta D_\theta$ is not singular, so this is impossible.
Hence we have $J_{W^{(k)}}(\theta)\leq J_{W_0^{(k)}}(\theta)$ with equality if and only if $W^{(k)}$ corresponds to a pure Gaussian state so that $W^{(k)}=W_0^{(k)}$.

In the following we assume that $W^{(k)}$ corresponds to a pure state and show that $\int_0^\pi J_{W^{(k)}}(\theta)d\theta=k\pi(4E-2)$.
Let
\[\lambda_+\coloneqq \frac{\lambda+\lambda^{-1}}{2},\qquad\lambda_-\coloneqq \frac{\lambda-\lambda^{-1}}{2},\qquad L_\theta\coloneqq (R(\theta)\diag(1,-1)R(\theta)^\top)\otimes I_k,\]
then
\[V_\theta^{\oplus k}=\lambda_+I_{2k}+\lambda_-L_\theta,\qquad\partial_\theta^2L_\theta=-4L_\theta,\qquad\partial_\theta^2D_\theta=\partial_\theta^2V_\theta^{\oplus k}=-4\lambda_-L_\theta.\]
Since
\[J_{W^{(k)}}(\theta)=\frac12\tr((D_\theta^{-1}\partial_\theta D_\theta)^2)\]
and
\[\partial_\theta^2\log\det D_\theta=\tr(D_\theta^{-1}\partial_\theta^2D_\theta)-\tr((D_\theta^{-1}\partial_\theta D_\theta)^2),\]
we have
\[J_{W^{(k)}}(\theta)=\frac12\tr(D_\theta^{-1}\partial_\theta^2D_\theta)-\frac12\partial_\theta^2\log\det D_\theta.\]
Since $D_\theta$ is $\pi$-periodic, $\int_0^\pi\partial_\theta^2\log\det D_\theta d\theta=0$, therefore
\[\int_0^\pi J_{W^{(k)}}(\theta)\,\mathrm{d}\theta=-2\lambda_-\int_0^\pi\tr(D_\theta^{-1}L_\theta)\,\mathrm{d}\theta\,.\]
Define
\[\mathcal{I}_0\coloneqq \int_0^\pi\tr(D_\theta^{-1})\,\mathrm{d}\theta\,,\qquad\mathcal{I}_1\coloneqq \int_0^\pi\tr(D_\theta^{-1}L_\theta)\,\mathrm{d}\theta\,.\]
We derive two linear equations for $I_0$ and $I_1$.
First, recall that $W^{(k)}$ has Williamson decomposition $W^{(k)}=SS^\top$. Set
\[\tilde{V}_\theta\coloneqq S^{-1}V_\theta^{\oplus k}S^{-\top}.\]
Then $\tilde{V}_\theta$ corresponds to a pure Gaussian state, so its eigenvalues occur in reciprocal pairs
\[\lambda_1,\lambda_1^{-1},\ldots,\lambda_k,\lambda_k^{-1}.\]
Since
\[D_\theta=S(\tilde{V}_\theta+I_{2k})S^\top,\]
we have
\[\tr(D_\theta^{-1}V_\theta^{\oplus k})=\tr((\tilde{V}_\theta+I_{2k})^{-1}\tilde{V}_\theta)=\sum_{i=1}^k\frac{\lambda_i}{1+\lambda_i}+\frac{\lambda_i^{-1}}{1+\lambda_i^{-1}}=k.\]
Hence
\begin{equation}
\lambda_+\mathcal{I}_0+\lambda_-\mathcal{I}_1=\int_0^\pi\tr(D_\theta^{-1}(\lambda_+I_{2k}+\lambda_-L_\theta))\,\mathrm{d}\theta=\int_0^\pi\tr(D_\theta^{-1}V_\theta^{\oplus k})\,\mathrm{d}\theta=k\pi.\label{eq:integral1}
\end{equation}

Second, suppose $W^{(k)}$ corresponds to $B\in\mathfrak{D}_k$. Recall that $V_\theta^{\oplus k}$ corresponds to $\rho e^{2i\theta}I_k\in\mathfrak{D}_k$ with $\rho=\frac{\lambda-1}{\lambda+1}$.
By Lemma~\ref{lem:siegel_det}, we get
\[\det(D_\theta)=2^{2k}\frac{|\det(I_k-\rho e^{-2i\theta}B)|^2}{\det(I_k-\rho^2I_k)\det(I_k-BB^*)}=2^{2k}\frac{|\det(I_k-\rho e^{-2i\theta}B)|^2}{(1-\rho^2)^k\det(I_k-BB^*)}.\]
Because $\rho\norm{B}_{\sf op}<1$,
\[\log\det(I_k-\rho e^{-2i\theta}B)=-\sum_{q\geq1}\frac{\rho^qe^{-2iq\theta}}{q}\tr(B^q)\]
converges uniformly in $\theta$.
Hence
\[\int_0^\pi\log|\det(I_k-\rho e^{-2i\theta}B)|^2\,\mathrm{d}\theta=2\operatorname{Re}\int_0^\pi\log\det(I_k-\rho e^{-2i\theta}B)\,\mathrm{d}\theta=0.\]
Thus
\begin{align*}
&\frac{d}{d\lambda}\int_0^\pi\log\det D_\theta\,\mathrm{d}\theta\\
=&\frac{d}{d\lambda}\int_0^\pi2k\log2+\log|\det(I_k-\rho e^{-2i\theta}B)|^2-k\log(1-\rho^2)-\log\det(I_k-BB^*)\,\mathrm{d}\theta\\
=&-k\pi\frac{d}{d\lambda}\log(1-\rho^2)=k\pi\frac{\lambda-1}{\lambda(\lambda+1)}.
\end{align*}
On the other hand,
\begin{align}
\frac{d}{d\lambda}\int_0^\pi\log\det D_\theta\,\mathrm{d}\theta&=\int_0^\pi\tr(D_\theta^{-1}\partial_\lambda D_\theta)\,\mathrm{d}\theta\notag\\
&=\int_0^\pi\tr\Bigl(D_\theta^{-1}\Bigl(\frac{1-\lambda^{-2}}{2}I_{2k}+\frac{1+\lambda^{-2}}{2}L_\theta\Bigr)\Bigr)\,\mathrm{d}\theta\notag\\
&=\frac{1-\lambda^{-2}}{2}\mathcal{I}_0+\frac{1+\lambda^{-2}}{2}\mathcal{I}_1.\label{eq:integral2}
\end{align}
Solving~\eqref{eq:integral1} and~\eqref{eq:integral2} gives
\[\mathcal{I}_1=k\pi\frac{1-\lambda}{1+\lambda}.\]
Hence
\[\int_0^\pi J_{W^{(k)}}(\theta)\,\mathrm{d}\theta=-2\lambda_-\mathcal{I}_1=k\pi(\lambda+\lambda^{-1}-2)=k\pi(4E-2).\]
The equality holds when we assume $W^{(k)}$ corresponds to a pure Gaussian state.
When this is not the case, the inequality is strict.
\end{proof}

\subsection{Adaptivity-sample tradeoffs}
\label{sec:adaptivity-sample tradeoffs}
We now embed Lemma~\ref{lem:fisher_of_gaussian} into Lemma~\ref{lem:posterior_fisher_update}, the recursion becomes
\begin{equation}
    \EE[F(p_{H_t})\mid H_{t-1}] \leq F(p_{H_{t-1}}) + k\pi (4E-2)\left(\frac{1}{\pi}+\sqrt{F(p_{H_{t-1}})}\right)\leq F(p_{H_{t-1}}) + 8k\pi E\sqrt{1+F(p_{H_{t-1}})},
\end{equation}
where $k=k(H_{t-1})$ is the number of copies used in the measurement at history $H_{t-1}$. 
For clarity, define $G(H_t) = \frac{1+F(p_{H_t})}{(8\pi E)^2}$, the recursion becomes
\begin{equation}
    \EE[G(H_t)\mid H_{t-1}] \leq G(H_{t-1}) + k(H_{t-1})\sqrt{G(H_{t-1})}. \label{eq:posterior_fisher_recursion_gaussian}
\end{equation}
The initial condition is $G(H_0)=\frac{1}{(8\pi E)^2}$.
By Lemma~\ref{lem:posterior_fisher_success_probability}, the success probability is at most
\begin{equation}
    p_\text{succ} \leq \frac{2\pi \epsilon}{E}\left(\frac{1}{\pi} + 8\pi E\sqrt{\EE[G(H_T)]}\right)\leq \frac{2\epsilon}{E} + 16\pi^2\epsilon \sqrt{\EE[G(H_T)]}.\label{eq:success_probability_bound_gaussian}
\end{equation}
Therefore, to achieve a constant success probability, we need $\EE[G(H_T)] = \Omega(1/\epsilon^2)$. The goal is to upper bound $\EE[G(H_T)]$ using the recursion Eq.~\eqref{eq:posterior_fisher_recursion_gaussian}.
However, the recursion is complicated by the fact that $k$ can depend on $H_{t-1}$. This barrier is clear when we try to get a 2-step recursion relating $\EE[G(H_t)\mid H_{t-2}]$ and $\EE[G(H_{t-2})]$:
\begin{equation}
    \EE[G(H_t)\mid H_{t-2}] = \EE[\EE[G(H_t)\mid H_{t-1}]\mid H_{t-2}] \leq \EE[G(H_{t-1}) + k(H_{t-1})\sqrt{G(H_{t-1})}\mid H_{t-2}].
\end{equation}
The first term is just $\EE[G(H_{t-1})\mid H_{t-2}]$ and can be upper bounded by the previous recursion. Ideally, if $k(H_{t-1})$ is decoupled from $G(H_{t-1})$, then the second term can be upper bounded by $\EE[k(H_{t-1})\mid H_{t-2}] \cdot \sqrt{\EE[G(H_{t-1})\mid H_{t-2}]}$ by Cauchy-Schwarz inequality, and we can also get a closed recursion. However, in general, $k(H_{t-1})$ can be highly correlated with $G(H_{t-1})$, so it is not clear how to upper bound the second term.

We resolve this issue with two different approaches, which lead to different lower bounds that are tight in the bounded-adaptivity ($T\to 1$) and fully adaptive ($T\to \infty$) regimes, respectively.
\begin{enumerate}
    \item The first method is to upper bound $k(H_{t-1})$ by the total sample complexity $S$. When $T=1$, this is lossless since all $S$ samples are indeed used in the single measurement. It is thus natural to expect that when $T$ is small, this method is not too loose.
    \item The second method is to upper bound $x+k\sqrt{x}$ by the $k(H_{t-1})$-fold composition of the function $f(x)=x + \sqrt{x}$ (namely, to show that $x+k\sqrt{x} \leq f^{\circ k}(x)$). Intuitively, it means that whenever we use $k$ copies in a single measurement, we replace it by $k$ adaptive single-copy rounds. Since this will substantially increase the number of rounds, the method is tight only when $T$ is large enough to accommodate the expanded number of rounds, which is the case when $T\to \infty$.
\end{enumerate}

We elaborate on these two methods in the following two subsections.

\subsubsection{Bounded-adaptivity regime}
Suppose $\calP$ is a $T$-round learning protocol with total sample complexity $S$. Then we have $k(H_{t-1}) \leq S$ for every history $H_{t-1}$, so the recursion Eq.~\eqref{eq:posterior_fisher_recursion_gaussian} becomes
\begin{equation}\label{eq:method1_recursion}
    \EE[G(H_t)\mid H_{t-1}] \leq G(H_{t-1}) + S\sqrt{G(H_{t-1})}.
\end{equation}
Let $B_t=\EE[G(H_t)]$. Taking expectations over $H_{t-1}$ on both sides of Eq.~\eqref{eq:method1_recursion}, we have
\begin{equation}
    B_t \leq B_{t-1} + S\EE[\sqrt{G(H_{t-1})}] \leq B_{t-1} + S\sqrt{\EE[G(H_{t-1})]} = B_{t-1} + S\sqrt{B_{t-1}}.
\end{equation}
Let $a_t = 2(1-2^{-t}), b_t=2^{-t}, K_t=(t+1)^2$ such that $K_t\leq K_{t-1}+\sqrt{K_{t-1}}$. We prove by induction that
\begin{equation}\label{eq:method1_induction}
    B_t\leq K_t (1+S)^{a_t} B_0^{b_t}.
\end{equation}
The base case $t=0$ is trivial. For the induction step, we have
\begin{align}
    B_t &\leq K_{t-1} (1+S)^{a_{t-1}} B_0^{b_{t-1}} + S\sqrt{K_{t-1} (1+S)^{a_{t-1}} B_0^{b_{t-1}}} \\
    &\leq K_{t-1} (1+S)^{a_t} B_0^{b_t} + \sqrt{K_{t-1}} (1+S)^{a_t/2+1}B_0^{b_t/2} \\
    &= K_{t-1} (1+S)^{a_t} B_0^{b_t} + \sqrt{K_{t-1}} (1+S)^{a_t} B_0^{b_t} \\
    &\leq K_t (1+S)^{a_t} B_0^{b_t}.
\end{align}
Therefore, Eq.~\eqref{eq:method1_induction} holds for all $t$.
Now we prove the lower bound for $S$ in terms of $T$ and $\epsilon$.
\begin{theorem}
    Let $E>1$, $\epsilon < 1/2$. Suppose $\calP$ is a $T$-adaptive Gaussian learning protocol with sample complexity $S$. Then the success probability of $\calP$ in solving Gaussian state tomography with energy parameter $E$ and accuracy parameter $\epsilon$ is at most 
    \begin{equation}
        p_\text{succ} \leq \frac{2\epsilon}{E} + 16\pi^2\epsilon (T+1)(1+S)^{1-2^{-T}} (8\pi E)^{-2^{-T}}. \label{eq:method1_success_probability_bound}
    \end{equation}
    In conclusion, to achieve a $2/3$ success probability, the sample complexity is at least
    \[
        \Omega(T^{-1} \epsilon^{-2^T/(2^T-1)} E^{1/(2^T-1)}).
    \]
\end{theorem}
\begin{proof}
    By the reduction in Lemma~\ref{lem:reduction}, it suffices to prove the bound for Task~\ref{task:single-parameter}. By Eq.~\eqref{eq:method1_induction}, the averaged posterior Fisher information at round $T$ is upper bounded by $B_T \leq (T+1)^2 (1+S)^{2(1-2^{-T})}B_0^{2^{-T}}$, where $B_0 = 1/(8\pi E)^2$. Plugging this into Eq.~\eqref{eq:success_probability_bound_gaussian}, we obtain the upper bound Eq.~\eqref{eq:method1_success_probability_bound} for the success probability. To achieve a $2/3$ success probability, we need $16\pi^2\epsilon (T+1)(1+S)^{1-2^{-T}} (8\pi E)^{-2^{-T}} = \Omega(1)$, which implies $S = \Omega(T^{-2^T/(2^T-1)} \epsilon^{-2^T/(2^T-1)} E^{1/(2^T-1)})$.
    $T^{-2^T/(2^T-1)}$ can be simplified to $T^{-1}$ since $T^{1/(2^T-1)}=\Theta(1)$.
\end{proof}

Taking $T=1$, we have the following corollary for the nonadaptive protocols.

\begin{corollary}\label{cor:nonadaptive}
    Let $E>1$, $\epsilon < 1/2$. The sample complexity for Gaussian state tomography using nonadaptive Gaussian protocols is at least $\Omega(\frac{E}{\epsilon^2})$. 
\end{corollary}
This generalizes the $\Omega(\frac{E}{\epsilon^2})$ bound in \cite{chen2026towards} for single-copy Gaussian measurements to joint Gaussian measurements.

\subsubsection{Fully adaptive regime}
We now turn to the second method to upper bound $\EE[G(H_t)\mid H_{t-1}]$. Define $f(x)=x+\sqrt{x}$, then $f(x)$ is a strictly increasing and concave function for $x\geq 0$. Denote $f^{\circ k}(x) = f\circ f \circ \cdots \circ f(x)$ as the $k$-fold composition of $f$.
It is easy to verify by induction that $f^{\circ k}(x) \geq x + k\sqrt{x}$ for every $x\ge 0$ and $k\ge 0$.
Therefore, the recursion \eqref{eq:posterior_fisher_recursion_gaussian} implies that
\begin{equation}
    \EE[G(H_t)\mid H_{t-1}] \leq f^{\circ k(H_{t-1})}(G(H_{t-1})).
\end{equation}
Recall that $S$ is the sample complexity of $\calP$ and $s(H_{t-1})$ is the number of samples used up to round $t-1$. We have $k(H_{t-1}) \leq S - s(H_{t-1})$.
Since $f(x)$ is strictly increasing and concave in $x$, so is $f^{\circ S-s(H_{t})}(x)$. Applying $f^{\circ S-s(H_{t})}$ on both sides of the inequality, we have
\begin{equation}
    \EE[f^{\circ (S-s(H_{t}))}(G(H_t))\mid H_{t-1}] \leq f^{\circ (S-s(H_{t}))}(\EE[G(H_t)\mid H_{t-1}]) \leq f^{\circ (S-s(H_{t-1}))}(G(H_{t-1})),
\end{equation}
where the first inequality is by Jensen's inequality and the second inequality is by the monotonicity of $f^{\circ (S-s(H_{t-1}))}$ and $s(H_t)=s(H_{t-1})+k(H_{t-1})$.
Define $R(H_t) \coloneqq f^{\circ (S-s(H_{t}))}(G(H_t))$, then we have $\EE[R(H_t)\mid H_{t-1}] \leq R(H_{t-1})$ for every $t$.
Hence, $\EE[R(H_T)]\leq R(H_0)=f^{\circ S}(G(H_0))$.
Furthermore, since $R(H_T) = f^{\circ (S-s(H_{T}))}(G(H_T)) \geq G(H_T)$, we arrive at the following upper bound for $\EE[G(H_T)]$:
\begin{equation}\label{eq:method2_final_bound}
    \EE[G(H_T)] \leq f^{\circ S}(G(H_0)) = f^{\circ S}\left(\frac{1}{(8\pi E)^2}\right).
\end{equation}

The following lemma controls the growth of $f^{\circ S}(x)$ as $S$ increases.

\begin{lemma}\label{lem:composition_bound}
    Let $N>1$ and $t$ be a nonnegative integer such that $t\leq \log_2\log_2 N$, we have
    \begin{equation}
        f^{\circ t}(1/N)\leq 4 N^{-2^{-t}}. 
    \end{equation}
\end{lemma}
\begin{proof}
    We prove by induction on $t$. The base case $t=0$ is trivial. For the induction step, suppose we have $f^{\circ (t-1)}(1/N)\leq 4 N^{-2^{-(t-1)}}$ for some $t\leq \log_2\log_2 N$, then
    \begin{align}
        \frac{f^{\circ t}(1/N)}{4 N^{-2^{-t}}} \leq \frac{4 N^{-2^{-(t-1)}} + 2 N^{-2^{-t}}}{4 N^{-2^{-t}}} = \frac{1}{N^{2^{-t}}} + \frac{1}{2} \leq 1. &\qedhere
    \end{align}
\end{proof}

We now prove another lower bound for sample complexity $S$

\begin{theorem}
    Let $E>1$, $\epsilon < 1/2$. The sample complexity for Gaussian state tomography using adaptive Gaussian protocols is at least $\Omega(\log\log E)$. 
\end{theorem}
\begin{proof}
    By the reduction in Lemma~\ref{lem:reduction}, it suffices to prove the bound for Task~\ref{task:single-parameter}. By Eq.~\eqref{eq:method2_final_bound} and Lemma~\ref{lem:composition_bound}, the averaged posterior Fisher information at round $T$ is upper bounded by
    \[
        \EE[G(H_T)] \leq 4 (8\pi E)^{-2^{-S+1}}
    \]
    when $S=O(\log\log E)$.
    By Eq.~\eqref{eq:success_probability_bound_gaussian}, the success probability is at most $\frac{2\epsilon}{E} + 16\pi^2\epsilon (8\pi E)^{-2^{-S+1}/2}$. To achieve a $2/3$ success probability, we need $16\pi^2\epsilon (8\pi E)^{-2^{-S}} = \Omega(1)$, which implies $S = \Omega(\log\log E)$.
\end{proof}

\section{Optimal energy-dependent upper bounds with bounded adaptivity}
\label{sec:upper-bounds-bounded-adaptivity}

In this section we restrict to the single-mode case and present a single-copy Gaussian protocol with $T$ rounds of adaptivity whose sample complexity nearly matches the lower bound.
We consider general single-mode Gaussian states, possibly mixed, with mean $\bm$ and covariance $V=R(\theta)\diag(a,b)R(\theta)^\top$, where $0<b\leq a$ and $\theta\in\TT$.
We assume that the energy $\frac{a+b}{4}+\frac12\norm{\bm}_2^2$ is at most $E$, so we have $a+b\leq4E$.
By the psd constraint $V + i\Omega \succeq 0$, we have $ab\geq1$.
Define the condition number $\kappa=\frac ab$.
For $T=1$,~\cite[Algorithm 1]{chen2026towards} already gives a nonadaptive learning algorithm that matches the lower bound up to $\log E$ factors, so we only consider $T\geq2$ here.

\subsection{Useful lemmas}
We begin by collecting several auxiliary results that will be used throughout the analysis.
These include concentration results for statistical estimates and some properties of Gaussian states.

\begin{lemma}[Hanson-Wright inequality for Gaussian vectors]\label{lem:HW}
Let $A\in\mathrm{Sym}_n(\mathbb{R})$.
Let $Z\sim\mathcal{N}(0,I_n)$ be a Gaussian vector of length $n$.
We have for $x\geq0$,
\[\operatorname{Pr}[|Z^\top AZ-\operatorname{tr}(A)|\geq2\|A\|_F\sqrt{x}+2\|A\|_{\sf op}x]\leq2e^{-x}.\]
\end{lemma}

\begin{proof}
Let $Y=Z^\top AZ-\operatorname{tr}(A)$,~\cite[Lemma 3]{cortinovis2022randomized} shows that $Y$ and $-Y$ are sub-Gamma with parameters $(2\|A\|_F^2,2\|A\|_{\sf op})$.
The result then follows from~\cite[Lemma 1]{cortinovis2022randomized}.
\end{proof}

\begin{lemma}\label{lem:cov}
Let $Y_1,\ldots,Y_N$ be $N$ i.i.d. samples from $n$-dimensional Gaussian distribution $\mathcal{N}(\bm,\frac D2)$.
Let $\bar{Y}=\frac1N\sum_{k=1}^NY_k$ and $\hat{D}=\frac{2}{N-1}\sum_{k=1}^N(Y_k-\bar{Y})(Y_k-\bar{Y})^\top$.
Then for all $x\geq0$ and $1\leq i,j\leq n$,
\[\operatorname{Pr}\Bigl[|(\hat{D}-D)_{ij}|\geq2\sqrt{\frac{D_{ii}D_{jj}x}{N-1}}+2\sqrt{D_{ii}D_{jj}}\frac{x}{N-1}\Bigr]\leq2e^{-x}.\]
\end{lemma}

\begin{proof}
$\hat{D}$ has the same distribution law as
\[\hat{D}=\sqrt{D}\Bigl(\frac{1}{N-1}\sum_{k=1}^{N-1}Z_kZ_k^\top\Bigr)\sqrt{D},\]
where $Z_k\sim\mathcal{N}(0,I_n)$.
Therefore, if we set $u_i=\sqrt{D}e_i$, then
\[\hat{D}_{ij}=u_i^\top\Bigl(\frac{1}{N-1}\sum_{k=1}^{N-1}Z_kZ_k^\top\Bigr)u_j.\]
Thus we can write
\[(\hat{D}-D)_{ij}=Z^\top A_{ij}Z-\operatorname{tr}A_{ij},\]
where $Z=(Z_1^\top,\ldots,Z_{N-1}^\top)\sim\mathcal{N}(0,I_{n(N-1)})$ and $A_{ij}=I_{N-1}\otimes\frac{1}{2(N-1)}(u_iu_j^\top+u_ju_i^\top)$.
$A_{ij}$ is symmetric, and we have
\[\Bigl\|\frac{1}{2(N-1)}(u_iu_j^\top+u_ju_i^\top)\Bigr\|_F\leq\frac{1}{2(N-1)}(\|u_iu_j^\top\|_F+\|u_ju_i^\top\|_F)=\frac{\sqrt{D_{ii}D_{jj}}}{N-1},\]
so $\|A_{ij}\|_F\leq\sqrt{\frac{D_{ii}D_{jj}}{N-1}}$.
Similarly,
\[\Bigl\|\frac{1}{2(N-1)}(u_iu_j^\top+u_ju_i^\top)\Bigr\|_{\sf op}\leq\frac{1}{2(N-1)}(\|u_iu_j^\top\|_{\sf op}+\|u_ju_i^\top\|_{\sf op})=\frac{\sqrt{D_{ii}D_{jj}}}{N-1},\]
so $\|A_{ij}\|_{\sf op}\leq\frac{\sqrt{D_{ii}D_{jj}}}{N-1}$.
Hence by Lemma~\ref{lem:HW},
\begin{align*}
\MoveEqLeft\operatorname{Pr}\Bigl[|Z^\top A_{ij}Z-\operatorname{tr}(A_{ij})| \geq2\sqrt{\frac{D_{ii}D_{jj}x}{N-1}}+2\sqrt{D_{ii}D_{jj}}\frac{x}{N-1}\Bigr] \\
&\leq \operatorname{Pr}[|Z^\top A_{ij}Z-\operatorname{tr}(A_{ij})|\geq2\|A_{ij}\|_F\sqrt{x}+2\|A_{ij}\|_{\sf op}x]\leq2e^{-x}\,.\qedhere
\end{align*}
\end{proof}

\begin{lemma}[{Corollary of Weyl's inequality, see, e.g.,~\cite[Corollary III.2.2]{bhatia2013matrix}}]\label{lem:Weyl}
Let $A,\hat{A}\in\mathrm{Sym}_n(\mathbb{R})$ with eigenvalues $\lambda_1\geq\cdots\geq\lambda_n$ and $\hat{\lambda}_1\geq\cdots\geq\hat{\lambda}_n$ respectively.
Then
\[\max_{1\leq i\leq n}|\lambda_i-\hat{\lambda}_i|\leq\|A-\hat{A}\|_{\sf op}.\]
\end{lemma}

\begin{lemma}[Variant of Davis-Kahan $\sin\theta$ theorem,~{\cite[Corollary 1]{yu2015useful}}]\label{lem:sin}
Let $A,\hat{A}\in\mathrm{Sym}_n(\mathbb{R})$ with eigenvalues $\lambda_1\geq\cdots\geq\lambda_n$ and $\hat{\lambda}_1\geq\cdots\geq\hat{\lambda}_n$ respectively.
Fix $1\leq i\leq n$ and assume that $\min\{\lambda_{i-1}-\lambda_i,\lambda_i-\lambda_{i+1}\}>0$, where we define $\lambda_0=\infty$ and $\lambda_{n+1}=-\infty$.
If $v,\hat{v}\in\mathbb{R}^n$ satisfy $Av=\lambda_iv$ and $\hat{A}\hat{v}=\hat{\lambda}_i\hat{v}$, then
\[\sin\angle(v,\hat{v})\leq\frac{2\|\hat{A}-A\|_{\sf op}}{\min\{\lambda_{i-1}-\lambda_i,\lambda_i-\lambda_{i+1}\}},\]
where $\angle(v,\hat{v})$ is the acute angle between $v$ and $\hat{v}$.
\end{lemma}

\begin{lemma}[{Corollary of \cite[Theorem 8]{bittel2025energy}}]\label{lem:trace}
For any two Gaussian states $\rho(\bm,V),\rho(\hat{\bm},\hat{V})$,
\begin{align*}
\MoveEqLeft\Dtr(\rho(\hat{\bm},\hat{V}),\rho(\bm,V))\\
&\leq\frac12\|V^{-1/2}(\hat{\bm}-\bm)\|_2+\frac{1+\sqrt{3}}{8}(\|V^{-1/2}(\hat{V}-V)V^{-1/2}\|_{\sf tr}+\|\hat{V}^{-1/2}(\hat{V}-V)\hat{V}^{-1/2}\|_{\sf tr})\,. \qedhere
\end{align*}
\end{lemma}

\begin{proof}
By \cite[Eq. (89)]{bittel2025energy}, we have
\begin{align*}
\Dtr(\rho(\hat{\bm},\hat{V}),\rho(\bm,V)) &\leq\frac12\|V^{-1/2}(\hat{\bm}-\bm)\|_2\\
&\quad+\frac{1+\sqrt{3}}{4}\int_0^1\|(V+\alpha(\hat{V}-V))^{-1/2}(\hat{V}-V)(V+\alpha(\hat{V}-V))^{-1/2}\|_{\sf tr} \,\mathrm{d}\alpha\,.
\end{align*}
We show that the integrand is convex in $\alpha$.
Let $A=V+\alpha(\hat{V}-V)$, $B=\hat{V}-V$, then the integrand is $\|A^{-1/2}BA^{-1/2}\|_{\sf tr}$.
Let $\tilde{V}=V^{-1/2}\hat{V}V^{-1/2}$, suppose its eigenvalues are $\lambda_i>0$.
We have
\begin{align*}
(V^{-1/2}A^{1/2})(A^{-1/2}BA^{-1/2})(V^{-1/2}A^{1/2})^{-1}&=V^{-1/2}BA^{-1}V^{1/2}\\
&=(V^{-1/2}BV^{-1/2})(V^{-1/2}AV^{-1/2})^{-1}\\
&=(\tilde{V}-I)(I+\alpha(\tilde{V}-I))^{-1}.
\end{align*}
Hence $A^{-1/2}BA^{-1/2}$ is similar to $(\tilde{V}-I)(I+\alpha(\tilde{V}-I))^{-1}$, so they have the same eigenvalues, namely $\frac{\lambda_i-1}{1+\alpha(\lambda_i-1)}$.
Hence $\|A^{-1/2}BA^{-1/2}\|_{\sf tr}=\sum_i\Bigl|\frac{\lambda_i-1}{1+\alpha(\lambda_i-1)}\Bigr|$.
Since for $0\leq\alpha\leq1$, $1+\alpha(\lambda_i-1)>0$, we have
\[\frac{d^2}{d\alpha^2}\Bigl|\frac{\lambda_i-1}{1+\alpha(\lambda_i-1)}\Bigr|=\frac{d^2}{d\alpha^2}\frac{|\lambda_i-1|}{1+\alpha(\lambda_i-1)}=\frac{2|\lambda_i-1|^3}{(1+\alpha(\lambda_i-1))^3}\geq0,\]
so the integrand is indeed convex in $\alpha$, and the integral is bounded by the average of its values at the two endpoints.
\end{proof}

\begin{lemma}[{\cite[Theorem 1]{bittel2025energy}}]\label{lem:heterodyne}
Let $\epsilon,\delta\in(0,1)$, and let $\rho=\rho(\bm,V)$ be an $n$-mode Gaussian state. Then, the heterodyne tomography algorithm (\cite[Algorithm 1]{bittel2025energy}) performs heterodyne measurements on $N$ copies of $\rho$ and returns an estimated state $\hat{\rho}$ satisfying
\[\Dtr(\hat{\rho},\rho)\leq4.3(2n+\operatorname{tr}(V^{-1}))\frac{\sqrt{2n}+\sqrt{2\log(2/\delta)}}{\sqrt{N}}\]
with probability at least $1-\delta$.
\end{lemma}

\subsection{Bounded adaptivity protocol}
\begin{theorem}\label{thm:bounded_adaptivity}
For $T\geq2$, there exists an algorithm using single-copy Gaussian measurements with $T$ rounds of adaptivity that learns any single-mode Gaussian state $\rho(\bm,V)$ to trace distance $\epsilon$ with probability at least $1-\delta$, promised that the energy of $\rho$ is at most $E$.
The sample complexity is
\[O\Bigl((T\epsilon^{-2^T/(2^T-1)}E^{1/(2^T-1)}+\epsilon^{-2})\log(T/\delta)\cdot \log(E)\cdot \log(E/\epsilon\delta)\Bigr).\]
\end{theorem}

\noindent Our algorithm is described in Algorithm~\ref{alg:bounded}.
For the first $T-1$ rounds, we gradually refine our confidence interval for the squeezing direction by running an algorithm that closely resembles the adaptive unsqueezing algorithm in~\cite{bittel2025energy}.
A key difference is that instead of matching the squeezing of the seed covariance matrix with the squeezing of the unknown state, we match the squeezing of the seed covariance matrix with the size of our confidence interval.
For the final round, we run a variant of the nonadaptive algorithm in~\cite{chen2026towards} to turn our confidence interval into accurate estimates of the state.

\begin{algorithm}[htbp]
\DontPrintSemicolon
\caption{Learning single-mode Gaussian states with bounded adaptivity}\label{alg:bounded}
\KwInput{copies of $\rho$, $E$, $T$, $\epsilon\leq0.01$, $\delta$.}
\KwOutput{With probability at least $1-\delta$, output a Gaussian state $\hat{\rho}$ such that $\Dtr(\hat{\rho},\rho)\leq\epsilon$.}
Initialize $\hat{\theta}_0=0$, $\Delta_0=\pi$, $\hat{a}_0=+\infty$, $L=\log\frac{12T}{\delta}$ and $B_0=4\pi E$.\\
Set $N_t=1+\lceil250000\cdot\epsilon^{-1}(E/\epsilon)^{1/(2^T-1)}\log\frac{12T}{\delta}\rceil$.\\
\For{$t=1$ \KwTo $T-1$}{
    Let $\lambda_t=\mathrm{median}\{\frac{\pi}{\Delta_{t-1}},1,\frac23\hat{a}_{t-1}\}$ and $W_t=R(\hat{\theta}_{t-1})\begin{pmatrix}\lambda_t&0\\0&\lambda_t^{-1}\end{pmatrix}R(\hat{\theta}_{t-1})^\top$.\\
    Perform a Gaussian measurement with seed covariance matrix $W_t$ on $N_t$ copies of $\rho$, obtaining outcomes $\{Y_{t,i}\}_{i=1}^{N_t}$.\\
    Let $\bar{Y}_t=\frac{1}{N_t}\sum_{i=1}^{N_t}Y_{t,i}$ and $\hat{V}_t=\frac{2}{N_t-1}\sum_{i=1}^{N_t}(Y_{t,i}-\bar{Y}_t)(Y_{t,i}-\bar{Y}_t)^\top-W_t$.\\
    Let $\hat{a}_t$ and $\hat{b}_t$ be the largest and smallest eigenvalues of $\hat{V}_t$, respectively. Let $\hat{\theta}_t\in[-\frac\pi2,\frac\pi2)$ be the direction of the eigenvector corresponding to $\hat{a}_t$.\\
    \eIf{$\hat{a}_t-\hat{b}_t\leq\frac23\hat{a}_t$}{
        \Return{\textsc{Heterodyne tomography}}($\lceil4400\frac{1}{\epsilon^2}\log\frac4\delta\rceil$ copies of $\rho$, $\frac\delta2$)
    }{
        Let $\hat{s}_t=\max\{\hat{b}_t,0\}+\lambda_t^{-1}+\frac14\lambda_t\Delta_{t-1}^2$, $\Delta_t=\min\{60\sqrt{\frac{\hat{s}_tL}{\hat{a}_t(N_t-1)}},\frac\pi2\}$, and $B_t=200\sqrt{\frac{(1+B_{t-1})L}{N_t-1}}$.
    }
}
\Return{\textsc{Final round algorithm}}(copies of $\rho$, $E$, $[\hat{\theta}_{T-1}-\frac{\Delta_{T-1}}{2},\hat{\theta}_{T-1}+\frac{\Delta_{T-1}}{2})$, $B_{T-1}$, $\epsilon$, $\frac\delta2$)
\end{algorithm}

The heterodyne tomography algorithm is the algorithm in Lemma~\ref{lem:heterodyne}, and the final round algorithm is described in Algorithm~\ref{alg:final}.
We remark that when $T=\Theta(\log\log E)$, by replacing the final round algorithm with another round of the for loop, and setting $N_t=\Theta(\log(T/\delta))$ for $1\leq t\leq T-1$ and $N_T=\Theta(\epsilon^{-2}\log(T/\delta))$, we get an algorithm with sample complexity similar to the fully adaptive tomography algorithm in~\cite{bittel2025energy}.

Ideally, in the first $T-1$ rounds, we obtain good estimates of the parameters.
If we find that the condition number is small, then the heterodyne tomography algorithm should be able to learn the state easily.
If we find that the condition number is large, then we should be able to refine our confidence interval of $\theta$ and pass the final confidence interval to the final round algorithm.

To formally capture this idea, for $0\leq t\leq T-1$, define the \emph{good event} $\mathcal{G}_t$, where $\mathcal{G}_0$ is the event that always occurs, and for $1\leq t\leq T-1$, $\mathcal{G}_t$ is the event that the following conditions hold:
\begin{enumerate}
\item \underline{No branching}: $a-b\geq\frac12a$ and the algorithm does not branch out into heterodyne tomography at round $t$.\label{cond:no_branch}
\item \underline{Good estimate of top eigenvalue}: $|\hat{a}_t-a|\leq60a\sqrt{\frac{L}{N_t-1}}$.\label{cond:a}
\item \underline{Good estimate of angle}: $\theta\in[\hat{\theta}_t-\frac{\Delta_t}{2},\hat{\theta}_t+\frac{\Delta_t}{2})$.\label{cond:theta}
\item \underline{Control on confidence interval length}: $\sqrt{\kappa}\Delta_t\leq B_t$.\label{cond:bound}
\end{enumerate}
For $1\leq t\leq T-1$, also define the \emph{easy event} $\mathcal{E}_t$ to be the event that $b\geq\frac{1}{\sqrt{6}}$ and the algorithm branches out into heterodyne tomography at round $t$. Note that the condition of $a - b \ge \frac{1}{2}a$ from $\mathcal{G}_t$ and the condition of $b\ge \frac{1}{\sqrt{6}}$ cannot both be violated, because $ab\ge 1$. The ``no branching'' part of the good event corresponds to the case where the state is squeezed enough that simply performing heterodyne tomography would not necessarily result in a good estimate, and the easy event corresponds to the case where the state is unsqueezed enough that heterodyne suffices. We will show that at least one of these events happens in round $t$ with high probability, conditional on $\mathcal{G}_{t-1}$:

\begin{lemma}\label{lem:good_or_branch}
Suppose $N_t\geq1+250000L$.
Then for $1\leq t\leq T-1$, $\operatorname{Pr}[\mathcal{G}_t\cup\mathcal{E}_t|\mathcal{G}_{t-1}]\geq1-\frac{\delta}{2T}$.
\end{lemma}

\begin{proof}
Assume $\mathcal{G}_{t-1}$ holds.

\paragraph{Part 1: basic properties of $\lambda_t$.}
We show that for all $1 \le t \le T-1$, 
\begin{equation}
    1\leq\lambda_t\leq \min(a,\frac{\pi}{\Delta_{t-1}})\,.
\end{equation}
We always have $\Delta_{t-1}\leq\pi$, and $\lambda_1 = 1$, so the claim holds for $t=1$.
For $t > 1$, we have $\lambda_t \geq\min\{\frac{\pi}{\Delta_{t-1}},1\}=1$.
Since $\mathcal{G}_{t-1}$ holds we have $|\hat{a}_{t-1}-a|\leq60a\sqrt{\frac{L}{N_t-1}}\leq\frac12a$, so $\hat{a}_{t-1}\leq\frac32a$.
So $\lambda_t\leq\max\{1,\frac23\hat{a}_{t-1}\}\leq a$.
We also have $\lambda_t\leq\max\{\frac{\pi}{\Delta_{t-1}},1\}=\frac{\pi}{\Delta_{t-1}}$.

\paragraph{Part 2: bounding the covariance matrix.}
Denote the error in estimating the angle by $u_{t-1}\triangleq \hat{\theta}_{t-1}-\theta$ (shift by multiples of $\pi$ to make $u_{t-1}\in[-\frac\pi2,\frac\pi2)$). Note that for all $1\leq t < T$,
\begin{equation}
    |u_{t-1}|\leq\frac{\pi}{2}\,.
\end{equation}
For $2\leq t\leq T-1$, this holds by Condition~\ref{cond:theta} of event $\mathcal{G}_{t-1}$, and for $t = 1$ this holds trivially.

In the following, we work in the basis under which $\theta=0$.
This basis is unknown to the algorithm, but it is convenient to do calculations in this basis.
In this basis, we have $V=\diag(a,b)$ and
\[W_t=R(u_{t-1})
\begin{pmatrix}
\lambda_t&0\\
0&\lambda_t^{-1}
\end{pmatrix}
R(u_{t-1})^\top=
\begin{pmatrix}
\lambda_t\cos^2u_{t-1}+\lambda_t^{-1}\sin^2u_{t-1}&(\lambda_t-\lambda_t^{-1})\cos u_{t-1}\sin u_{t-1}\\
(\lambda_t-\lambda_t^{-1})\cos u_{t-1}\sin u_{t-1}&\lambda_t\sin^2u_{t-1}+\lambda_t^{-1}\cos^2u_{t-1}
\end{pmatrix}.
\]
Measuring with seed matrix $W_t$ results in samples $Y_{t,i}$ from the Gaussian distribution $\mathcal{N}(\bm,\frac{D_t}{2})$, where $D_t \triangleq V+W_t$. We can compute
\[(D_t)_{11}=a+\lambda_t\cos^2u_{t-1}+\lambda_t^{-1}\sin^2u_{t-1}\leq a+\lambda_t+\lambda_t^{-1}\leq a+a+1\leq3a,\]
\[(D_t)_{22}=b+\lambda_t\sin^2u_{t-1}+\lambda_t^{-1}\cos^2u_{t-1}\leq b+\frac14\lambda_t\Delta_{t-1}^2+\lambda_t^{-1}\triangleq s_t.\]
Moreover we have
\begin{equation}
    s_t\leq a+1+\frac14\frac{\pi}{\Delta_{t-1}}\Delta_{t-1}^2\leq\Bigl(\frac{\pi^2}{4}+2\Bigr)a\,. \label{eq:s_upper}
\end{equation}

\paragraph{Part 3: concentration of the empirical covariance.}
If we define covariance estimation error
\[H_t=\hat{D}_t-D_t,\]
then by Lemma~\ref{lem:cov}, with probability at least $1-2e^{-x}$, we have $|(H_t)_{ij}|\leq2\sqrt{\frac{D_{ii}D_{jj}x}{N_t-1}}+2\sqrt{D_{ii}D_{jj}}\frac{x}{N_t-1}$.
Take $x=L$ and set $\eta_t\triangleq \sqrt{\frac{L}{N_t-1}}$.
Then by hypothesis $\eta_t\leq\frac{1}{500}$, and each of the following event happens with probability at least $1-\frac{\delta}{6T}$:
\begin{enumerate}
\item $|(H_t)_{11}|\leq2(D_t)_{11}\eta_t+2(D_t)_{11}\eta_t^2\leq12a\eta_t$
\item $|(H_t)_{12}|\leq2\sqrt{(D_t)_{11}(D_t)_{22}}\eta_t+2\sqrt{(D_t)_{11}(D_t)_{22}}\eta_t^2\leq4\sqrt{3as_t}\eta_t$
\item $|(H_t)_{22}|\leq2(D_t)_{22}\eta_t+2(D_t)_{22}\eta_t^2\leq4s_t\eta_t$
\end{enumerate}
Letting $\mathcal{C}_t$ be the intersection of the three events, by the union bound we have $\operatorname{Pr}[\mathcal{C}_t|\mathcal{G}_{t-1}]\geq1-\frac{\delta}{2T}$.
Note that $\mathcal{C}_t$ implies that
\[\|H_t\|_{\sf op}\leq|(H_t)_{11}|+2|(H_t)_{12}|+|(H_t)_{22}|\leq12a\eta_t+8\sqrt{3as_t}\eta_t+4s_t\eta_t\leq60a\eta_t.\]
In the following we assume $\mathcal{C}_t$ happens and show that this implies $\mathcal{G}_t\cup\mathcal{E}_t$, so that $\operatorname{Pr}[\mathcal{G}_t\cup\mathcal{E}_t|\mathcal{G}_{t-1}]\geq1-\frac{\delta}{2T}$.

\paragraph{Part 4: accurate estimate of $a$, and safe branching.}
$\mathcal{C}_t$ ensures that the covariance is estimated accurately, so we get moderately good estimates of $a$ and $b$.
Indeed, since $H_t=\hat{D}_t-D_t=\hat{V}_t-V$, by Lemma~\ref{lem:Weyl}, 
\begin{equation}
    |\hat{a}_t-a|\leq\|H_t\|_{\sf op}\leq60a\eta_t\,, \label{eq:a_error}
\end{equation}
so Condition~\ref{cond:a} holds.
Also by Lemma~\ref{lem:Weyl}, $|\hat{b}_t-b|\leq\|H_t\|_{\sf op}\leq60a\eta_t$.

These estimates of $a,b$ allow us to determine whether the algorithm should branch out into heterodyne tomography or not. If $\hat{a}_t-\hat{b}_t>\frac23\hat{a}_t$, then the algorithm does not branch out.
Then,
\[\hat{b}_t\leq\frac13\hat{a}_t\leq\frac13(1+60\eta_t)a,\]
so
\[a-b\geq a-\frac13(1+60\eta_t)a-60\eta_ta \geq\frac12a.\]
So in this case Condition~\ref{cond:no_branch} holds.
On the other hand, if $\hat{a}_t-\hat{b}_t\leq\frac23\hat{a}_t$ holds, then the algorithm branches out at round $t$.
In this case
\[b\geq\frac13(a-60\eta_ta)-60\eta_ta\geq\frac16a.\]
This together with $ab \ge 1$ implies that $\mathcal{E}_t$ happens.

\paragraph{Part 5: accurate estimate of $\theta$.}
Now it suffices to consider the case where the algorithm does not branch out at round $t$.
Then we have $a-b\geq\frac12a$.
In our $\theta=0$ basis, when diagonalizing $\hat{V}_t=V+H_t$, the large eigenvalue direction should be $v_t=(\cos u_t,\sin u_t)^\top$, and the small eigenvalue direction should be $v_{t,\perp}=(-\sin u_t,\cos u_t)^\top$.
We have $v_{t,\perp}^\top(V+H_t)v_t=0$, expanding, we get
\[(H_t)_{12}(\cos^2u_t-\sin^2u_t)+((H_t)_{22}-(H_t)_{11}+b-a)\sin u_t\cos u_t=0.\]
Since
\[|b-a+(H_t)_{22}-(H_t)_{11}|\geq|b-a|-|(H_t)_{22}|-|(H_t)_{11}|\geq\frac12a-4s_t\eta_t-12a\eta_t\geq\frac37a,\]
we have
\[|\sin2u_t|=\frac{2|(H_t)_{12}||\cos2u_t|}{|b-a+(H_t)_{22}-(H_t)_{11}|}\leq56\sqrt{\frac{s_t}{3a}}\eta_t.\]
By Lemma~\ref{lem:sin}, we have $|\sin u_t|\leq\frac{2\|H_t\|_{\sf op}}{a-b}\leq240\eta_t\leq\frac12$, so $|u_t|\leq\frac\pi6$.
Hence, we can bound our angular estimation error by
\begin{equation}
    |u_t|\leq\frac{\pi}{3\sqrt{3}}|\sin2u_t|\leq20\sqrt{\frac{s_t}{a}}\eta_t\,. \label{eq:theta_error}
\end{equation}

\paragraph{Part 6: accurate estimate of $b$ and $s_t$.}
Here we further refine our error analysis for estimating $b$, and thus for estimating $s_t$.
Because $\hat{b}_t=v_{t,\perp}^\top(V+H_t)v_{t,\perp}$, we have
\[\hat{b}_t-b=v_{t,\perp}^\top Vv_{t,\perp}-b+v_{t,\perp}^\top H_tv_{t,\perp}.\]
Note that
\[|v_{t,\perp}^\top Vv_{t,\perp}-b|=|(a-b)\sin^2u_t|\leq au_t^2\leq400s_t\eta_t^2,\]
and
\begin{align*}
|v_{t,\perp}^\top H_tv_{t,\perp}|&=|(H_t)_{11}\sin^2u_t-2(H_t)_{12}\sin u_t\cos u_t+(H_t)_{22}\cos^2u_t|\\
&\leq|(H_t)_{11}||u_t|^2+2|(H_t)_{12}||u_t|+|(H_t)_{22}|\\
&\leq12a\eta_t\Bigl(20\sqrt{\frac{s_t}{a}}\eta_t\Bigr)^2+2\cdot4\sqrt{3as_t}\eta_t\cdot\Bigl(20\sqrt{\frac{s_t}{a}}\eta_t\Bigr)+4s_t\eta_t\\
&=4800s_t\eta_t^3+160\sqrt{3}s_t\eta_t^2+4s_t\eta_t.
\end{align*}
Hence,
\begin{equation}
    |\hat{b}_t-b|\leq|v_{t,\perp}^\top Vv_{t,\perp}-b|+|v_{t,\perp}^\top H_tv_{t,\perp}|\leq6s_t\eta_t\leq\frac14s_t\,. \label{eq:b_error}
\end{equation}
Note that because $|\hat{s}_t-s_t|=|\max\{\hat{b}_t,0\}-b|$, this also implies that
\begin{equation}
    |\hat{s}_t - s_t| \le \frac{1}{4}s_t\,. \label{eq:s_error}
\end{equation}

\paragraph{Part 7: safe choice for $\Delta_t$.}
Recall from Eq.~\eqref{eq:a_error} that $|\hat{a}_t-a|\leq60a\eta_t\leq\frac12a$. This together with Eq.~\eqref{eq:s_error}, our hypothesis that $\eta_t\leq\frac{1}{500}$, and the bound on $s_t$ in Eq.~\eqref{eq:s_upper} implies that
\[\Delta_t = 60\sqrt{\frac{\hat{s}_t}{\hat{a}_t}}\eta_t\leq60\sqrt{\frac{\frac54s_t}{\frac12a}}\eta_t\leq\frac\pi2\,.\]
Additionally, by invoking Eq.~\eqref{eq:theta_error}, we have
\[\Delta_t = 60\sqrt{\frac{\hat{s}_t}{\hat{a}_t}}\eta_t\geq60\sqrt{\frac{\frac34s_t}{\frac32a}}\eta_t>40\sqrt{\frac{s_t}{a}}\eta_t\geq2|u_t|,\]
meaning that Condition~\ref{cond:theta} holds at round $t$.

\paragraph{Part 8: controlling the refinement.}
It remains to establish Condition~\ref{cond:bound}. We have
\begin{equation}
    \sqrt{\kappa}\Delta_t=\sqrt{\frac ab}\cdot60\sqrt{\frac{\hat{s}_t}{\hat{a}_t}}\eta_t\leq\sqrt{\frac ab}\cdot60\sqrt{\frac{\frac54s_t}{\frac12a}}\eta_t\leq100\sqrt{\frac{s_t}{b}}\eta_t\,, \label{eq:kapdelta}
\end{equation}
where
\[\frac{s_t}{b}=1+\frac{1}{\lambda_tb}+\frac{\lambda_t\Delta_{t-1}^2}{4b}.\]
Since $\lambda_t\leq\frac{\pi}{\Delta_{t-1}}$ and $b=\frac{\sqrt{ab}}{\sqrt{\kappa}}\geq\frac{1}{\sqrt{\kappa}}$, we have $\frac{\lambda_t\Delta_{t-1}^2}{4b}\leq\frac\pi4\Delta_{t-1}\sqrt{\kappa}$.

We claim that 
\begin{equation}
    \frac{1}{\lambda_tb}\leq3+\frac1\pi\sqrt{\kappa}\Delta_{t-1}\,. \label{eq:lambda_b}
\end{equation}
To show this, we proceed by casework:
\begin{itemize}
    \item When $\lambda_t=\frac{\pi}{\Delta_{t-1}}$: we have $\frac{1}{\lambda_tb}=\frac{\Delta_{t-1}}{\pi b}\leq\frac1\pi\sqrt{\kappa}\Delta_{t-1}$.
    \item When $\lambda_t=\frac23\hat{a}_{t-1}$: this only happens for $2\leq t\leq T-1$, and since $\mathcal{G}_{t-1}$ holds, $\hat{a}_{t-1}\geq\frac12a$, so $\lambda_t\geq\frac a3$ and thus $\frac{1}{\lambda_tb}\leq\frac{3}{ab}\leq3$.
    \item When $\lambda_t=1$: if $t=1$, then $\Delta_0=\pi$, so $\frac{1}{\lambda_tb}=\frac1b\leq\sqrt{\kappa}=\frac1\pi\sqrt{\kappa}\Delta_{t-1}$.
    If $2\leq t\leq T-1$, then $\frac23\hat{a}_{t-1}\leq1$, so $\frac12a\leq\hat{a}_{t-1}\leq\frac32$, and $\frac{1}{\lambda_tb}=\frac1b\leq3$.
\end{itemize}
Substituting Eq.~\eqref{eq:lambda_b} into the expression for $\frac{s_t}{b}$, we have
\[\frac{s_t}{b}\leq1+3+\frac1\pi\sqrt{\kappa}\Delta_{t-1}+\frac\pi4\sqrt{\kappa}\Delta_{t-1}\leq4(1+\sqrt{\kappa}\Delta_{t-1}).\]
Substituting this into Eq.~\eqref{eq:kapdelta}, we have
\[\sqrt{\kappa}\Delta_t\leq200\sqrt{1+\sqrt{\kappa}\Delta_{t-1}}\cdot \eta_t,\]
for $1\leq t\leq T-1$.
We will show
\begin{equation}
    \sqrt{\kappa}\Delta_{t-1}\leq B_{t-1}\,.
\end{equation} 
Indeed, if $2\leq t\leq T-1$, since $\mathcal{G}_{t-1}$ holds, in particular Condition~\ref{cond:bound} holds for $t-1$, we have $\sqrt{\kappa}\Delta_{t-1}\leq B_{t-1}$.
If $t=1$, we have $\sqrt{\kappa}\Delta_0=\pi\sqrt{\frac ab}\leq\pi a\leq4\pi E=B_0$.

Hence,
\[\sqrt{\kappa}\Delta_t\leq200\sqrt{1+B_{t-1}}\cdot \eta_t=B_t,\]
establishing Condition~\ref{cond:bound} for round $t$.
\end{proof}

\noindent If $\mathcal{E}_t$ happens for $1\leq t\leq T-1$, then the algorithm branches out and calls the heterodyne tomography algorithm. The correctness for this protocol is standard and guaranteed by the following:
\begin{lemma}\label{lem:branch}
For $1\leq t\leq T-1$, on $\mathcal{E}_t$, heterodyne tomography algorithm returns $\hat{\rho}$ satisfying
$\Dtr(\hat{\rho},\rho)\leq\epsilon$ with probability at least $1-\frac\delta2$.
\end{lemma}

\begin{proof}
On $\mathcal{E}_t$, we have $\operatorname{tr}(V^{-1})=a^{-1}+b^{-1}\leq1+\sqrt{6}$.
We also have $1\leq\log\frac4\delta$, so by Lemma~\ref{lem:heterodyne},
\[\Dtr(\hat{\rho},\rho)\leq4.3(2+\operatorname{tr}(V^{-1}))\frac{\sqrt{2}+\sqrt{2\log\frac4\delta}}{\sqrt{\lceil4400\frac{1}{\epsilon^2}\log\frac4\delta\rceil}}\leq4.3(2+1+\sqrt{6})\frac{2\sqrt{2\log\frac4\delta}}{\sqrt{4400\frac{1}{\epsilon^2}\log\frac4\delta}}\leq\epsilon\]
with probability at least $1-\frac\delta2$.
\end{proof}

On the other hand, if $\bigcap_{t=1}^{T-1}\mathcal{G}_t$ happens, then with our choice of $N_t$, the confidence interval after round $T-1$ should be small enough for the final round algorithm.

\begin{lemma}\label{lem:B}
$B_{T-1}=O(\epsilon(E/\epsilon)^{1/(2^T-1)}+\epsilon^{1/2}(E/\epsilon)^{-1/(2(2^T-1))})$.
\end{lemma}

\begin{proof}
Let $N_t=1+40000\cdot C_tL$, then 
\[C_t=\frac{\lceil250000\cdot\epsilon^{-1}(E/\epsilon)^{1/(2^T-1)}L\rceil}{40000\cdot L}\geq\frac{25}{4}\frac1\epsilon\Bigl(\frac E\epsilon\Bigr)^{1/(2^T-1)}\coloneqq C\geq4.\]
We prove by induction that $B_t\leq B_0^{1/2^t}C^{-1+1/2^t}+\sqrt{2}C^{-1/2}$.
For $t=0$, $B_0\leq B_0+\sqrt{2}C^{-1/2}$, trivial.
Assume the claim holds for $t-1$ where $t\geq1$.
Then
\[B_t=200\sqrt{\frac{(1+B_{t-1})L}{N_t-1}}=\sqrt{\frac{1+B_{t-1}}{C_t}}\leq\sqrt{\frac{1+B_0^{1/2^{t-1}}C^{-1+1/2^{t-1}}+\sqrt{2}C^{-1/2}}{C}}.\]
Because $\sqrt{2}C^{-1/2}\leq1$ and $\sqrt{\sigma+\tau}\leq\sqrt{\sigma}+\sqrt{\tau}$ for all $\sigma,\tau\geq0$, we have
\[B_t\leq\sqrt{\frac{B_0^{1/2^{t-1}}C^{-1+1/2^{t-1}}}{C}}+\sqrt{\frac{2}{C}}=B_0^{1/2^t}C^{-1+1/2^t}+\sqrt{2}C^{-1/2}.\]
Hence by induction the claim holds for $0\leq t\leq T-1$.
Now plug in $C=\frac{25}{4}\frac1\epsilon(E/\epsilon)^{1/(2^T-1)}$ and $B_0=4\pi E$, we have
\[B_{T-1}=O\Bigl(\epsilon\Bigl(\frac E\epsilon\Bigr)^{1/(2^T-1)}+\epsilon^{1/2}\Bigl(\frac E\epsilon\Bigr)^{-1/(2(2^T-1))}\Bigr).\]
\end{proof}

Now we present the final round algorithm.
The algorithm is described in Algorithm~\ref{alg:final}.

\begin{algorithm}[ht!]
\DontPrintSemicolon
\caption{Final round algorithm}\label{alg:final}
\KwInput{copies of $\rho$, $E$, $I$, $B$, $\epsilon\leq0.01$, $\delta$.}
\Promise{$\theta\in I$, $\sqrt{\kappa}|I|\leq B$, $|I|\leq\frac\pi2$, $a-b\geq\frac12a$.}
\KwOutput{With probability at least $1-\delta$, output a Gaussian state $\hat{\rho}$ such that $\Dtr(\hat{\rho},\rho)\leq\epsilon$.}
Set $I_\perp=I+\frac\pi2$, $K=\bigl\lceil(1+50B)\log\frac2\delta\bigr\rceil$, $J=1+\lceil\log_2(8E)\rceil$, $\Phi=\emptyset$, $M=1+\lceil4000000\frac{1}{\epsilon^2}\log\frac{8K(2J+4)}{\delta}\rceil$.\\
\For{$r=1$ \KwTo $K$}{
    Sample $\psi_r$ uniformly from $I_\perp$.\\
    Put $\psi_r$, $\psi_r-\frac\pi2$ and $\psi_r\pm h_j$ where $h_j=\frac{2^j}{8E}$ for $0\leq j\leq J$ in $\Phi$.
}
\For{$\phi\in\Phi$}{
    Perform homodyne measurement with squeezing parameter $E$ along $\phi$ on $M$ copies of $\rho$, obtaining outcomes $\{z_{\phi,i}\}_{i=1}^{M}$.\\
    Let $\hat{m}_\phi=\frac1M\sum_{i=1}^Mz_{\phi,i}$ and $\hat{v}_\phi=\frac2{M-1}\sum_{i=1}^M(z_{\phi,i}-\hat{m}_\phi)^2-\frac1E$.
}
Let $r_0=\arg\min_{1\leq r\leq K}\hat{v}_{\psi_r}$, $\phi_0=\psi_{r_0}$, $\phi_{0,\perp}=\psi_{r_0}-\frac\pi2$ and $\hat{\kappa}=\max\{\frac{\hat{v}_{\phi_{0,\perp}}}{\hat{v}_{\phi_0}},1\}$.\\
Let $j_0=\max\{0\leq j\leq J:h_j\leq\frac{1}{\sqrt{\hat{\kappa}}}\}$. If the set is empty, let $j_0=0$. Let $\phi_\pm=\psi_{r_0}\pm h_{j_0}$.\\
Let $\hat{P}=\frac{{\hat{v}}_{\phi_+}-2{\hat{v}}_{\phi_0}+{\hat{v}}_{\phi_-}}{2\sin^2h_{j_0}}$ and $\hat{Q}=\frac{{\hat{v}}_{\phi_+}-{\hat{v}}_{\phi_-}}{2\sin h_{j_0}\cos h_{j_0}}$.\\
Let $\hat{\xi}_0=\frac12\operatorname{atan2}(\hat{Q},\hat{P})$, $\hat{\theta}=\phi_{0,\perp}-\hat{\xi}_0$, $\hat{c}=\sqrt{\hat{P}^2+\hat{Q}^2}$, $\hat{b}={\hat{v}}_{\phi_0}-\hat{c}\sin^2\hat{\xi}_0$, $\hat{a}=\hat{b}+\hat{c}$, $\hat{\bm}=\hat{m}_{\phi_0}\be_{\phi_0}+\hat{m}_{\phi_{0,\perp}}\be_{\phi_{0,\perp}}$, where $\be_\phi=(\cos\phi,\sin\phi)^\top$.\\
\Return{$\hat{\rho}\Bigl(\hat{\bm},R(\hat{\theta})\begin{pmatrix}\hat{a}&0\\0&\hat{b}\end{pmatrix}R(\hat{\theta})^\top\Bigr)$.}
\end{algorithm}

Here, homodyne measurement with squeezing parameter $E$ along $\phi$ is the measurement that first performs a Gaussian measurement with seed covariance matrix $R(\phi)\begin{pmatrix}E^{-1}&0\\0&E\end{pmatrix}R(\phi)^\top$ and then project the outcome to the direction $\be_\phi$.
In this way, the resulting sample $z$ follows the Gaussian distribution $\mathcal{N}(m_\phi,d_\phi\coloneqq\frac{v_\phi+E^{-1}}{2})$, where $m_\phi=\be_\phi^T\bm$ and $v_\phi=b+(a-b)\cos^2(\phi-\theta)$, see~\cite[Lemma E.1]{chen2026towards}.

\begin{lemma}\label{lem:final}
With probability at least $1-\delta$, Algorithm~\ref{alg:final} outputs a Gaussian state $\hat{\rho}$ such that $\Dtr(\hat{\rho},\rho)\leq\epsilon$.
The sample complexity of Algorithm~\ref{alg:final} is $O(\frac{1+B}{\epsilon^2}\log\frac1\delta\log E(\log\frac1\delta+\log(1+B)+\log\log E))$.
\end{lemma}

\begin{proof}
The proof is very similar to the analysis of Algorithm 1 in~\cite{chen2026towards}.
Take $\alpha=\frac{\epsilon}{100}\leq10^{-4}$ and $l=\log\frac{8|\Phi|}{\delta}$.
Define the small variance direction $\zeta\coloneqq\theta+\frac\pi2\in I_\perp$.
Define $c\coloneqq a-b$.
\paragraph{Part 1: concentration of empirical mean and variance.}
Fix a $\phi\in\Phi$.
The empirical-mean error follows a Gaussian distribution $\hat{m}_\phi-m_\phi\sim\mathcal{N}(0,\frac{v_\phi+E^{-1}}{2M})$.
So by standard tail bound of Gaussian distribution,
\[\operatorname{Pr}\Bigl[|\hat{m}_\phi-m_\phi|\geq\sqrt{\frac{v_\phi+E^{-1}}{2M}}\sqrt{2l}\Bigr]\leq2e^{-l}.\]
Note that $M\geq\frac{400}{\alpha^2}l$ and $v_\phi\geq b\geq\frac1a\geq\frac{1}{4E}$, so $\sqrt{\frac{v_\phi+E^{-1}}{2M}}\sqrt{2l}\leq\alpha\sqrt{\frac{v_\phi}{80}}$, hence
\[\operatorname{Pr}\Bigl[|\hat{m}_\phi-m_\phi|\geq\alpha\sqrt{\frac{v_\phi}{80}}\Bigr]\leq2e^{-l}.\]

Meanwhile, if we let $\hat{d}_\phi=\frac{2}{M-1}\sum_{i=1}^M(z_{\phi,i}-\hat{m}_\phi)^2$, then by Lemma~\ref{lem:cov}, we have
\[\operatorname{Pr}\Bigl[\bigl|\hat{d}_\phi-(v_\phi+E^{-1})\bigr|\geq2(v_\phi+E^{-1})\sqrt{\frac{l}{M-1}}+2(v_\phi+E^{-1})\frac{l}{M-1}\Bigr]\leq2e^{-l}.\]
Note that $\frac{l}{M-1}\leq\frac{\alpha^2}{400}\leq1$, so $(v_\phi+E^{-1})(2\sqrt{\frac{l}{M-1}}+2\frac{l}{M-1})\leq20v_\phi\sqrt{\frac{l}{M-1}}\leq\alpha v_\phi$.
Note also that $\hat{d}_\phi-(v_\phi+E^{-1})=\hat{v}_\phi-v_\phi$, so
\[\operatorname{Pr}[|\hat{v}_\phi-v_\phi|\geq\alpha v_\phi]\leq2e^{-l}.\]
Therefore by the union bound, we have $|\hat{m}_\phi-m_\phi|\leq\alpha\sqrt{\frac{v_\phi}{80}}$ and $|\hat{v}_\phi-v_\phi|\leq\alpha v_\phi$ for all $\phi\in\Phi$ with probability at least $1-2|\Phi|\cdot2e^{-l}=1-\frac\delta2$.

\paragraph{Part 2: existence of a good base angle.}
Let $G=\{\phi\in I_\perp:|\phi-\zeta|_\pi\leq\frac{1}{50\sqrt{\kappa}}\}$, where $|x|_\pi\coloneqq\min_{k\in\ZZ}|x-k\pi|$.
Since $\theta\in I$, $\zeta\in I_\perp$, so $|G|\geq\min\{|I|,\frac{1}{50\sqrt{\kappa}}\}$.
Therefore for uniformly random $\psi_r\in I_\perp$, we have
\[\operatorname{Pr}[\psi_r\in G]\geq\min\{1,\frac{1}{50\sqrt{\kappa}|I|}\}\geq\min\{1,\frac{1}{50B}\}\geq\frac{1}{1+50B}.\]
Hence
\[\operatorname{Pr}[\exists1\leq r\leq K,\psi_r\in G]=1-(1-\operatorname{Pr}[\psi_r\in G])^K\geq1-\Bigl(1-\frac{1}{1+50B}\Bigr)^K\geq1-\frac\delta2.\]
Hence with probability at least $1-\delta$, $|\hat{m}_\phi-m_\phi|\leq\alpha\sqrt{\frac{v_\phi}{80}}$ and $|\hat{v}_\phi-v_\phi|\leq\alpha v_\phi$ for all $\phi\in\Phi$ and $\exists1\leq r_*\leq K$ such that $\psi_{r_*}\in G$.
In the following we show that these imply accurate estimation of $\rho$.

\paragraph{Part 3: accurate estimate of $\kappa$.}
Since $\psi_{r_*}\in G$, $|\psi_{r_*}-\zeta|_\pi\leq\frac{1}{50\sqrt{\kappa}}$.
Therefore
\[v_{\psi_{r_*}}=b+c\sin^2(\psi_{r_*}-\zeta)\leq b+\frac{a-b}{2500\kappa}\leq b+\frac{a}{2500\kappa}=1.0004b.\]
Then
\[(1-\alpha)v_{\phi_0}\leq\hat{v}_{\phi_0}\leq\hat{v}_{\psi_{r_*}}\leq(1+\alpha)v_{\psi_{r_*}}\leq(1+\alpha)1.0004b,\]
so $1.0007b\geq v_{\phi_0}=b+c\sin^2(\phi_0-\zeta)$.
Because $\phi_0,\zeta\in I_\perp$, $|\phi_0-\zeta|_\pi\leq|I_\perp|\leq\frac\pi2$, we have
\[\Bigl(\frac2\pi|\phi_0-\zeta|_\pi\Bigr)^2\leq\sin^2(\phi_0-\zeta)\leq\frac{0.0007b}{a-b}\leq\frac{0.0014}{\kappa},\]
where in the last inequality we used the promise $a-b\geq\frac12a$.
Hence $|\phi_0-\zeta|_\pi\leq\frac{0.06}{\sqrt{\kappa}}$.
We also have $\hat{v}_{\phi_0}\geq(1-\alpha)v_{\phi_0}\geq(1-\alpha)b\geq0.9999b$, so $0.9999b\leq\hat{v}_{\phi_0}\leq(1+\alpha)1.0004b\leq1.0006b$.
Meanwhile, we have
\[a\geq v_{\phi_{0,\perp}}=a-c\sin^2(\phi_0-\zeta)\geq a-0.0007b\geq0.9993a.\]
Therefore
\[0.9992a\leq(1-\alpha)0.9993a\leq(1-\alpha)v_{\phi_{0,\perp}}\leq\hat{v}_{\phi_{0,\perp}}\leq(1+\alpha)v_{\phi_{0,\perp}}\leq(1+\alpha)a\leq1.0001a.\]

Now we have $\frac{\hat{v}_{\phi_{0,\perp}}}{\hat{v}_{\phi_0}}\geq\frac{0.9992a}{1.0006b}>1$, so
\[\frac{0.9999}{\sqrt{\kappa}}\leq\sqrt{\frac{0.9999b}{1.0001a}}\leq\frac{1}{\sqrt{\hat{\kappa}}}=\sqrt{\frac{\hat{v}_{\phi_0}}{\hat{v}_{\phi_{0,\perp}}}}\leq\sqrt{\frac{1.0006b}{0.9992a}}\leq\frac{1.0008}{\sqrt{\kappa}}.\]

\paragraph{Part 4: accurate estimate of $P$ and $Q$.}
Since $1\geq\frac{1}{\sqrt{\kappa}}=\sqrt{\frac ba}=\frac{\sqrt{ab}}{a}\geq\frac1a\geq\frac{1}{4E}$, we have $h_0=\frac{1}{8E}\leq\frac{1}{\sqrt{\hat{\kappa}}}\leq2\leq h_J$.
Hence (recall that we are promised $\kappa\geq2$)
\[\frac{0.9999}{2\sqrt{\kappa}}\leq\frac{1}{2\sqrt{\hat{\kappa}}}\leq h_{j_0}\leq\frac{1}{\sqrt{\hat{\kappa}}}\leq\frac{1.0008}{\sqrt{\kappa}}.\]
Therefore if we define $\xi_0\coloneqq\phi_0-\zeta$, $\xi_\pm\coloneqq\phi_\pm-\zeta$ (shift by multiples of $\pi$ to make them in $[-\frac\pi2,\frac\pi2)$), then
\[|\xi_0|\leq\frac{0.06}{\sqrt{\kappa}},\qquad\xi_+\in\Bigl[\frac{0.43}{\sqrt{\kappa}},\frac{1.07}{\sqrt{\kappa}}\Bigr],\qquad\xi_-\in\Bigl[-\frac{1.07}{\sqrt{\kappa}},-\frac{0.43}{\sqrt{\kappa}}\Bigr]\]
Thus we have
\[v_{\phi_\pm}=b+c\sin^2\xi_\pm\leq b+c\xi_\pm^2\leq b+(a-b)\frac{1.07^2}{\kappa}\leq2.2b.\]
Hence
\[|\hat{v}_{\phi_0}-v_{\phi_0}|\leq1.0007\alpha b,\qquad|\hat{v}_{\phi_\pm}-v_{\phi_\pm}|\leq2.2\alpha b.\]

Define
\[P=\frac{v_{\phi_+}-2v_{\phi_0}+v_{\phi_-}}{2\sin^2h_{j_0}},\qquad Q=\frac{v_{\phi_+}-v_{\phi_-}}{2\sin h_{j_0}\cos h_{j_0}}.\]
One can verify that $P=c\cos2\xi_0$ and $Q=c\sin2\xi_0$.
Now
\[|\hat{P}-P|\leq\frac{|\hat{v}_{\phi_+}-v_{\phi_+}|+2|\hat{v}_{\phi_0}-v_{\phi_0}|+|\hat{v}_{\phi_-}-v_{\phi_-}|}{2\sin^2h_{j_0}}\leq\frac{2.2\alpha b+2\cdot1.0007\alpha b+2.2\alpha b}{2(\frac2\pi\frac{0.9999}{2\sqrt{\kappa}})^2}\leq32\alpha a,\]
while
\[|\hat{Q}-Q|\leq\frac{|\hat{v}_{\phi_+}-v_{\phi_+}|+|\hat{v}_{\phi_-}-v_{\phi_-}|}{2\sin h_{j_0}\cos h_{j_0}}\leq\frac{2.2\alpha b+2.2\alpha b}{\frac2\pi\cdot2\cdot\frac{0.9999}{2\sqrt{\kappa}}}\leq7\alpha\sqrt{ab},\]
here in the penultimate inequality we used $0\leq2h_{j_0}\leq2\frac{1.0008}{\sqrt{2}}\leq\frac\pi2$.

\paragraph{Part 5: accurate estimate of $\theta$, $b$, $a$, and $\bm$.}
We have
\[P=c\cos2\xi_0\geq(a-b)(1-2\xi_0^2)\geq\frac12a(1-2(0.06)^2)\geq0.49a>0,\]
so
\[\hat{P}\geq P-|\hat{P}-P|\geq0.49a-32\alpha a\geq0.48a>0.\]
Hence $2\xi_0=\arctan\frac QP$, $2\hat{\xi}_0=\arctan\frac{\hat{Q}}{\hat{P}}$.
We also have
\[|Q|=|c\sin2\xi_0|\leq(a-b)|2\xi_0|\leq 2a\frac{0.06}{\sqrt{\kappa}}=0.12\sqrt{ab}.\]
Since $\arctan$ is $1$-Lipschitz, we have
\begin{align*}
|\hat{\theta}-\theta|&=|\hat{\xi}_0-\xi_0|\leq\frac12\Bigl|\frac{\hat{Q}}{\hat{P}}-\frac QP\Bigr|=\frac{|P\hat{Q}-\hat{P}Q|}{2P\hat{P}}\leq\frac{|\hat{Q}-Q|P+|\hat{P}-P||Q|}{2P\hat{P}}\\
&=\frac{|\hat{Q}-Q|}{2\hat{P}}+\frac{|\hat{P}-P||Q|}{2P\hat{P}}\leq\frac{7\alpha\sqrt{ab}}{2\cdot0.48a}+\frac{32\alpha a\cdot0.12\sqrt{ab}}{2\cdot0.49a\cdot0.48a}\leq16\frac{\alpha}{\sqrt{\kappa}}.
\end{align*}

Next, because $c=\sqrt{P^2+Q^2}$ and $\hat{c}=\sqrt{\hat{P}^2+\hat{Q}^2}$, we have $|\hat{c}-c|\leq\sqrt{(\hat{P}-P)^2+(\hat{Q}-Q)^2}\leq33\alpha a$.
Hence $\hat{c}\leq c+|\hat{c}-c|\leq a+33\alpha a\leq1.0033a$.
Because $b=v_{\phi_0}-c\sin^2\xi_0$ and $\hat{b}=\hat{v}_{\phi_0}-\hat{c}\sin^2\hat{\xi}_0$, we have
\[|\hat{b}-b|=|\hat{v}_{\phi_0}-v_{\phi_0}-\hat{c}\sin^2\hat{\xi}_0+c\sin^2\xi_0|\leq|\hat{v}_{\phi_0}-v_{\phi_0}|+\hat{c}|\sin^2\hat{\xi}_0-\sin^2\xi_0|+|\hat{c}-c|\sin^2\xi_0.\]
Note that
\begin{align*}
|\sin^2\hat{\xi}_0-\sin^2\xi_0|&=|\sin(\hat{\xi}_0+\xi_0)\sin(\hat{\xi}_0-\xi_0)|\leq|\hat{\xi_0}+\xi_0||\hat{\xi}_0-\xi_0|\leq(|\hat{\xi_0}-\xi_0|+2|\xi_0|)|\hat{\xi_0}-\xi_0|\\
&\leq(16\frac{\alpha}{\sqrt{\kappa}}+2\cdot\frac{0.06}{\sqrt{\kappa}})16\frac{\alpha}{\sqrt{\kappa}}\leq2\frac{\alpha}{\kappa},
\end{align*}
so
\[|\hat{b}-b|\leq1.0007\alpha b+1.0033a\cdot2\frac{\alpha}{\kappa}+33\alpha a\Bigl(\frac{0.06}{\sqrt{\kappa}}\Bigr)^2\leq4\alpha b,\]
and
\[|\hat{a}-a|\leq|\hat{b}-b|+|\hat{c}-c|\leq4\alpha b+33\alpha a\leq37\alpha a.\]
Since $\bm=m_{\phi_0}\be_{\phi_0}+m_{\phi_{0,\perp}}\be_{\phi_{0,\perp}}$, and
\[\|V^{-1/2}\be_{\phi_0}\|_2^2=\frac{\sin^2\xi_0}{a}+\frac{\cos^2\xi_0}{b}\leq\frac{0.06^2}{a}+\frac1b\leq\frac{1.004}{b},\]
\[\|V^{-1/2}\be_{\phi_0,\perp}\|_2^2=\frac{\cos^2\xi_0}{a}+\frac{\sin^2\xi_0}{b}\leq\frac1a+\frac{(0.06/\sqrt{\kappa})^2}{b}\leq\frac{1.004}{a},\]
we get
\begin{align*}
\|V^{-1/2}(\hat{\bm}-\bm)\|_2&=\|(\hat{m}_{\phi_0}-m_{\phi_0})V^{-1/2}\be_{\phi_0}+(\hat{m}_{\phi_{0,\perp}}-m_{\phi_{0,\perp}})V^{-1/2}\be_{\phi_{0,\perp}}\|_2\\
&\leq|\hat{m}_{\phi_0}-m_{\phi_0}|\|V^{-1/2}\be_{\phi_0}\|_2+|\hat{m}_{\phi_{0,\perp}}-m_{\phi_{0,\perp}}|\|V^{-1/2}\be_{\phi_{0,\perp}}\|_2\\
&\leq\alpha\sqrt{\frac{v_{\phi_0}}{80}}\sqrt{\frac{1.004}{b}}+\alpha\sqrt{\frac{v_{\phi_{0,\perp}}}{80}}\sqrt{\frac{1.004}{a}}\\
&\leq\alpha\sqrt{\frac{1.0007b}{80}}\sqrt{\frac{1.004}{b}}+\alpha\sqrt{\frac{a}{80}}\sqrt{\frac{1.004}{a}}\leq0.23\alpha.
\end{align*}

\paragraph{Part 6: accurate estimate of $\rho$.}
Now we have $|\hat{\theta}-\theta|\leq16\frac{\alpha}{\sqrt{\kappa}}$, $|\hat{a}-a|\leq37\alpha a$, $|\hat{b}-b|\leq4\alpha b\leq37\alpha b$.
Let $D=\begin{pmatrix}a&0\\0&b\end{pmatrix}$ and $\hat{D}=\begin{pmatrix}\hat{a}&0\\0&\hat{b}\end{pmatrix}$.
By direct calculation,
\begin{align*}
&\|\hat{D}^{-1/2}R(\hat{\theta}-\theta)\hat{D}R(\hat{\theta}-\theta)^\top\hat{D}^{-1/2}-I\|_{\sf op}\\
=&\frac{2}{-1+\sqrt{1+\frac{4\hat{a}\hat{b}}{(\hat{a}-\hat{b})^2\sin^2(\hat{\theta}-\theta)}}}\\
\leq&\frac{2}{-1+\sqrt{1+\frac{4(1-37\alpha)a(1-4\alpha)b}{(1+37\alpha)^2a^2(16\frac{\alpha}{\sqrt{\kappa}})^2}}}\leq17\alpha,
\end{align*}
so
\begin{align*}
(1-55\alpha)D&\preceq(1-17\alpha)(1-37\alpha)D\preceq(1-17\alpha)\hat{D}\preceq R(\hat{\theta}-\theta)\hat{D}R(\hat{\theta}-\theta)^\top\preceq(1+17\alpha)\hat{D}\\
&\preceq(1+17\alpha)(1+37\alpha)D\preceq(1+55\alpha)D.
\end{align*}
Hence $(1-55\alpha)V\preceq\hat{V}\preceq(1+55\alpha)V$, and also $(1-59\alpha)\hat{V}\preceq V\preceq (1+59\alpha)\hat{V}$.
By Lemma~\ref{lem:trace},
\begin{align*}
&\Dtr(\rho(\hat{\bm},\hat{V}),\rho(\bm,V))\\
\leq&\frac12\|V^{-1/2}(\hat{\bm}-\bm)\|_2+\frac{1+\sqrt{3}}{8}(\|V^{-1/2}(\hat{V}-V)V^{-1/2}\|_{\sf tr}+\|\hat{V}^{-1/2}(\hat{V}-V)\hat{V}^{-1/2}\|_{\sf tr})\\
\leq&\frac12\cdot0.23\alpha+\frac{1+\sqrt{3}}{8}(2\cdot55\alpha+2\cdot59\alpha)\leq80\alpha\leq\epsilon.
\end{align*}

\paragraph{Part 7: sample complexity.}
We have $|\Phi|=K(2+2(J+1))=K(2J+4)=O((1+B)\log\frac1\delta\log E)$.
So the sample complexity is $|\Phi|M=O(\frac{|\Phi|}{\epsilon^2}\log\frac{|\Phi|}{\delta})=O(\frac{1+B}{\epsilon^2}\log\frac1\delta\log E(\log\frac1\delta+\log(1+B)+\log\log E))$.
\end{proof}

\noindent Finally, we are ready to prove Theorem~\ref{thm:bounded_adaptivity}:
\begin{proof}[Proof of Theorem~\ref{thm:bounded_adaptivity}]
We show that Algorithm~\ref{alg:bounded} is the desired algorithm.
By Lemma~\ref{lem:good_or_branch} and the union bound, the disjoint union of the following $T$ events
\[\mathcal{E}_1,\,\mathcal{G}_1\cap\mathcal{E}_2,\,\ldots,\,\mathcal{G}_1\cap\cdots\cap\mathcal{G}_{T-2}\cap\mathcal{E}_{T-1},\,\mathcal{G}_1\cap\cdots\cap\mathcal{G}_{T-1}\]
happens with probability at least $1-\frac\delta2$.
By Lemma~\ref{lem:branch}, conditioned on each of the first $T-1$ events, the returned $\hat{\rho}$ is within $\epsilon$ trace distance with probability at least $1-\frac\delta2$.
Conditioned on the last event, $\mathcal{G}_{T-1}$ happens, so $\theta\in[\hat{\theta}_{T-1}-\frac{\Delta_{T-1}}{2},\hat{\theta}_{T-1}+\frac{\Delta_{T-1}}{2})$ (Condition~\ref{cond:theta}), $\sqrt{\kappa}\Delta_{T-1}\leq B_{T-1}$ (Condition~\ref{cond:bound}), $\Delta_{T-1}\leq\frac\pi2$ (by definition of $\Delta_{T-1}$), $a-b\geq\frac12a$ (Condition~\ref{cond:no_branch}).
Hence the requirements of Algorithm~\ref{alg:final} are met, and by Lemma~\ref{lem:final}, the returned $\hat{\rho}$ is within $\epsilon$ trace distance with probability at least $1-\frac\delta2$.
Hence by the union bound, $\Dtr(\hat{\rho},\rho)\leq\epsilon$ with probability at least $1-\delta$.

For sample complexity, let $r=(E/\epsilon)^{1/(2^T-1)}>1$.
The loop over the first $T-1$ rounds takes at most
\[\sum_{t=1}^{T-1}N_t=O\Bigl(T+T\frac1\epsilon\Bigl(\frac E\epsilon\Bigr)^{1/(2^T-1)}\log\frac T\delta\Bigr)=O\Bigl(T\frac r\epsilon\log\frac T\delta\Bigr)\]
samples.
If the heterodyne tomography algorithm is called, it takes additional $O(\frac{1}{\epsilon^2}\log\frac1\delta)$ samples.
If the final round algorithm is called, it takes additional $O(\frac{1+B_{T-1}}{\epsilon^2}\log\frac1\delta\log E(\log\frac1\delta+\log(1+B_{T-1})+\log\log E))$ samples.
Then by Lemma~\ref{lem:B}, $B_{T-1}=O(\epsilon r+\sqrt{\frac\epsilon r})$.
Since $\sqrt{1\cdot\epsilon r}\geq\sqrt{\frac\epsilon r}$, $\sqrt{\frac\epsilon r}\leq\max\{1,\epsilon r\}$.
Hence $1+B_{T-1}=O(1+\epsilon r)$.
Then $\log(1+B_{T-1})=O(\log(1+\epsilon r))=O(\log r)$.
We conclude that the total number of samples consumed is at most
\begin{align*}
&O\Bigl(T\frac r\epsilon\log\frac T\delta\Bigr)+O\Bigl(\frac{1}{\epsilon^2}\log\frac1\delta\Bigr)+O\Bigl(\frac{1+\epsilon r}{\epsilon^2}\log\frac1\delta\log E\Bigl(\log\frac1\delta+\log r+\log\log E\Bigr)\Bigr)\\
=&O\Bigl((T\epsilon^{-2^T/(2^T-1)}E^{1/(2^T-1)}+\epsilon^{-2})\log\frac T\delta\log E(\log\frac1\delta+\log E+\log\frac1\epsilon)\Bigr).
\end{align*}
\end{proof}

We remark that in the regime $T=\Theta(\log\log E)$, the upper bound $\sqrt{\kappa}\Delta\leq B$ is small enough so that the final round algorithm can be replaced by homodyne measurements in the directions $\hat{\theta}$ and $\hat{\theta}+\pi/2$. This will recover the $\log\log E$ sample complexity in \cite{bittel2025energy}.

\section{Energy-independent upper bound with non-Gaussian measurements}

\subsection{Reducing to the centered case}

In this section, we show a simple reduction from general pure Gaussian states to ones with \emph{zero displacement}.

\begin{lemma}\label{lem:reduce_to_meanzero}
    Suppose there is a learning protocol which, for any $0 < \epsilon \le 1/2$, $0<\delta \le 1$ and given $m(n,\epsilon,\delta)$ copies of a pure Gaussian state $\rho(0,V)$, outputs the covariance matrix $\hat{V}$ of a pure Gaussian state for which $\Dtr(\rho(0,V),\rho(0,\hat{V})) \le \epsilon$ 
    with probability at least $1 - \delta$. Then there is a learning protocol which, given $2m(n,\epsilon/2,\delta/2) + O((n + \log(1/\delta))/\epsilon^2)$ 
    copies of a pure Gaussian state $\rho(\bm, V)$, outputs an estimate $(\hat{\bm}, \hat{V})$ for which $\Dtr(\rho(\bm,V), \rho(\hat{\bm},\hat{V})) \le \epsilon$.
\end{lemma}

\begin{proof}
    Apply a 50:50 beam splitter to $m(n, \epsilon, \delta)$ pairs of copies of $\rho$ to get $m$ copies of the state $\rho(0, V)\otimes \rho(\sqrt{2}\bm, V)$. Apply the learning protocol to the copies of $\rho(0, V)$ to obtain $\hat{V}$. Let $\hat{V} = \hat{S}\hat{S}^\top$ denote the Williamson decomposition, and let $\hat{U}^\dagger$ denote the Gaussian unitary corresponding to $\hat{S}^{-1}$. On each of the copies of $\rho(\sqrt{2}\bm, V)$, apply $\hat{U}^\dagger$ to obtain copies of $\rho(\sqrt{2}\hat{S}^{-1}\bm, \hat{S}^{-1}V\hat{S}^{-1\top})$. By Lemma~\ref{lem:basic_unsqueeze}, the eigenvalues of the covariance are within $[1-O(\epsilon), 1+O(\epsilon)]$, so performing heterodyne measurements on $O((n + \log(1/\delta))/\epsilon^2)$ copies of this state results in a mean estimate $\sqrt{2}\hat{\bn}$ for which $\norm{\hat{\bn} - \hat{S}^{-1}\bm}_2 \le \epsilon$ with probability at least $1 - \delta$. We have
    \begin{equation}
        (\hat{S}\hat{\bn} - \bm)^\top (\hat{S}\hat{S}^\top)^{-1} (\hat{S}\hat{\bn} - \bm) = \norm{\hat{S}^{-1}(\hat{S}\hat{\bn} - \bm)}^2_2 \le \epsilon^2\,.
    \end{equation}
    Furthermore, $\det(2\hat{S}\hat{S}^\top) = 2^{2n}$ by purity, so by Lemma~\ref{lem:distance_pure},
    \begin{equation}
        F(\rho(\hat{S}\hat{\bn}, \hat{S}\hat{S}^\top), \rho(\bm, \hat{S}\hat{S}^\top)) \ge \exp(-\epsilon^2/4)\,,
    \end{equation}
    and thus their trace distance is at most $\epsilon/2$.
    We conclude that
    \begin{align}
        \Dtr(\rho(\hat{S}\hat{\bn}, \hat{S}\hat{S}^\top),\rho(\bm, SS^\top)) &\le \Dtr(\rho(\hat{S}\hat{\bn},\hat{S}\hat{S}^\top),\rho(\bm, \hat{S}\hat{S}^\top)) +\Dtr(\rho(\bm, \hat{S}\hat{S}^\top),\rho(\bm, SS^\top)) \\
        &\le \epsilon/2 +\Dtr(\rho(0,\hat{S}\hat{S}^\top),\rho(0,SS^\top)) < 2\epsilon
    \end{align}
    with probability at least $1 - 2\delta$. Rescaling $\epsilon$ and $\delta$ gives the claim.
\end{proof}

\noindent The argument above uses the following standard estimate:

\begin{lemma}\label{lem:basic_unsqueeze}
    Suppose $S, \hat{S}$ are symplectic matrices for which $\Dtr(\rho(0,SS^\top),\rho(0,\hat{S}\hat{S}^\top)) \le \epsilon$ for $0 \le \epsilon \le 1/2$. Then
    \begin{equation}
        \norm{\hat{S}^{-1}SS^\top\hat{S}^{-1\top} - \Id}_{\sf op} \le e^{4\epsilon} - 1 = O(\epsilon)\,.
    \end{equation}
\end{lemma}

\begin{proof}
    Define $T \triangleq \hat{S}^{-1}S$, which is also symplectic. If $T$ has Euler decomposition $T = O_1 D O_2$ where $D = \mathrm{diag}(e^{r_1},e^{-r_1},\ldots,e^{r_n},e^{-r_n})$, then $TT^\top$ has eigenvalues $\{e^{2r_i}, e^{-2r_i}\}$ and $\norm{TT^\top - \Id}_{\sf op} = \max_i e^{2r_i} - 1$. It remains to relate this to the trace distance between the two states. By Lemma~\ref{lem:distance_pure}, we have
    \begin{equation}
        F(\rho(0,SS^\top), \rho(0,\hat{S}\hat{S}^\top)) = \frac{2^n}{\sqrt{\det(SS^\top + \hat{S}\hat{S}^\top)}} = \prod^n_{i=1} \frac{2}{\sqrt{(1 + e^{-2r_i})(1 + e^{2r_i})}} = \prod^n_{i=1} \sech(r_i)\,.
    \end{equation}
    As the fidelity is at least $1 - \epsilon^2$ by assumption,
    and each $\sech(r_i) \in (0,1]$, we have $\sech(\max_i r_i) \ge 1 - \epsilon^2$ , or equivalently, $\max_i r_i \le \arccosh(\frac{1}{1 - \epsilon^2}) \le 2\epsilon$, from which the claim follows.
\end{proof}

\subsection{Entangled non-Gaussian protocol}

In this section we give a protocol for learning any pure Gaussian state with energy-independent sample complexity, using a certain entangled non-Gaussian measurement. The measurement can be thought of as a bosonic analogue of Hayashi's optimal entangled measurement for pure-state tomography in finite dimensions~\cite{hayashi1998asymptotic}.

\begin{theorem}\label{thm:main_energyfree}
    There is an absolute constant $C > 0$ such that for any $n \ge 1$, $0 < \epsilon \le1$, and $0 < \delta < 1$, there is a POVM on 
    \begin{equation}
        m = \frac{C}{\epsilon^2}(n^2 + \log(1/\delta)) \label{eq:mbound}
    \end{equation}
    copies of an unknown pure $n$-mode Gaussian state $\rho(0,V)$ whose outcome specifies a pure Gaussian state $\rho(0,\hat{V})$ such that $\Dtr(\rho(0,V),\rho(0,\hat{V})) \le \epsilon$ with probability at least $1 - \delta$.

    By Lemma~\ref{lem:reduce_to_meanzero}, the same learning guarantee holds for pure Gaussian states of arbitrary displacement.
\end{theorem}

\subsubsection{Construction of POVM}
\label{sec:povm}

Recall from Section~\ref{sec:siegel} the Siegel disk $\mathfrak{D}_n$ consisting of $n\times n$ complex symmetric matrices for which $KK^\dagger \prec I$. For any $K\in \mathfrak{D}_n$, recall from Eq.~\eqref{eq:diskparam} that the corresponding pure Gaussian state is given by $\ket{\psi_K} = \det(I - KK^\dagger)^{1/4} \exp(\frac{1}{2}\hat{\ba}^{\dagger\top} K \hat{\ba}^\dagger)\ket{0}$. Lemma~\ref{lem:overlap} gives an expression for the inner product between any pair of such states.

Let $\mathrm{d}K$ denote the Lebesgue measure on $\mathfrak{D}_n$. Define the following measure over $\mathfrak{D}_n$:
\begin{equation}
    \mathrm{d}\mu(K) \triangleq \frac{\mathrm{d}K}{\det(I - KK^\dagger)^{n+1}}\,, \label{eq:mainmeasure}
\end{equation}
Note that this is not a probability measure because of the singularity at the boundary of the Siegel disk.
Nevertheless, we will be able to construct from $\mathrm{d}\mu$ a valid POVM.

First, define the subspace
\begin{equation}
    V_m \triangleq \mathrm{span}(\{\ket{\psi_K}^{\otimes m}\}_{K\in\mathfrak{D}_n})
\end{equation}
and denote by $\Pi_m$ the orthogonal projector onto $V_m$.

We have the following integrability results:

\begin{lemma}\label{lem:matrix_integrability}
    Given $\gamma\in\mathbb{R}$, define the normalizing constant
    \begin{equation}
        Z_\gamma \triangleq \int_{\mathfrak{D}_n} \det(I - KK^\dagger)^\gamma\,\mathrm{d}K\,.
    \end{equation}
    If $\gamma > -1$, then $Z_\gamma < \infty$.
\end{lemma}

\begin{lemma}\label{lem:bergman}
    If $m > 2n$, then
    \begin{equation}
        \frac{1}{Z_{m/2 - n - 1}}\int \ket{\psi_K}\bra{\psi_K}^{\otimes m}\,\mathrm{d}\mu(K) = \Pi_m\,.
    \end{equation}
\end{lemma}

\noindent We can thus define the following POVM on $V_m$: 
\begin{equation}
    M_m(\mathrm{d}K) = \frac{1}{Z_{m/2-n-1}} \ket{\psi_K}^{\otimes m} \bra{\psi_K}^{\otimes m} \frac{\mathrm{d}K}{\det(I - KK^\dagger)^{n + 1}}\,.
\end{equation}
Technically this is only a POVM on $V_m$, as $M_m(\mathfrak{D}_n) = \Pi_m$, so to define a POVM on the full Hilbert space we can include the orthogonal projector $\Pi^\perp_m$. This will not matter, however, because we will only ever consider measuring states of the form $\ket{\psi}^{\otimes m}$ with this POVM.

\subsubsection{Matrix integral estimates}
\label{sec:matrix_integrability}

In this section, we study the matrix integral from Lemma~\ref{lem:matrix_integrability}. We first pass to a different coordinate system for the space of complex symmetric matrices. Every such matrix admits a \emph{Takagi decomposition} $K = UDU^\top$, where $U$ is unitary and $D$ is the diagonal matrix given by the singular values $\lambda_1,\ldots,\lambda_n$ of $K$. Note that for $K\in\mathfrak{D}_n$, we have that $0 \le \lambda_i < 1$ for all $i$. By \cite[Eq. (3.5.2)]{hua1963harmonic}, the Lebesgue measure can be written as
\begin{equation}
    \mathrm{d}K = 2^n \prod_{i<j} |\lambda^2_i - \lambda^2_j| \,\lambda_1\cdots \lambda_n \,\mathrm{d}\lambda_1\cdots\mathrm{d}\lambda_n\,\mathrm{d}\omega(U\mathrm{Z}^n_2)\,,
\end{equation}
where $\mathrm{d}\lambda_1,\ldots,\mathrm{d}\lambda_n$ is the Lebesgue measure on Euclidean space, and $\mathrm{d}\omega$ denotes the Haar measure over the unitary group modulo the action of diagonal matrices with $\pm 1$ entries. 

\begin{proof}[Proof of Lemma~\ref{lem:matrix_integrability}]
Changing variables to $t_i = \lambda^2_i$ and noting that $\det(I - KK^\dagger)^\gamma = \prod_i (1 - t_i)^\gamma$, we conclude that $Z_\gamma$ can be rewritten as a simple scalar integral, up to an $n$-dependent constant:
\begin{align}
    Z_\gamma &\propto \int_{[0,1]^n}\prod_{i < j} |t_i - t_j| \cdot \prod_i (1 - t_i)^\gamma\,\mathrm{d}\vec{t}\,. \label{eq:Zgam}
	\intertext{Note that $|t_i - t_j| \le 1$ because $K \in \mathfrak{D}_n$, so we can bound the above by}
    &\le \Bigl(\int^1_0 (1 - t)^\gamma\,\mathrm{d}t\Bigr)^n\,,
\end{align}
which is finite for any $\gamma > -1$.
\end{proof}

\noindent In fact, there is an exact expression for the integral in Eq.~\eqref{eq:Zgam} when $\gamma > -1$:

\begin{theorem}[Selberg's formula, \cite{selberg1944bemerkninger}]\label{thm:selberg}
	If $\gamma > -1$, then
	\begin{equation}
		\int_{[0,1]^n}\prod_{i < j} |t_i - t_j| \cdot \prod_i (1 - t_i)^\gamma\,\mathrm{d}\vec{t} = \prod_{0 \le j < n} \frac{\Gamma(1 + j/2)\Gamma(\gamma + 1 + j/2)\Gamma(1 + (j+1)/2)}{\Gamma(\gamma + 2 + (n + j - 1)/2)\Gamma(3/2)}\,.
	\end{equation}
\end{theorem}

\noindent We use this to deduce the following estimate for $Z_\gamma$ which we will use near the end of our argument:

\begin{lemma}\label{lem:hua_estimate}
	Suppose $\gamma\ge 0$. 
    Then
    \begin{equation}
        Z_{\gamma/2} \le 2^{n(n+1)} Z_{\gamma}\,.
    \end{equation}
\end{lemma}

\begin{proof}
	Because $Z_\gamma$ is equal to Selberg's integral up to an $n$-dependent constant, by Selberg's formula we have
	\begin{align}
		\frac{Z_{\gamma/2}}{Z_\gamma} &= \prod_{0 \le j < n} \frac{\Gamma(\gamma/2 + 1 + j/2)}{\Gamma(\gamma/2 + 2 + (n + j - 1)/2)} \Big\slash \frac{\Gamma(\gamma + 1 + j/2)}{\Gamma(\gamma + 2 + (n + j - 1)/2)} \,. \\
		\intertext{For convenience, define $q(x) = \Gamma(x + (n+1)/2)/\Gamma(x)$ so that the above can be rewritten as}
		&= \prod_{0 \le j < n} \frac{q(\gamma + 1 + j/2)}{q(\gamma/2 + 1 + j/2)}\,. \label{eq:Zratio}
	\end{align}
	Let $h \triangleq \Gamma'/\Gamma$ denote the digamma function. We have
	\begin{equation}
		\frac{\mathrm{d}}{\mathrm{d}x}\log q(x) = h(x + (n+1)/2) - h(x) = \int^{x+(n+1)/2}_x h'(y) \,\mathrm{d}y\,.
	\end{equation}
    Note that $h'(y) = \sum^\infty_{j=0} \frac{1}{(y+j)^2}\le1/y^2+1/y \le 2/y$ for $y\geq1$.
    We conclude that for any $1 \le x \le x'$,
    \begin{equation}
    \begin{split}
        \log\frac{q(x')}{q(x)}
        &=\int^{x'}_{x} \int^{\tilde{x} + (n+1)/2}_{\tilde{x}} h'(y)\,\mathrm{d}y\ \mathrm{d}\tilde{x}
        \le\int_x^{x'}2\ln\Bigl(1+\frac{(n+1)/2}{\tilde{x}}\Bigr)\,\mathrm{d}\tilde{x}\\
        &\le \int^{x'}_x \frac{n+1}{\tilde{x}}\, \mathrm{d}\tilde{x} = (n+1)\log(x'/x)\,.
    \end{split}
    \end{equation}
    For all $j\ge 0$, $1\le\gamma/2+1+j/2\le\gamma+1+j/2$ and $\frac{\gamma + 1 + j/2}{\gamma/2 + 1 + j/2} \le 2$, so applying the above to the factors of Eq.~\eqref{eq:Zratio}, we conclude that $Z_{\gamma/2}/Z_\gamma \le 2^{n(n+1)}$ as claimed.
\end{proof}

\subsubsection{Bergman space}

To compute the integral in Lemma~\ref{lem:bergman}, we first relate the space $V_m$ to a certain function space called a (weighted) \emph{Bergman space}. Given $\lambda > n$, define the measure
\begin{equation}
    \mathrm{d}\nu_\lambda(K) \triangleq \frac{1}{Z_{\lambda - n - 1}} \det(I - KK^\dagger)^{\lambda - n - 1}\,\mathrm{d}K\,,
\end{equation}
noting that this is a valid probability measure by Lemma~\ref{lem:matrix_integrability}.
Let $A^2_\lambda(\mathfrak{D}_n)$ denote the space of holomorphic functions on $\mathfrak{D}_n$ that are square-integrable with respect to $\nu_\lambda$. It is a classical fact that such spaces are reproducing kernel Hilbert spaces (RKHS)~\cite{faraut1990function}; we state below a version in our notation:

\begin{lemma}[Eq.~(25) in~\cite{englivs2024weighted}]
    Define the reproducing kernel $C_\lambda(K,K') = \det(I - KK'^\dagger)^{-\lambda}$. Then for every $f \in A^2_\lambda(\mathfrak{D}_n)$,
    \begin{equation}
        f(K) = \int_{\mathfrak{D}_n} C_\lambda(K,K') f(K')\,\mathrm{d}\nu_\lambda(K')\,.
    \end{equation}
\end{lemma}

\noindent From this RKHS property, we can deduce an integrability result that closely mirrors Lemma~\ref{lem:bergman}. Given $K$, define $C_K(\cdot) \in A^2_\lambda(\mathfrak{D}_n)$ by $C_K(K') = C_\lambda(K,K')$, and define its normalization $c_K = C_K / \norm{C_K}$, noting that $\norm{C_K}^2 = C_\lambda(K,K) = \det(I - KK^\dagger)^{-\lambda}$ so that
\begin{equation}
    c_K(K') = C_K(K')\cdot \det(I - KK^\dagger)^{\lambda/2}\,.
\end{equation}

\begin{lemma}\label{lem:bergman2}
    Let $\Pi_{A^2_\lambda(\mathfrak{D}_n)}$ denote projector onto the weighted Bergman space. Then
    \begin{equation}
        \Pi_{A^2_\lambda(\mathfrak{D}_n)} = Z_{\lambda - n -1}^{-1} \int_{\mathfrak{D}_n} \ket{c_K}\bra{c_K}\,\mathrm{d}\mu(K)\,.
    \end{equation}
\end{lemma}

\begin{proof}
    For any $g\in A^2_\lambda(\mathfrak{D}_n)$, we have
    \begin{equation}
        \braket{c_K|g} = g(K)\cdot \det(I - KK^\dagger)^{\lambda/2}\,,
    \end{equation}
    so for $f,g\in A^2_\lambda(\mathfrak{D}_n)$,
    \begin{equation}
        \int_{\mathfrak{D}_n} \braket{f|c_K}\braket{c_K|g}\,\mathrm{d}\mu(K) = \int_{\mathfrak{D}_n} \overline{f(K)} g(K) \det(I - KK^\dagger)^\lambda\,\mathrm{d}\mu(K) = Z_{\lambda - n - 1} \bra{f}\Pi_{A^2_\lambda(\mathfrak{D}_n)}\ket{g}
    \end{equation}
    as claimed.
\end{proof}

\begin{proof}[Proof of Lemma~\ref{lem:bergman}]
    Observe that
    \begin{equation}
        \int_{\mathfrak{D}_n} C_\lambda(K'',K) C_\lambda(K',K'') \,\mathrm{d}\nu_\lambda(K'') = C_\lambda(K',K) = \det(I - K'K^\dagger)^{-\lambda} = \det(I - K^\dagger K')^{-\lambda}\,.
    \end{equation}
    Therefore,
    \begin{equation}
        \braket{c_K|c_{K'}}_{A^2_\lambda(\mathfrak{D}_n)} = \frac{\det(I - KK^\dagger)^{\lambda/2} \det(I - K'K'^\dagger)^{\lambda/2}}{\det(I - K^\dagger K')^\lambda}\,.
    \end{equation}
    Note the resemblance to the overlap formula from Lemma~\ref{lem:overlap}. Indeed, this implies that for $\lambda = m/2$, the map $c_K \mapsto \ket{\psi_K}^{\otimes m}$ extends to a unitary isomorphism between $A^2_\lambda(\mathfrak{D}_n)$ and the original space $V_m$, thus concluding the proof of Lemma~\ref{lem:bergman}.
\end{proof}

\subsubsection{Covariance with respect to Gaussian unitaries}
\label{sec:covariant}

The Siegel disk is naturally equipped with the action of the symplectic group, corresponding to the action of centered Gaussian unitaries on the space of centered pure Gaussian states. In this section, we will not need anything about the particular form of this action beyond the following standard facts, which we derive from first principles in Appendix~\ref{app:metaplectic} for the sake of completeness.

\begin{restatable}[Transitivity, e.g., Theorem 7 of~\cite{pantaleoni2026gaussian}]{lemma}{transitivity}\label{lem:transitivity}
    The action of $\mathrm{Sp}(2n)$ on $\mathfrak{D}_n$ is transitive, that is, for every $K \in \mathfrak{D}_n$, there exists $S\in \mathrm{Sp}(2n)$ which maps the vacuum $0$ to $S\cdot 0= K$.
\end{restatable}

\begin{restatable}[Metaplectic representation, e.g., Definition 15 from~\cite{pantaleoni2026gaussian}]{lemma}{metaplectic}\label{lem:metaplectic}
    Given $S\in\mathrm{Sp}(2n)$, there is an associated centered Gaussian unitary $U_S$ such that for every $K\in\mathfrak{D}_n$, there is a phase $\phi(S,K)$ for which
    \begin{equation}
        U_S \ket{\psi_K} = e^{i\phi(S,K)} \ket{\psi_{S\cdot K}}\,.
    \end{equation}
\end{restatable}

\begin{restatable}[Invariance of $\mathrm{d}\mu$, e.g., (8.3.20) from~\cite{perelomov1986}]{lemma}{invariance}\label{lem:invariance}
    The measure $\mathrm{d}\mu$ is invariant under the action of $\mathrm{Sp}(2n)$.
\end{restatable}

\noindent Altogether, these ingredients allow us to deduce that the POVM constructed in Section~\ref{sec:povm} is covariant with respect to the action of centered Gaussian unitaries:

\begin{lemma}\label{lem:covariance}
    For any centered Gaussian unitary $U_S$ associated to a symplectic matrix $S$, and for every Borel set $\mathcal{B}\subseteq \mathfrak{D}_n$, we have the following equality of operators on $V_m$:
    \begin{equation}
        U_S^{\otimes m} M_m(\mathcal{B}) U_S^{\otimes m\dagger} = M_m(S\cdot \mathcal{B})\,.
    \end{equation}
\end{lemma}

\begin{proof}
	By Lemma~\ref{lem:metaplectic}, we have
	\begin{equation}
		U_S^{\otimes m} \ket{\psi_K}^{\otimes m} \bra{\psi_K}^{\otimes m} U_S^{\otimes m\dagger} = \ket{\psi_{S\cdot K}}^{\otimes m}\bra{\psi_{S\cdot K}}^{\otimes m}\,,
	\end{equation}
	so by the definition of the POVM $M_m$ and invariance of $\mathrm{d}\mu$ (Lemma~\ref{lem:invariance}),
	\begin{align}
		U_S^{\otimes m} M_m(\mathcal{B}) U_S^{\otimes m\dagger} &= \frac{1}{Z_{m/2-n-1}} \int_{\mathcal{B}} \ket{\psi_{S\cdot K}}^{\otimes m} \bra{\psi_{S\cdot K}}^{\otimes m} \,\mathrm{d}\mu(K) \\
		&= \frac{1}{Z_{m/2-n-1}} \int_{S\cdot \mathcal{B}} \ket{\psi_{K}}^{\otimes m} \bra{\psi_{K}}^{\otimes m} \,\mathrm{d}\mu(K) \\
		&= M_m(S\cdot \mathcal{B})\,.\qedhere
	\end{align}
\end{proof}

\noindent Covariance (Lemma~\ref{lem:covariance}) and transitivity (Lemma~\ref{lem:transitivity}) allow us to essentially reduce our analysis of the distribution over measurement outcomes to the special case where the underlying state is vacuum:

\begin{corollary}\label{cor:reduce_to_vacuum}
	Let $K^*\in\mathfrak{D}_n$, and let $S\in\mathrm{Sp}(2n)$ be an element for which $S\cdot 0 = K^*$. If $K$ is the random outcome of measuring $\ket{\psi_{K*}}^{\otimes m}$ with the POVM $M_m$, then the distribution of $S^{-1}\cdot K$ is the same as that of the random outcome of measuring the vacuum state $\ket{\psi_0}^{\otimes m}$ with $M_m$.
\end{corollary}

\subsubsection{Analysis for the vacuum}

It remains to show that if $K$ is the outcome of measuring $\ket{\psi_0}^{\otimes m}$ with $M_m$, then the trace distance between the associated estimated state $\ket{\psi_K}$ and $\ket{\psi_0}$ is small with high probability. By Lemmas~\ref{lem:distance_pure} and~\ref{lem:overlap}, the squared trace distance is
\begin{align}
	\Dtr(\ket{\psi_0}, \ket{\psi_K})^2 &= 1 -|\braket{\psi_0 | \psi_K}|^2 \\
	&= 1 - \det(I - KK^\dagger)^{1/2} \\
	&\le \norm{K}^2_F\,, \label{eq:frob}
\end{align}
where in the last step we used the elementary inequality $1 - \prod_i (1 - \lambda_i)^{1/2} \le \sum_i \lambda_i$ for $0 \le \lambda_i \le 1$.
By another application of Lemma~\ref{lem:overlap}, the density at any $K\in\mathfrak{D}_n$ is given by
\begin{equation}
	p_0(K)\,\mathrm{d}K \triangleq \frac{1}{Z_{m/2 - n - 1}}\det(I - KK^\dagger)^{m/2 - n-1}\,\mathrm{d}K\,.
\end{equation}

\begin{lemma}\label{lem:vacuum_analysis}
	If $K\in\mathfrak{D}_n$ is distributed according to probability density $p_0(K)\,\mathrm{d}K$, then provided $m \ge 2n+2$, for any $\eta > 0$ we have
	\begin{equation}
		\Pr{\norm{K}^2_F \ge \eta} \le e^{-(m/2-n-1)\eta/2} 2^{n(n+1)}\,.
	\end{equation}
    In particular, there is an absolute constant $C > 0$ such that 
    \begin{equation}
        \Pr{\norm{K}^2_F \ge \eta} \le \delta \qquad \text{if} \qquad m \ge 2n + 2 + \frac{C}{\eta}(n^2 + \log(1/\delta))\,.
    \end{equation}
\end{lemma}

\begin{proof}
    For convenience, define $\gamma\triangleq m/2 - n - 1$; by assumption $\gamma \ge 0$, so that we may apply the estimates from Section~\ref{sec:matrix_integrability}. Let $t_i$ be the squared singular values in the Takagi decomposition of $K$ from Section~\ref{sec:matrix_integrability}, so that $\norm{K}^2_F = \sum_i t_i$, and $p_0(K) = \frac{1}{Z_\gamma} \prod_i (1 - t_i)^{\gamma}$.

    On the event that $\sum_i t_i \ge \eta$, note that $\sum_i \log(1 - t_i) \le -\eta$, so
    \begin{equation}
        p_0(K)\,\mathrm{d}K = \frac{1}{Z_\gamma}\prod_i (1 - t_i)^{\gamma/2} \cdot \prod_i (1 - t_i)^{\gamma/2} \,\mathrm{d}K \le \frac{1}{Z_\gamma} e^{-\gamma \eta/2} \prod_i (1 - t_i)^{\gamma/2}\,\mathrm{d}K\,.
    \end{equation}
    Note that $\int_{\mathfrak{D}_n} \prod_i (1 - t_i)^{\gamma/2} \,\mathrm{d}K = Z_{\gamma/2}$, so integrating both sides over $K$ for which $\norm{K}^2_F \ge \eta$ and applying Lemma~\ref{lem:hua_estimate},
    \begin{equation}
        \int_{\mathfrak{D}_n} p_0(K)\cdot\mathds{1}[\norm{K}^2_F \ge \eta]\,\mathrm{d}K \le e^{-\gamma \eta/2} \frac{Z_{\gamma/2}}{Z_\gamma} \le e^{-\gamma \eta/2} 2^{n(n+1)}
    \end{equation}
    as claimed.
\end{proof}

\noindent The main theorem follows immediately:

\begin{proof}[Proof of Theorem~\ref{thm:main_energyfree}]
    By Corollary~\ref{cor:reduce_to_vacuum}, we may assume without loss of generality that $K^* = 0$. By Lemma~\ref{lem:vacuum_analysis}, provided $m$ satisfies Eq.~\eqref{eq:mbound}, with probability at least $1 - \delta$ the outcome $\ket{\psi_K}$ of measuring $\ket{\psi_0}^{\otimes m}$ with the POVM $M_m$ satisfies $\norm{K}^2_F \le \epsilon^2$. By Eq.~\eqref{eq:frob}, this implies that the estimate is $\epsilon$-close in trace distance to the true state.
\end{proof}

\subsection{Learning squeezed vacuum states with single-copy measurements}
\label{sec:simple_energyfree}

We now give a learning protocol that achieves energy-independent sample complexity for pure Gaussian states on a single mode without highly entangled measurements:

\begin{theorem}\label{thm:simple_energyfree_n1}
    There is an absolute constant $C > 0$ such that for any $0 < \epsilon \le1/3$, and $0 < \delta < 1$, there is a learning protocol that uses single-copy measurements on
    \begin{equation}
        m = \frac{C\log(1/\delta)}{\epsilon^2}\label{eq:mbound_simple}
    \end{equation}
    copies of an unknown pure single-mode Gaussian state $\rho(0,V)$, and outputs a description of a pure Gaussian state $\rho(0,\hat{V})$ such that $\Dtr(\rho(0,V),  \rho(0,\hat{V})) \le \epsilon$ with probability at least $1 - \delta$.

    By Lemma~\ref{lem:reduce_to_meanzero}, the same learning guarantee holds for pure single-mode Gaussian states of arbitrary displacement, with two-copy measurements.
\end{theorem}

\noindent We will give a single-copy measurement protocol for zero displacement pure Gaussian states. Recall from Section~\ref{sec:squeezed} that any such state takes the form $\rho(0,V_{\lambda,\theta})$, where $V_{\lambda,\theta}$ is defined in Eq.~\eqref{eq:Vtheta}. Recall that on a single mode, the Siegel disk is the complex unit disk, and under the Siegel disk parametrization we associate to $\rho(0,V_{\lambda,\theta})$ the scalar
\begin{equation}
    z = \frac{\lambda - 1}{\lambda + 1} e^{2i\theta}
\end{equation}
so that $\rho(0,V_{\lambda,\theta}) = \ketbra{\psi_z}$.
For convenience, define $r = \ln(\lambda)/2$ so that 
\begin{equation}
    |z| = \tanh(r) = \frac{\lambda-1}{\lambda+1}\,.
\end{equation}
Finally, denote the \emph{wrap-around metric} on $[0,2\pi)$ by $d_w(x,y) = \min_{j\in\mathbb{Z}} |x - y + 2\pi j|$.

We will consider the following non-Gaussian measurement:

\begin{definition}[Even canonical phase POVM]\label{def:canonical}
    Define the following operator-valued measure on the even Fock space $\mathcal{H}_{\rm even} = \overline{\mathrm{span}}(\{\ket{2k}\}_{k\in\mathbb{Z}_{\ge 0}})$. For any Borel set $\mathcal{B}\subseteq[0,2\pi)$, the \emph{even canonical phase POVM} is given by
    \begin{equation}
        \Phi(\mathcal{B}) = \sum^\infty_{k,\ell = 0} \frac{1}{2\pi} \int_{\mathcal{B}} e^{i(k-\ell)\alpha}\,\mathrm{d}\alpha \ket{2k}\bra{2\ell}\,,
    \end{equation}
    where the integral is defined in the weak sense, i.e. for $\ket{\psi}, \ket{\phi}\in\mathcal{H}_{\rm even}$,
    \begin{equation}
        \bra{\psi}\Phi(\mathcal{B})\ket{\phi}  = \frac{1}{2\pi} \int_{\mathcal{B}} \overline{\Bigl(\sum_k \psi_{2k} e^{-ik\alpha}\Bigr)} \Bigl(\sum_\ell \phi_{2\ell} e^{-i\ell\alpha}\Bigr)\mathrm{d}\alpha\,.
    \end{equation}
\end{definition}

\noindent Although $\Phi$ is only a POVM over the subspace $\mathcal{H}_{\rm even}$, all of the states we will measure with $\Phi$ live in this subspace.

\subsubsection{Measurement outcome distribution concentrates}

We now analyze the distribution over measurement outcomes resulting from the even canonical phase POVM. Towards this goal, we first record the following consequence of Lemmas~\ref{lem:diskfock} and~\ref{lem:kmode}.

\begin{proposition}[Fock coefficients of squeezed vacuum states]\label{prop:focksqueezed}
    In the Fock basis we have
    \begin{equation}
        \ket{\psi_z} \triangleq \ket{\psi_{r,\theta}} = (1 - \tanh(r)^2)^{1/4} \sum^\infty_{k=0} b_k z^k\ket{2k} \ \ \ \text{for} \ \ \ b_k \coloneqq \frac{\sqrt{\binom{2k}{k}}}{2^k}\,,
    \end{equation}
    where $z = \tanh(r)e^{2i\theta}$.
\end{proposition}

\begin{corollary}[Measurement distribution for even canonical phase POVM]\label{cor:measurement}
    Let $p_{r,\theta}$ denote the probability density over measurement outcomes. Then for any $\alpha \in [0,2\pi)$,
    \begin{equation}
        p_{r,\theta}(\alpha) = p_r(\alpha - 2\theta) \ \ \ \text{for} \ \ \ p_r(t) \coloneqq \frac{\sqrt{1 - \tanh(r)^2}}{2\pi}\,\biggl|\sum^\infty_{k=0} b_{k} \tanh(r)^{k} e^{-ikt}\biggr|^2\,,
    \end{equation}
    where $\alpha - 2\theta$ is computed modulo $2\pi$.
\end{corollary}

\noindent Note that $p_r$ is an even function, meaning that $p_{r,\theta}$ is symmetric around $2\theta$. Provided that $p_r$ has significant mass in a sufficiently small neighborhood around the origin, this means that we can estimate $\theta$ accurately by performing many measurements of copies of $\ket{\psi_z}$ with the even canonical phase POVM and outputting an appropriate ``median'' of these samples.

To show that the density $p_r$ has significant mass around the origin, we first need some elementary estimates:

\begin{lemma}\label{lem:b_k}
\[b_k\asymp(1+k)^{-1/4}\,.\]
\end{lemma}

\begin{proof}
By Stirling's formula, for $n\geq1$, $n!=\sqrt{2\pi n}(n/e)^ne^{\epsilon_n}$, where the error $\epsilon_n$ is bounded by $\frac{1}{12n+1}<\epsilon_n<\frac{1}{12n}$~\cite{robbins1955remark}.
For $k=0$ we have $b_k=1=(1+k)^{-1/4}$, and for $k\geq1$ we have
\[b_k=\frac{\sqrt{(2k)!}}{k!2^k}=\frac{\sqrt{\sqrt{2\pi2k}(2k/e)^{2k}e^{\epsilon_{2k}}}}{\sqrt{2\pi k}(k/e)^ke^{\epsilon_k}2^k}=(\pi k)^{-1/4}e^{\epsilon_{2k}/2-\epsilon_k}\,,\]
and
\[(\pi(k+1))^{-1/4}e^{1/(48k+2)-1/(12k)}<(\pi k)^{-1/4}e^{\epsilon_{2k}/2-\epsilon_k}<(\pi(k+1)/2)^{-1/4}e^{1/(48k)-1/(12k+1)}\,,\]
so $b_k\asymp(1+k)^{-1/4}$.
\end{proof}

\begin{lemma}\label{lem:infinite_sum}
For $\alpha>-1$,
\[S_\alpha(r)\coloneqq\sum_{k=0}^\infty(1+k)^\alpha\tanh(r)^k\asymp_\alpha(1-\tanh(r))^{-\alpha-1}\,.\]
\end{lemma}

\begin{proof}
Let $q\coloneqq1-\tanh(r)$, then $0<q\leq1$, and
\[S_\alpha(r)\leq\sum_{k=0}^\infty(1+k)^\alpha e^{-qk}=e^q\sum_{k'=1}^\infty k'^\alpha e^{-qk'}\,,\]
where we set $k'=k+1$.
For $x\in[k'-1,k']$, we have $x+1\in[k',k'+1]\subset[k',2k']$ (only $k'\geq1$ appears in the summation), so $k'^\alpha\lesssim_\alpha(x+1)^\alpha$.
Moreover, $e^{-qk'}\leq e^{-qx}$, so write $y=q(x+1)$,
\[S_\alpha(r)\lesssim_\alpha e^q\int_0^\infty(x+1)^\alpha e^{-qx}\,\mathrm{d}x=q^{-\alpha-1}e^{2q}\int_q^\infty y^\alpha e^{-y}\,\mathrm{d}y\leq q^{-\alpha-1}e^{2}\Gamma(\alpha+1)\asymp_\alpha(1-\tanh(r))^{-\alpha-1}\,.\]

For the reverse direction, if $\tanh(r)\leq\frac12$, then $q\geq\frac12$, $q^{-\alpha-1}\leq2^{\alpha+1}$ while $S_\alpha(r)\geq1$ (the $k=0$ term).
Hence $S_\alpha(r)\geq 2^{-\alpha-1}q^{-\alpha-1}$.
If on the other hand, $\tanh(r)>\frac12$, then $q<\frac12$.
Let $M\coloneqq\lfloor\frac1q\rfloor$.
For $0\leq k\leq M$ we have
\[\tanh(r)^k=(1-q)^k\geq(1-q)^M\geq(1-q)^{1/q}\geq\frac14\,.\]
So
\[S_\alpha(r)\geq\sum_{k=0}^M(1+k)^\alpha\tanh(r)^k\gtrsim\sum_{k=0}^M(1+k)^\alpha\gtrsim_\alpha(M+1)^{\alpha+1}\geq q^{-\alpha-1}\,.\]
\end{proof}

\noindent The two facts above can now be used to show that $p_{r,\theta}$ places large mass on angles close to $\theta$:

\begin{lemma}\label{lem:conc}
    There are absolute constants $c,C > 0$ such that for all $t$ satisfying $d_w(t,0) \le c(1 - \tanh(r))$, $p_r(t) \ge \frac{C}{1 - \tanh(r)}$.
\end{lemma}

\begin{proof}
Combining Lemma~\ref{lem:b_k} and Lemma~\ref{lem:infinite_sum}, we have
\begin{align}
p_r(0)&=\frac{\sqrt{1-\tanh(r)^2}}{2\pi}\biggl(\sum_{k=0}^\infty b_k\tanh(r)^k\biggr)^2\\
&\asymp\sqrt{1-\tanh(r)}\biggl(\sum_{k=0}^\infty(1+k)^{-1/4}\tanh(r)^k\biggr)^2\\
&\asymp\frac{1}{1-\tanh(r)}\,.
\end{align}

    Next, for $t\ge 0$ note that
    \begin{align}
        |p_r(t)-p_r(0)|&=\frac{\sqrt{1-\tanh(r)^2}}{2\pi}\biggl(\biggl(\sum_{k=0}^\infty b_k\tanh(r)^k\biggr)^2-\biggl|\sum^\infty_{k=0} b_{k} \tanh(r)^{k} e^{-ikt}\biggr|^2\biggr)\\
        &=\frac{\sqrt{1-\tanh(r)^2}}{2\pi}
        \begin{aligned}[t]
        &\biggl(\biggl(\sum_{k=0}^\infty b_k\tanh(r)^k\biggr)+\biggl|\sum^\infty_{k=0} b_{k} \tanh(r)^{k} e^{-ikt}\biggr|\biggr)\\
        &\biggl(\biggl(\sum_{k=0}^\infty b_k\tanh(r)^k\biggr)-\biggl|\sum^\infty_{k=0} b_{k} \tanh(r)^{k} e^{-ikt}\biggr|\biggr)\\
        \end{aligned}\\
        &\leq\frac{\sqrt{1-\tanh(r)^2}}{2\pi}\cdot2\biggl(\sum_{k=0}^\infty b_k\tanh(r)^k\biggr)\biggl(\sum_{k=0}^\infty b_k\tanh(r)^k|e^{-ikt}-1|\biggr)\\
        &\lesssim\sqrt{1-\tanh(r)}\cdot t\biggl(\sum_{k=0}^\infty(1+k)^{-1/4}\tanh(r)^k\biggr)\biggl(\sum_{k=0}^\infty(1+k)^{3/4}\tanh(r)^k\biggr)\\
        &\lesssim t(1-\tanh(r))^{1/2}(1-\tanh(r))^{-3/4}(1-\tanh(r))^{-7/4}\\
        &=t(1-\tanh(r))^{-2}\,,
    \end{align}
    where we have applied Lemma~\ref{lem:b_k} in the third line and Lemma~\ref{lem:infinite_sum} in the fourth.
    So provided that $d_w(t,0) \le c(1 - \tanh(r))$ for $c$ a sufficiently small constant, $|p_r(t) - p_r(0)| \le p_r(0)/2$, and the claim follows.
\end{proof}

\subsubsection{Estimating the angle using median over the torus}

Next, we record an observation about the median estimator over the real line, which is immediate from a Chernoff bound:

\begin{lemma}\label{lem:median}
    Let $0 < \epsilon \le 1$ and $0 < \gamma \le 1/2$, and let $t_0 < t_1$. Let $p$ be an arbitrary probability density over $\mathbb{R}$ with cumulative density $F(t) \triangleq \int^t_{-\infty} p(s)\,\mathrm{d}s$ satisfying $F(t_0) \le 1/2 - \gamma$ and $F(t_1) \ge 1/2 + \gamma$. Then given $N$ samples independently sampled from $p$, their median lies in $[t_0, t_1]$ with probability at least $1 - 2e^{-2N\gamma^2}$.
\end{lemma}

\noindent Unfortunately, this on its own is insufficient as $p_{r,\theta}$ is a distribution over the \emph{torus} $[0,2\pi)$, and the median will not work off the shelf in the presence of this periodicity. Instead, we need an additional step where we coarsely estimate $\theta$ in order to \emph{re-center} the distribution.

\begin{lemma}\label{lem:coarse}
    Let $x_1,\ldots,x_N$ be independent samples from $p_{r,\theta}$, and define $\hat{\mu} \triangleq \frac{1}{N}\sum^N_{j=1} e^{ix_j}$. Let $\wh{\phi}\in[0,2\pi)$ denote $\mathrm{arg}(\hat{\mu})$. There are constants $A, c, c'$ such that for every $0 < \eta \le 1$ and $0 < \delta < 1$, if $N \ge c\log(2/\delta)/\eta^2$ and $\tanh(r) \ge A\eta$, then with probability at least $1 - \delta$, $d_w(\hat{\mu}, 2\theta) \le c'\eta/\tanh(r)\le \pi/8$.
\end{lemma}

\begin{proof}
    We first compute the expectation of $\hat{\mu}$. Observe that for $x\sim p_{r,\theta}$, by Corollary~\ref{cor:measurement},
    \begin{equation}
        \EE[e^{ix}] = e^{2i\theta}\cdot \frac{\sqrt{1 - \tanh^2 r}}{2\pi} \sum^\infty_{k=0} b_k b_{k+1} \tanh(r)^{2k+1}\,,
    \end{equation}
    so its phase reveals the angle $\theta$ up to period. Furthermore, its modulus is bounded: $0 \le b_{k+1} = \sqrt{\frac{2k+1}{2k+2}}\,b_k$, so $1/\sqrt{2} \le \frac{b_{k+1}}{b_k} \le 1$, so $|\EE[e^{ix}]| \in \tanh(r)\cdot [1/\sqrt{2}, 1]$.

    By our choice of $N$ and a Chernoff bound, $|\hat{\mu} - \EE[e^{ix}]| c''\eta$ for some constant $c'' > 0$ with probability at least $1 - \delta$. If $\tanh(r) \ge A\eta$ for sufficiently large constant $A > 0$, then $c''\eta \le |\EE[e^{ix}]|/10$, ensuring that $d_w(\wh{\phi}, 2\theta) \lesssim \eta/\tanh(r)$, and provided $A$ is sufficiently large, this is at most $\pi/8$.
\end{proof}

\noindent Next, we show that $p_{r,\theta}$ places little mass on angles far from $\theta$.

\begin{lemma}\label{lem:faraway}
    $p_r(t) \lesssim \sqrt{1 - \tanh(r)}$ if $d_w(t,0) \ge \pi/2$.
\end{lemma}

\begin{proof}
    We can write
    \begin{equation}
        \sum^\infty_{k=0} b_k\tanh(r)^k e^{-ikt} = \sum^\infty_{k=0} (b_k\tanh(r)^k - b_{k+1}\tanh(r)^{k+1}) \sum^k_{\ell = 0} e^{-i\ell t}\,.
    \end{equation}
    The inner summation is upper bounded in magnitude by $\frac{2}{|1 - e^{-it}|} \le \sqrt{2}$ for all $k$ and all $t$ satisfying $d_w(t,0) \ge \pi/2$. Furthermore, 
    \begin{equation}
        b_k\tanh(r)^k - b_{k+1}\tanh(r)^{k+1} \le b_k\tanh(r)^k \le 1\,.
    \end{equation}
    So $p_r(t) \lesssim \sqrt{1 - \tanh(r)^2} \lesssim \sqrt{1 - \tanh(r)}$ as claimed.
\end{proof}

\noindent We are now ready to upgrade the estimator in Lemma~\ref{lem:median}, which operates over the real line, to an estimator over the torus.

\begin{lemma}[Estimating angle]\label{lem:circular_median}
    There are absolute constants $A, c, c', \overline{\eta} > 0$ such that the following holds for all $0 < \eta \le \overline{\eta}$. Consider the following procedure:
    \begin{enumerate}
        \item Perform the canonical phase POVM on $N\ge c\log(2/\delta)/\eta^2$ copies of unknown state $\ket{\psi_z}$
        \item For the first $N/2$ measurement outcomes $x_1,\ldots,x_{N/2}$, form the estimator $\hat{\phi} = \mathrm{arg}\Bigl(\frac{1}{N}\sum^N_{j=1} e^{ix_j}\Bigr)$.
        \item For each $x_j$ from the next $N/2$ outcomes, form $y_j = x_j - \hat{\phi} \mod{[-\pi,\pi)}$. 
        \item Let $\hat{m}$ denote the median of the $y_j$'s and define $\tilde{\mu} = \hat{m} + \hat{\phi}\mod{2\pi}$.
    \end{enumerate}
    Then with probability at least $1 - \delta$, we have that either
    \begin{equation}
        d_w(2\theta, \tilde{\mu}) \le c'\eta\frac{1 - \tanh(r)^2}{\tanh(r)}
    \end{equation}
    or $\tanh(r) \le A\eta$. In either case,
    \begin{equation}
        \Dtr(\ket{\psi_{r,\theta}}, \ket{\psi_{r,\tilde{\mu}/2}}) \le \eta\,.
    \end{equation}
\end{lemma}

\begin{proof}
    Let $\Delta \in [-\pi,\pi)$ denote $2\theta - \hat{\phi}\mod{[-\pi,\pi)}$ so that $|\Delta| \le \pi/8$ and $|\Delta| \lesssim \eta/\tanh(r)$ with probability at least $1 - \delta/2$ by Lemma~\ref{lem:coarse}. Because the $y_j$'s lie in $[\pi,-\pi)$, their law has probability density given by
    \begin{equation}
        q_\Delta(y) = p_r(y - \Delta \mod 2\pi)\,.
    \end{equation}
    Letting $F_\Delta$ denote the cumulative density of $q_\Delta$, we see from Lemma~\ref{lem:faraway} that
    \begin{equation}
        F_\Delta(\Delta) - 1/2 = \mathrm{sgn}(\Delta) \cdot \int^\pi_{\pi - |\Delta|} p_r(t)\,\mathrm{d}t \lesssim \frac{\eta\sqrt{1 - \tanh(r)}}{\tanh(r)}\,, \label{lem:use_faraway}
    \end{equation}
    where in the inequality we used that $\pi - |\Delta| \ge \pi/2$.

    Next, take $u = C\eta\cdot\frac{1 - \tanh(r)^2}{\tanh(r)}$ for sufficiently large absolute constant $C > 0$. If $A$ in the bound $A\eta \le \tanh(r)$ is chosen sufficiently large relative to $C$, then $u \lesssim 1 - \tanh(r)$ and $u\le \pi/8$. Thus $\Delta \pm u \in (-\pi, \pi)$, so by Lemma~\ref{lem:conc} we have that
    \begin{equation}
        |F_\Delta(\Delta + u) - F_\Delta(\Delta)|, |F_\Delta(\Delta - u) - F_\Delta(\Delta)| \gtrsim \frac{u}{1 - \tanh(r)}\,. \label{lem:use_conc}
    \end{equation}
    Combining Eqs.~\eqref{lem:use_conc} and~\eqref{lem:use_faraway}, and noting that the lower bound in the former dominates the upper bound in the latter for large enough $C$, we conclude that
    \begin{equation}
        F_\Delta(\Delta - u) \le \frac{1}{2} - \Theta(\eta) \qquad \text{and} \qquad F_\Delta(\Delta + u) \ge \frac{1}{2} + \Theta(\eta)\,.
    \end{equation}
    Lemma~\ref{lem:median} then implies that $|\hat{m} - \Delta| \lesssim u$ with probability at least $1 - \delta_2$. We conclude that $d_w(2\theta, \tilde{\mu}) \lesssim u$, establishing the first part of the claim.

    For the remaining part, we invoke Lemma~\ref{lem:overlap} to get
    \begin{equation}
        |\braket{\psi_{r,\theta} | \psi_{r,\tilde{\mu}/2}}|^2 = (1 + 2(1 - \cos(2d_w(2\theta,\tilde{\mu}))\sinh(r)^2\cosh(r)^2)^{-1/2}
    \end{equation}
    Under the assumption that $\tanh(r) \ge A\eta$, 
    we have $d_w(2\theta,\tilde{\mu}) \lesssim u$ from the above, and by our choice of $u$, the above is at least $1 - O(\eta^2)$ and we are done after absorbing constants.

    Under the assumption that $\tanh(r) \le A\eta$, the above is maximized when $\cos(d_w(2\theta,\tilde{\mu})) = 0$, in which case we have
    \begin{align}
        |\braket{\psi_{r,\theta} | \psi_{r,\tilde{\mu}/2}}|^2 &= (1 + 2\sinh(r)^2\cosh(r)^2)^{-1/2} \ge 1 - O(\eta^2)\,. \qedhere
    \end{align}
\end{proof}

\subsubsection{Estimating the squeezing parameter}

It remains to estimate the squeezing parameter. For this, there are two options: (1) unrotate the state by the estimated angle and perform a homodyne measurement, or (2) perform photon number counting on the original state. Both are relatively straightforward, and we will pursue the latter:

\begin{lemma}[Estimating squeezing]\label{lem:photoncount}
    Let $0 < \epsilon \le 1/3$. If one performs photon number counting on $O(\log(1/\delta)/\epsilon^2)$ copies of $\ket{\psi_z}$ and uses a median-of-means estimator on the outcomes, then the resulting estimate $\hat{m}$ satisfies $|\hat{m} - \sinh^2(r)| \lesssim \epsilon\cosh(r)\sinh(r)$ with probability at least $1 - \delta$, in which case
    \begin{equation}
        |\arcsinh(\sqrt{\hat{m}}) - r| \lesssim \epsilon\,.
    \end{equation}
\end{lemma}

\begin{proof}
    By Proposition~\ref{prop:focksqueezed}, photon number counting results in a distribution over photon number $N$ which places mass
    \begin{equation}
        \frac{\binom{2k}{k}}{2^{2k}}\tanh^{2k}(r)\sech(r)
    \end{equation}
    on outcome $N = 2k$ and mass $0$ on outcome $N = 2k+1$ for all $k\in\mathbb{Z}_{\ge 0}$. Using the elementary equality $\sum^\infty_{k=0} \binom{2k}{k}(c/2)^{2k}\cdot (2k) = \frac{c^2}{(1-c^2)^{3/2}}$ for $0\le c < 1$ and taking $c = \tanh(r)$, we conclude that the mean photon number 
    \begin{equation}
        \E{N} =\sinh^2(r)\,.
    \end{equation}
    Using the elementary equality $\sum^\infty_{k=0} \binom{2k}{k} (c/2)^{2k}\cdot (2k)^2 = \frac{c^2 (2 + c^2)}{(1 - c^2)^{5/2}}$, we conclude that
    \begin{equation}
        \Var{N} = 2\cosh^2(r)\sinh^2(r)\,.
    \end{equation}
    By standard bounds for the median-of-means estimator, e.g., \cite[Proposition 12]{lerasle2019lecture}, with $O(\log(1/\delta)/\epsilon^2)$ samples from this distribution, the median-of-means estimator results in $\hat{m}$ for which $|\hat{m} - \sinh^2(r)| \lesssim \epsilon\sqrt{\Var{N}/2} = \epsilon\cosh(r)\sinh(r)$ with probability at least $1 - \delta$, thus proving the first part of the claim.

    For the last part of the claim, let $\xi \triangleq c\epsilon\cosh(r)\sinh(r)$ for sufficiently small constant $c > 0$, and let $h \triangleq \arcsinh(\sqrt{\max(0,\cdot)})$. First, suppose that $\sinh^2(r) \le 2\xi$, i.e., $\tanh(r) \le 2\epsilon$. Note that this implies that $r \le 3\epsilon$ for $\epsilon \le 1/3$, and that
    \begin{equation}
        \sinh^2(r) + \xi = \cosh^2(r) \cdot (\tanh^2(r) + \epsilon\tanh(r)) \lesssim \epsilon^2\,.
    \end{equation}
    Then for any $\hat{m}$ for which $|\hat{m} - \sinh^2(r)| \le \xi$, we have
    \begin{equation}
        h(\hat{m}) \in [0, h(\sinh^2(r) + \xi)] \subseteq [0,O(\epsilon)]
    \end{equation}
    and thus that $|h(\hat{m}) - r| \lesssim \epsilon$ as claimed.

    On the other hand, suppose that $\sinh^2(r) > 2\xi$.
    Note that $h'(x) = \frac{1}{2\sqrt{x(1+x)}}$. We thus have
    \begin{align}
        |h(\hat{m}) - r| &\le \xi \sup_{y\in [\sinh^2(r) - \xi, \sinh^2(r) + \xi]} h'(y) \le \frac{\xi}{\sinh(r)\sqrt{2 + \sinh^2(r)}} = \frac{\epsilon\cosh(r)}{\sqrt{2 + \sinh^2(r)}} \le \epsilon\,. \qedhere
    \end{align}
\end{proof}

\subsubsection{Putting everything together}

We can now conclude the proof of Theorem~\ref{thm:simple_energyfree_n1}:

\begin{proof}[Proof of Theorem~\ref{thm:simple_energyfree_n1}]
    For convenience, in this proof we use the notation $\ket{\psi_{r,\theta}}$ in place of $\ket{\psi_z}$ to make clear the modulus and phase.
    
    Apply Lemma~\ref{lem:circular_median} with $\eta$ a sufficiently small constant multiple of $\epsilon$ to get an estimate $\hat{\theta} = \tilde{\mu}/2$ of $\theta$ with  probability $1 - \delta/2$, using single-copy measurements of $O(\log(1/\delta)/\epsilon^2)$ copies of $\ket{\psi_z}$. By Lemma~\ref{lem:photoncount}, with photon number counting applied to $O(\log(1/\delta)/\epsilon^2)$ copies, we obtain an estimate $\hat{r}$ for which $|r - \hat{r}| \le \epsilon$ with probability $1 - \delta/2$. Letting $\Delta = \hat{\theta} - \theta$,
    \begin{align}
        \Dtr(\ket{\psi_{r,\theta}}, \ket{\psi_{\hat{r},\hat{\theta}}}) &\le \Dtr(\ket{\psi_{r,\theta}}, \ket{\psi_{r,\hat{\theta}}}) + \Dtr(\ket{\psi_{r,\hat{\theta}}}, \ket{\psi_{\hat{r},\hat{\theta}}}) \\
        &= \sqrt{1 - |\braket{\psi_{r,\Delta} | \psi_{r,0}}|^2} + \sqrt{1 - |\braket{\psi_{r,0} | \psi_{\hat{r},0}}|^2}\,. \label{eq:triangle_ineq}
    \end{align}
    The second part of Lemma~\ref{lem:circular_median} shows that the first term on the right-hand side is at most $\epsilon/2$.
    On the other hand,
    \begin{equation}
        |\braket{\psi_{r,0}|\psi_{\hat{r},0}}|^2 = \Bigl(\frac{1 + \cosh(2(r - \hat{r}))}{2}\Bigr)^{-1/2}\,,
    \end{equation}
    so provided that $|r - \hat{r}| \le c'\epsilon$ for sufficiently small constant $c'$, this is at least $1 - \epsilon^2/4$, from which the main claim follows by Eq.~\eqref{eq:triangle_ineq}.
\end{proof}

\bibliographystyle{alpha}
\bibliography{refs}

\appendix

\input{metaplectic}

\end{document}

%% file: metaplectic.tex
\section{Metaplectic representation}\label{app:metaplectic}

In this section we provide self-contained proofs of the lemmas used in Section~\ref{sec:covariant}. These follow from classical facts about the \emph{metaplectic representation}; we refer the reader to~\cite{folland2016harmonic} for a comprehensive treatment.

Given a symplectic matrix $S = \begin{pmatrix}
    A & B \\
    C & D
\end{pmatrix} \in \Sp(2n)$ and $K\in\mathfrak{D}_n$, the image of $K$ under the action of $S$ is given by
\begin{equation}
    S\cdot K = (PK + Q)(Q^*K + P^*)^{-1}\,, \label{eq:mobius}
\end{equation}
where $P$ and $Q$ are given by the \emph{Bogoliubov representation} of $S$:
\begin{equation}
    \begin{pmatrix}
        P & Q \\
        Q^* & P^*
    \end{pmatrix} = TST^{-1} \qquad\text{for} \qquad T = \frac{1}{\sqrt{2}}\begin{pmatrix}
        I & -iI \\
        I & iI
    \end{pmatrix}\,. \label{eq:bogoliuobov}
\end{equation}

\begin{lemma}\label{lem:bog_symplectic}
    The matrix \begin{equation}U\coloneqq\begin{pmatrix}
        P & Q \\
        Q^* & P^*
    \end{pmatrix}
    \end{equation} is (complex) symplectic.
\end{lemma}

\begin{proof}
    Note that $T$ in Eq.~\eqref{eq:bogoliuobov} satisfies $T^\top \Omega_n T = i\Omega_n$. Thus,
    \begin{align}
        U^\top \Omega_n U &= (TST^{-1})^\top \Omega_n (TST^{-1})\\
        &= T^{-\top} S^\top (i\Omega_n) S T^{-1} \\
        &= iT^{-\top} \Omega_n T^{-1} = i(-i\Omega_n) = \Omega_n\,. \qedhere
    \end{align}
\end{proof}

\begin{lemma}\label{lem:preserve}
    The action of $S$ preserves $\mathfrak{D}_n$, that is, $S\cdot K\in\mathfrak{D}_n$ for all $K\in\mathfrak{D}_n$.
\end{lemma}

\begin{proof}
    For convenience, write
    \begin{equation}
        N = PK + Q \qquad \text{and} \qquad M = Q^*K + P^*\,.
    \end{equation}
    Let $K' = S\cdot K$. We first prove that $K' = K'^\top$, i.e. that $M^\top N = N^\top M$. We have
    \begin{equation}
    M^\top N - N^\top M = K(\underbrace{Q^\dagger P - P^\top Q^*}_0)K + K(\underbrace{Q^\dagger Q - P^\top P^*}_{-I}) + (\underbrace{P^\dagger P - Q^\top Q^*}_I)K + (\underbrace{P^\dagger Q - Q^\top P^*}_0)\,,
    \end{equation}
    where we used Lemma~\ref{lem:bog_symplectic}. So $M^\top N = N^\top M$.

    Next, we verify $I - K'K'^\dagger \succ 0$. Symmetry of $K'$ implies that $M^\top K' = N^\top$, so
    \begin{equation}
        M^\top (I - K'K'^\dagger) M^* = M^\top M^* - M^\top K'K'^\dagger M^* = M^\top M^* - N^\top N^*\,.
    \end{equation}
    We can expand
    \begin{equation}
        M^\top M^* - N^\top N^* = K(\underbrace{Q^\dagger Q - P^\top P^*}_{-I})K^\dagger + K(\underbrace{Q^\dagger P - P^\top Q^*}_0) + (\underbrace{P^\dagger Q - Q^\top P^*}_0)K^\dagger + (\underbrace{P^\dagger P - Q^\top Q^*}_I)\,,
    \end{equation}
    where we again used Lemma~\ref{lem:bog_symplectic}.
    So $I - KK^\dagger = M^\top(I - K'K'^\dagger)M^*$, implying that $I - K'K'^\dagger \succ 0$.
\end{proof}

\transitivity*

\begin{proof}
    Apply the inverse Cayley transform to $K$ to obtain
    \begin{equation}
        L = i(I + K)(I - K)^{-1}\,.
    \end{equation}
    Note that because $I + K$ and $I - K$ commute, $L = L^\top$. Write $L = X + iY$ for real symmetric matrices $X = \mathrm{Re}(L)$ and $Y = \mathrm{Im}(L)$. Observe that
    \begin{align}
        Y = \frac{1}{2i}(L - L^*) &= \frac{1}{2i}(i(I+K)(I-K)^{-1} +i(I-K^*)^{-1}(I+K^*)) \\
        &= \frac{1}{2}(I-K^*)^{-1}\bigl[(I + K^*)(I-K) + (I-K^*)(I+K)\bigr](I-K)^{-1} \\
        &= (I - K^*)^{-1}(I - K^*K)(I-K)^{-1}\succ 0\,,
    \end{align}
    so $Y$ is invertible. Define
    \[S=
    \begin{pmatrix}
        Y^{1/2}&XY^{-1/2}\\
        0&Y^{-1/2}
    \end{pmatrix}\,,
    \]
    then $S\Omega S^\top=\Omega$, so $S\in\Sp(2n)$.
    We have
    \[TST^{-1}=\frac12
    \begin{pmatrix}
        Y^{1/2}+Y^{-1/2}+iXY^{-1/2}&Y^{1/2}-Y^{-1/2}-iXY^{-1/2}\\
        Y^{1/2}-Y^{-1/2}+iXY^{-1/2}&Y^{1/2}+Y^{-1/2}-iXY^{-1/2}
    \end{pmatrix}\,,
    \]
    so $S$ has Bogoliubov representation
    \begin{equation}
        P = \frac{1}{2}(Y + I + iX)Y^{-1/2} \qquad \text{and} \qquad Q = \frac{1}{2}(Y-I-iX)Y^{-1/2}\,.
    \end{equation}
    Finally, observe that
    \begin{equation}
        S\cdot 0 = QP^{-*} = (Y - I - iX)(Y + I - iX)^{-1} = (-i\cdot (L-iI))\cdot (-i\cdot (L + iI))^{-1} = (L - iI)(L + iI)^{-1} = K\,,
    \end{equation}
    where the last step undoes the inverse Cayley transform.
\end{proof}

\noindent The following shows that for any symplectic matrix there is an associated centered Gaussian unitary; this is called the \emph{metaplectic representation}.

\metaplectic*

\begin{proof}
    Given a symplectic matrix $S$ with Bogoliubov representation $(P,Q)$, associate a Gaussian unitary $U_S$ which maps the annihilation and creation operators to
    \begin{equation}
        U^\dagger_S \hat{\mathbm{a}} U_S = P\hat{\mathbm{a}} + Q\hat{\mathbm{a}}^\dagger \qquad \text{and} \qquad U^\dagger_S \hat{\mathbm{a}}^\dagger U_S = Q^*\hat{\mathbm{a}} + P^*\hat{\mathbm{a}}^\dagger\,.
    \end{equation}
    Here and below we slightly abuse the notation and let $U_S^\dagger$ act on each component of $\hat{\mathbm{a}}$ instead of recombining the components.
    By Lemma~\ref{lem:bog_symplectic} and Fact~\ref{fact:symplectic}, $PP^\dagger - QQ^\dagger = I$ and $PQ^\top = QP^\top$. Now consider $K' = S\cdot K$, defined in Eq.~\eqref{eq:mobius}.
    
    We will use the characterization of $\ket{\psi_K}$ in Lemma~\ref{lem:null_characterization} to conclude the proof. It thus suffices to prove that $(\hat{\mathbm{a}} - K'\hat{\mathbm{a}}^\dagger)U_S\ket{\psi_{K}} = 0$. We claim that
    \begin{equation}
        U^\dagger_S(\hat{\mathbm{a}} - K'\hat{\mathbm{a}}^\dagger)U_S = (P-K'Q^*)(\hat{\mathbm{a}} - K\hat{\mathbm{a}}^\dagger)\,,
    \end{equation}
    which would conclude the proof. We have
    \begin{equation}
        U^\dagger_S(\hat{\mathbm{a}} - K'\hat{\mathbm{a}}^\dagger)U_S = (P - K'Q^*)\hat{\mathbm{a}} + (Q - K'P^*)\hat{\mathbm{a}}^\dagger\,.
    \end{equation}
    Because $K' = (PK + Q)(Q^*K + P^*)^{-1}$, rearranging we have that $Q - K'P^* = -(P-K'Q^*)K$, so substituting into the above we find that
    \begin{equation}
        U^\dagger_S(\hat{\mathbm{a}} - K'\hat{\mathbm{a}}^\dagger)U_S = (P - K'Q^*)(\hat{\ba} - K\hat{\ba}^\dagger)
    \end{equation}
    as desired.
\end{proof}

\noindent In the proof above, we used the following standard property of symplectic matrices:

\begin{fact}\label{fact:symplectic}
    If a matrix $S = \begin{pmatrix}
        A & B\\
        C & D
    \end{pmatrix}$ is symplectic, then $AD^\top - BC^\top = I$ and $AB^\top$ is symmetric.
\end{fact}

\noindent Finally, we prove that the measure $\mathrm{d}\mu$ over $\mathfrak{D}_n$ defined in Eq.~\eqref{eq:mainmeasure} is invariant under the symplectic action.

\invariance*

\begin{proof}
Fix $S\in\mathrm{Sp}(2n)$, and write its Bogoliubov representation as in Eq.~\eqref{eq:bogoliuobov}. For $K\in\mathfrak{D}_n$, set
\begin{equation}
    N = PK + Q \qquad \text{and} \qquad M = Q^*K + P^*\,,
\end{equation}
so that $K'=S\cdot K=NM^{-1}$. Recall from the proof of Lemma~\ref{lem:preserve} that $I - KK^\dagger = M^\top(I - K'K'^\dagger)M^*$.

Taking determinants, and using that $KK^\dagger$ and $K^\dagger K$ have the same characteristic polynomial, yields
\begin{equation}
    \det(I-K'K'^\dagger)
    =
    |\det M|^{-2}\det(I-KK^\dagger).
    \label{eq:det-transform}
\end{equation}

It remains to compute the Jacobian of the map $K\mapsto K'$ on the complex vector space of symmetric matrices. For any symmetric perturbation $H$, let $K(t)=K+tH$, and define
\[
    N(t)=PK(t)+Q \qquad \text{and} \qquad M(t)=Q^*K(t)+P^*.
\]
Then $\dot{N}(0)=PH$ and $\dot{M}(0)=Q^*H$. Using
\[
    \frac{\mathrm{d}}{\mathrm{d}t}M(t)^{-1}
    =
    -M(t)^{-1}\dot{M}(t)M(t)^{-1},
\]
we get
\begin{align}
    D_K(S\cdot)[H]
    &= PHM^{-1}-NM^{-1}Q^*HM^{-1} \nonumber\\
    &= (P-K'Q^*)HM^{-1}. \label{eq:raw-differential}
\end{align}
Because the previous lemma shows $K'$ is symmetric and $K'M=N$, we have $M^\top K'=N^\top$. Hence
\begin{align}
    M^\top(P-K'Q^*)
    &=M^\top P-N^\top Q^* \nonumber\\
    &=(Q^*K+P^*)^\top P-(PK+Q)^\top Q^* \nonumber\\
    &=K^\top(Q^\dagger P-P^\top Q^*)+(P^\dagger P-Q^\top Q^*) \nonumber\\
    &=I.
\end{align}
Thus $P-K'Q^*=M^{-\top}$, and
\begin{equation}
    D_K(S\cdot)[H] = M^{-\top}HM^{-1}.
    \label{eq:differential}
\end{equation}

For any $R\in GL_n(\CC)$, consider such a congruence map
\[
    T_R:\operatorname{Sym}_n(\CC)\to \operatorname{Sym}_n(\CC),
    \qquad
    T_R(H)=R^\top HR.
\]
This is a complex-linear map on the complex vector space with coordinates $(H_{ij})$, $i\le j$. We claim that
\[
    \det_{\CC}(T_R)=\det(R)^{n+1}.
\]
First suppose $R=\diag(r_1,\ldots,r_n)$. Then the coordinate $H_{ij}$, $i\le j$, is multiplied by $r_ir_j$, and hence
\[
    \det_{\CC}(T_R)=\prod_{1\le i\le j\le n}r_ir_j.
\]
For each fixed $i$, the factor $r_i$ appears twice from the diagonal coordinate $H_{ii}$ and once from each of the $n-1$ off-diagonal coordinates involving $i$. Thus its total exponent is $n+1$, proving the claim for diagonal $R$. The general case follows because both $\det_{\CC}(T_R)$ and $\det(R)^{n+1}$ are polynomial functions of the entries of $R$ which are additionally invariant under similarity transformations of $R$, and diagonalizable matrices are dense in $GL_n(\CC)$.

Finally, the measure $\mathrm{d}K$ is real Lebesgue measure in the real and imaginary parts of the independent complex coordinates. Therefore, for a complex-linear map $T$, the Jacobian of its realification is $|\det_{\CC}T|^2$. Applying this to $T_{M^{-1}}$, we conclude that
\[
    \mathrm{d}K'=|\det M|^{-2(n+1)}\mathrm{d}K.
\]
Combining this with Eq.~\eqref{eq:det-transform},
\begin{align}
    \frac{\mathrm{d}K'}{\det(I-K'K'^\dagger)^{n+1}}
    &=
    \frac{|\det M|^{-2(n+1)}\mathrm{d}K}
    {\left(|\det M|^{-2}\det(I-KK^\dagger)\right)^{n+1}} \nonumber\\
    &=
    \frac{\mathrm{d}K}{\det(I-KK^\dagger)^{n+1}}\,,
\end{align}
so $\mathrm{d}\mu$ is indeed invariant under the action of $\Sp(2n)$.
\end{proof}